\documentclass[reprint, amsmath,amssymb,aps]{revtex4-2}

\usepackage{graphicx,caption}
\graphicspath{{Images/},{Appendix_Images/}}
\usepackage{amsthm}
\newtheorem{definition}{Definition}
\newtheorem{lemma}{Lemma}
\newtheorem{corollary}{Corollary}
\newtheorem{theorem}{Theorem}

\usepackage[ruled]{algorithm2e}
\usepackage{algorithmic}

\usepackage[normalem]{ulem}

\begin{document}

\title{Limited-Trust in Social Network Games}

\author{Timothy Murray}
\email{tsmurray.professional@gmail.com}
\affiliation{Department of Industrial and Enterprise Systems Engineering \\University of Illinois Urbana-Champaign}
\author{Jugal Garg}
\email{jugal@illinois.edu}
\affiliation{Department of Industrial and Enterprise Systems Engineering \\University of Illinois Urbana-Champaign}
\author{Rakesh Nagi}
%\email{nagi@illinois.edu}
\email{rakesh\_nagi@sutd.edu.sg; nagi@illinois.edu}
\affiliation{Engineering Systems and Design, Singapore University of Tehcnology and Design \\
ISE, University of Illinois Urbana-Champaign}
% \address{Department of Industrial and Enterprise Systems Engineering \\University of Illinois at Urbana-Champaign}

% \iffalse

% \begin{keyword}{Game Theory, Multi-Agent Systems, Networks}\end{keyword}

\iffalse
%\setcounter{page}{0}
\thispagestyle{empty}
\begin{center}
\huge{Limited-Trust in Social Network Games}
\vspace{16pt}

\Large{Timothy Murray, Jugal Garg, Rakesh Nagi}\\
\vspace{8pt}
\large{tsmurray.professional@gmail.com,\{jugal,nagi\}@illinois.edu}\\
\vspace{8pt}
\large{\textit{Department of Industrial and Enterprise Systems Engineering \\University of Illinois at Urbana-Champaign}}\\
\vspace{8pt}
\normalsize{Transportation Building Room 216B \\104 S. Matthews Ave, Urbana, IL USA, 61801}\\
\vspace{16pt}
\large{Corresponding Author: Timothy Murray}
\end{center}
\newpage
\setcounter{page}{1}
\fi

\iftrue

\begin{abstract}
 We consider agents in a social network competing to be selected as partners in collaborative, mutually beneficial activities. We study this through a model in which an agent {$i$} can initiate a limited number {$k_i>0$} of games and selects partners from its one-hop neighborhood. Each agent can accept as many games offered by its neighbors. Each game signifies a productive joint activity, and players attempt to maximize their individual utilities.
Unsurprisingly, more trustworthy agents, as measured by the game-theoretic concept of limited-trust, are more desirable as partners. Agents learn about their neighbors' trustworthiness through interactions, and their behaviors evolve in response. Empirical trials performed on realistic social networks show that when given the option, many agents become highly trustworthy; most or all become highly trustworthy when knowledge of their neighbors' trustworthiness is based on past interactions rather than known \textit{a priori}. This trustworthiness is not the result of altruism, instead, agents are intrinsically motivated to become trustworthy partners by competition. Two insights are presented: first, trustworthy behavior drives an increase in the utility of all agents, where maintaining a relatively modest level of trustworthiness may easily improve net utility by as much as 14.5\%. If only one agent exhibits modest trust among self-centered ones, it can increase its personal utility by up to 25\% in certain cases! Second, and counter-intuitively, when partnership opportunities are abundant, agents become less trustworthy.
\end{abstract}

\maketitle

\section{Introduction}
Choose your friends wisely. 
It's good advice, and it also applies to the problem of selecting partners to work with. 
Effective partnerships are based on trust: 
Suppose you must collaborate on a project and need to decide between two potential collaborators. 
Both candidates possess the same basic level of expertise, leading you to expect that the project will be a success with either of them. 
However, one candidate has a reputation for taking all of the credit in collaborations and using them to advance their own interests over those of their partners. 
Naturally, you would prefer to avoid the candidate with a poor reputation and instead collaborate with your other colleague with no such negative history.

In this paper, we consider the problem of partner selection within the larger context of a social network.
Agents in the network must secure partnership opportunities with mutual benefits. 
These agents are in a state of coopetition with each other rather than pure cooperation or competition, forcing them to rely on the concepts of trust and reputation to initiate strategic partnerships.
We use the following definitions of trust and reputation from \cite{mui2002}:
\begin{definition}[Reputation]
Perception that an agent creates through past actions about its intentions and norms.
\end{definition}
\begin{definition}[Trust]
A subjective expectation an agent has about another’s future behavior based on the history of their encounters.
\end{definition}
The two definitions go hand-in-hand, and in practice, we will use terms such as ``trustworthy'' both academically and in life to refer to an agent who maintains a highly positive reputation.
The interactions between agents occur as limited-trust leader-follower games, where limited-trust (and associated equilibria) is a concept recently developed by \cite{murray2021}. 
Loosely speaking, limited-trust assumes that an agent will help their fellow agent, provided that the cost is not too high and the net utility of all players improves. 
The trustworthiness of an agent $i$ is determined by a metric $\delta_i \geq 0$ so that if $\delta_i=1$, agent $i$ is willing to lose up to one unit of utility to improve the net utility of all agents. 
More precisely, for any $0<x<y$, agent $i$ is willing to incur a cost of $x$ to increase the net utility of other agents by $y$, provided that $x\leq \delta_i$. 
Limited-trust is naturally applicable to social network-based interactions, as it provides a mechanism for agents to increase their long-run utilities through more complex behaviors while still being fundamentally self-interested rather than altruistic.
We find that these self-interested agents maintain highly positive reputations and trust-confirming behavior, as deviating causes them to lose partnership opportunities.   

The main contribution of this paper is developing a system for modeling interactions between individuals in a social network. 
This system is thoroughly analyzed, with algorithms developed for individuals to learn their neighbors' reputations and alter their trust levels accordingly. 
Depending on how ties are broken when selecting partners, this can lead either to the majority of agents maintaining positive reputations or cycling behavior between low and highly positive reputation levels and trustworthiness. 
However, this cycling can be curtailed by subsidizing a small number of seed vertices to act as leaders and maintain a positive reputation. 
In both cases, we find these behavioral changes increase average reputation and trustworthiness, leading to a substantial increase in the total utility in the network. 
Empirically, such a system is substantiated by numerous studies in evolutionary biology and psychology, whose results mirror the behaviors we observe in the model. 
These results indicate that individuals behaving in a trustworthy manner are typically the most successful. 
What they lose in individual interactions they more than makeup for by increasing their opportunities. 
In the long term, they also spur other agents to adopt trustworthy behaviors, resulting in more utility from the same interactions. 
Further, because the trust level captures all temporal knowledge in a single, easily-updated value, otherwise myopic agents arrive at trustworthy behavior naturally without using external history-based mechanisms such as grim trigger or tit-for-tat strategies. 
As such, we feel that limited-trust concept is more natural and intuitive for the social scenario we consider.

Our own empirical studies also reveal two counter-intuitive insights. 
First, while one might expect selfishness when opportunities are limited and individuals try to make the most of them, we find that individuals are at their most trustworthy when opportunities are limited and do their best to maintain positive reputations to capture what few are available. 
Second, it is similarly natural to expect individuals to be more trustworthy when there are more opportunities, as taking advantage of any one opportunity is not worth the resulting reputational damage. 
Instead, we observe that the glut of opportunities outweighs the reputational consequences for selfish individuals, as their behavior will not limit their future opportunities (subject to network structure).

We consider a toy example given by the network in Figure \ref{fig_example0}. Suppose that each agent may lead $k=2$ leader-follower games and can select any 2 neighbors as followers for these games, which are randomly drawn from a known distribution. Let the games be over $2\times 2$ payoff matrices with each entry drawn independently and identically distributed (iid) at random from an exponential distribution where $\lambda=2$. If all agents play selfishly, with $\delta_i=0$ for every agent $i$, then each agent will select partners uniformly at random. That means that agent 7 can expect to be a leader in 2 games and be chosen as a follower for $\frac{17}{6}$ games.
% , as its expected games from neighbors 1, 3, 4, and 6 are $\frac{2}{3}, \frac{2}{3}, \frac{1}{2}, \text{ and }1$, respectively. 
Suppose that agent 7 now behaves in a slightly trustworthy manner, with $\delta_7=0.01$ whenever it interacts with any other agent, while all of its neighbors remain selfish: it now attracts game invitations from all of its neighbors and participates in 4 games as a follower. It will expect to achieve slightly less utility per game but will play in an additional $\frac{7}{6}$ games on average. Table \ref{tab_uniform_deltas} compares these two settings: when all other agents are selfish, agent 7 can improve its average utility by over 25\% by being only slightly trustworthy! The table also reveals a second insight: when all agents are equally trustworthy ($\delta_i=2$ for $i\in \{1,2,...,7\}$), all stand to gain a significant amount of utility (approximately $14.5\%$ per agent) driving a significant increase in the net utility of the system. Utilities in the table are the mean of 1000 rounds of interaction between the agents in the network. While exact utility increases vary according to network structure and the distribution of the values in the payoff matrices, we will see similar values in our numerical studies in Section \ref{sec_numerical_sng} and Appendix C.

\begin{figure}[ht]
    \centering
    \includegraphics[width=0.7\linewidth]{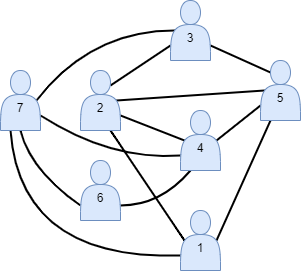}
    \caption[0]{Example Network}
    \label{fig_example0}
\end{figure}

\begin{table}[ht]
\centering
\caption{Comparison of utilities for different values of $\delta$ for the network in Figure \ref{fig_example0}}
\label{tab_uniform_deltas}
\resizebox{.8\linewidth}{!}{
\begin{tabular}{ll|l|l|l|}
\cline{3-5}
                              &        & $\delta_i=0$  & \begin{tabular}[c]{@{}l@{}}$\delta_i=0 \quad\forall i\neq 7$\\ $\delta_7=0.01$\end{tabular} & $\delta_i=2$      \\ \hline
\multicolumn{1}{|l|}{Player}  & Degree & Utility per Round & Utility per Round                                                      & Utility per Round \\ \hline
\multicolumn{1}{|l|}{1}       & 3      & 10.532            & 10.743                                                                 & 11.974            \\
\multicolumn{1}{|l|}{2}       & 4      & 12.880            & 11.469                                                                 & 14.935            \\
\multicolumn{1}{|l|}{3}       & 3      & 10.596            & 10.356                                                                 & 12.271            \\
\multicolumn{1}{|l|}{4}       & 4      & 13.420            & 13.395                                                                 & 15.479            \\
\multicolumn{1}{|l|}{5}       & 4      & 13.100            & 11.499                                                                 & 14.716            \\
\multicolumn{1}{|l|}{6}       & 2      & 9.178             & 8.570                                                                  & 10.289            \\
\multicolumn{1}{|l|}{7}       & 4      & 14.338            & 18.025                                                                 & 16.531            \\ \hline
\multicolumn{1}{|l|}{Average} &        & 12.012            & 12.014                                                                 & 13.749            \\ \hline
\end{tabular}}
\end{table}

After defining the model in Section \ref{sec_network_model}, we will return to this example to understand how the utility is jointly impacted by network position and the distribution of $\delta_i$ across players. We will see that generally, having a higher value of $\delta$ leads to an agent increasing its utility by increasing its number of interactions. We also note that while we constrain agent $i$ to express a single value $\delta_i$ to each of its neighbors for simplicity, in Section \ref{subsec_personal_delta} we will cover how agent $i$ can optimally set personal values $\delta_i(j)$ for each of its neighbors $j$.

The rest of the paper is organized as follows: the remainder of this section conducts a literature review of relevant work, particularly on the subjects of network games and evolutionary biology and psychology. In Section \ref{sec_network_model}, we define the model mechanics and how agents within it behave with complete knowledge, while in Section \ref{sec_unknown_delta}, we define the same functions for agents with incomplete knowledge. In Section \ref{sec_var_delta} we consider the metagame which occurs on top of the system when agents can adjust their levels of trustworthiness, before exploring the model numerically in Section \ref{sec_numerical_sng}. In Section \ref{sec_discussion_sng}, we discuss the numerical results and future directions of our work, then conclude our paper in Section \ref{sec_conclusion_sng}. Additionally, we provide an e-companion \cite{murray2021trust} for the exploration of additional topics related to our system (Appendix A), proofs of some theorems (Appendix B), and additional numerical results (Appendix C).

Readers who are primarily interested in the numerical results may wish to skip Sections \ref{sec_unknown_delta} and \ref{sec_var_delta}: while these sections are necessary to explain the system when players have incomplete information or vary their level of trustworthiness, they are in-depth descriptions of functions which can be grasped intuitively. 

We also acknowledge here that a large portion of this paper comes from work first appearing in the lead author's doctoral dissertation \cite{murray_dissertation}.

\subsection{Literature Review}

Explaining and modeling non-selfish behavior 
is an intriguing problem within Game Theory, one apparently at odds with the idea of Nash equilibria \cite{nash1950}. One situation in which it is explainable is in settings with incomplete information. \cite{kreps1982} shows mathematically that players in multi-stage games may deviate from apparent Nash equilibria to build reputations when information is incomplete, yet still result in personal utility maximization. \cite{kreppsmilgrom1982} considers the finitely repeated prisoner's dilemma and finds that players in this setting who lack information about each other similarly engage in cooperative behavior. Both papers also show that these behaviors do not emerge in settings with perfect information. Another explanation for non-selfish behavior suggested by \cite{ledyard1994} and formalized by \cite{chen2014} is $\alpha$-altruism. $\alpha$-altruism is loosely inspired by Hamilton's rule for kin selection \cite{Hamilton1963}, which defines a linear factor $r$ based on genetic closeness by common descent from shared ancestors; it states that for a non-selfish behavior to occur, the benefit to the recipient times $r$ must be greater than the cost to the provider. $\alpha$-altruism models this through perceived costs for each player, a convex combination of the player's personal cost and the net cost for all players. Hamilton's rule applies only to kin, however, and thus $\alpha$-altruism is less solidly grounded outside of this setting.

However, evolutionary biology offers another explanation for non-selfish behavior: partner selection. Studies such as \cite{barclay2004,barclay2007,sylwester2010,sylwester2013,barclay2013,barclay2016} consider various settings in which participants engage in 2-stage interactions: after random partnerships in the first stage, participants select partners in the second stage. In each study, participants who were generous in the first stage were more desirable as partners in the second; participants were also more likely to be generous in the first stage to build their reputation if they had prior knowledge of the second stage. \cite{barclay2007} also finds that generosity may be faked in the first round to take advantage of the second-round partner. 
% Similar to these empirical studies, \cite{barclay2011} proposes a theoretical model for partner selection within a pool of players, with players competing to help their fellows so as to receive a portion of their fellows' provided help in return. 
\cite{debove2015} conducts an empirical study which demonstrates that generosity and cooperation only tend to arise between partners of relatively similar opportunities. 
\cite{eisenbruch2019} empirically tests partner selection as a motivation for generosity with a competing theory, threat premium, which states that individuals are generous in order to avoid potential conflict or danger, and finds partner selection is a stronger motivator. Each study is also an example of evolutionary game theory, examining how behaviors evolve among groups over time and subsequent interactions.

This paper makes use of the recent concept of a limited-trust equilibrium (LTE) from \cite{murray2021}, a game theoretic modeling method to explain how and to what extent non-selfish behavior takes place in a game. It is explicitly motivated by partner selection and in this paper, the LTE is applied to study partnerships within social networks. 
Social networks are a frequent topic of study in evolutionary game theory; papers such as \cite{abramson2001,hanaki2007,dall2012,ozkancanbolat2016,bolouki2018,naghizadeh2018,jain2020} study how coalitions and cooperative behaviors form naturally within networks under various settings and assumptions. \cite{scata2016} study the problem of seed selection to trigger cooperative behavior in social networks. 
Similarly, \cite{aral2014} conducts a large-scale experiment to identify network structures that increase peer influence effects. For the interested reader, \cite{szabo2007} surveys evolutionary game theory through 2007, and \cite{jackson2015} provides surveys a larger class of games in social networks through 2015.

However, none of the papers mentioned above study partner selection in conjunction with social networks. To the best of our knowledge, there are only two other than ours which do so. The first is \cite{fu2008}. It finds that frequent partner switching helps to dissuade defection in the prisoner's dilemma, as selfish individuals quickly lose potential partners. As in our setting, players make partner selections based on reputation and past observations of their two-hop neighborhoods. However, these agents select partner groups by altering their one- and two-hop neighborhoods within the network and updating their own reputations by mimicking their successful neighbors rather than determining best responses. We consider a broader class of games than \cite{fu2008}, generated from arbitrary distributions which model any interactions including the prisoner's dilemma. The second is \cite{song2019}, in which agents similarly alter their local neighborhoods to determine who they can interact with. They make use of discounted future horizons to determine these alterations, with interactions taking place between all neighbors as a variant of Public Goods games. We consider these games within the limited-trust setting, which allows agents to avoid considering future payments as in \cite{song2019}, and find our results congruent with the recent papers on partner selection mentioned above. To the best of our knowledge, the limited-trust equilibrium is the first game theoretic model of trust. By extending it to social networks here, we provide the first mathematically precise study of trust in partner selection, unlike previous studies that focused on human experiments.  

\section{Game Model}\label{sec_network_model}

\subsection{Preliminary Concepts}
Before defining the systems we consider, we begin with a review of some standard concepts in game theory.

\begin{definition}[Strategy Profile of a Finite Game]
Given a finite $N$-player game in which each player $i$ has a set $\Sigma_i$ of pure non-mixed strategies, a valid pure strategy profile for the game is given by $\sigma=\{\sigma_1,\sigma_2,...,\sigma_N\}$ where $\sigma_i\in\Sigma_i$ is the pure strategy played by agent $i$.
\end{definition}
Note that $\Sigma_i$ is the set of pure strategies, not the probability simplex of mixed strategies over them.
\begin{definition}[Stackelberg Equilibrium]
A $2$-player leader-follower (Stackelberg) game displays a pure Stackelberg equilibrium $\{\sigma_1,\sigma_2\}$ when player 2 is playing its utility-maximizing response to player 1's strategy, and any deviation by player 1 from $\sigma_1$ to a new strategy $\sigma_1'$ will result in utility $u_1(\sigma_1',\sigma_2') \leq u_1(\sigma_1,\sigma_2)$ after player 2 makes its own utility-maximizing response $\sigma_2'$ to $\sigma_1'$.
\end{definition} 

% \JGn{Do we need the next paragraph at all? We are not anyway using mixed strategies so why bring this, and we are not using more than 2 players Stackelberg equilibria.}
% While it is common to consider mixed strategies, we will be considering leader-follower games with full knowledge. All such games possess pure-strategy Stackelberg equilibria, and any mixed strategy equilibria are convex combinations of pure equilibria. Therefore, we will not need to consider mixed strategies in this paper. 
% Additionally, although the Stackelberg equilibrium can be easily extended to the $N$-player setting, we will consider only one-on-one interactions between players which eliminates the need for this extension. 

We are interested in a related concept, the Limited-Trust Stackelberg Equilibrium (LTSE) \cite{murray2021}. The LTSE similarly possesses pure-strategy equilibria and will govern player interactions. {In a limited-trust game} each player $i$ has a {trust-level} $\delta_i \geq 0$ which it is willing to give up from its greedy best-response (the strategy which maximizes its own utility given the strategy of {the other player}) provided that doing so increases the net utility of all players. {Players are not motivated to do so by generosity or altruism but by a desire to promote similar actions in other players which they can later benefit from.}  %other players by at least that amount. 
{In the context of \cite{mui2002}, to an external agent $j$ $\delta_i$ is the reputation of agent $i$, with limited-trust providing a mechanism to translate this reputation into the behavior of corresponding trustworthiness. In contrast, agent $i$ views $\delta_i$ as the degree to which it can be trusted to behave in cooperative behavior. In Section \ref{sec_unknown_delta} we will consider what happens when $i$'s trust level (the true value of $\delta_i$) is not aligned with its reputation ($j$'s perception of $\delta_i$) but until then we will assume they are aligned.}

As an example, consider the 2-player game in Table \ref{tab_ex_game} in which player $2$ must decide between two strategies $a_2$ and $b_2$. Suppose that player 1 has selected $a_1$, and so player 2 must decide between $u_1(a_1,a_2) = 4,u_2(a_1,a_2)=3$ if it selects $a_2$ and $u_1(a_1,b_2) = 2,u_2(a_1,b_2)=4$ if it selects $b_2$. 
\begin{table}[ht]
\centering 
\caption{Example $2\times 2$ game}
\label{tab_ex_game}
% \resizebox{0.27\textwidth}{!}{
\begin{tabular}{llll}
 &  & \multicolumn{2}{l}{Player 2} \\
 &  & $a_2$ & $b_2$ \\ \cline{3-4} 
{Player 1} & \multicolumn{1}{l|}{$a_1$} & \multicolumn{1}{l|}{4,3} & \multicolumn{1}{l|}{2,4} \\ \cline{3-4} 
 & \multicolumn{1}{l|}{$b_1$} & \multicolumn{1}{l|}{3,2} & \multicolumn{1}{l|}{1,3} \\ \cline{3-4} 
\end{tabular}
% }
\end{table}
Suppose that $\delta_2 = 2$. Then the second player's best response to the first player is to play $a_2$, as it maximizes net utility $(4+3 > 2+4)$ and results in an acceptable loss of 1 from player 2's greedy best response,  given $\delta_2\geq 1$. {This is a common occurrence, as \cite{murray2021} provides empirical evidence that the limited-trust concept provides higher expected net utility for both Stackelberg and simultaneous games generated from several distributions.}
% {blue}{This evidence is also displayed here for a $2\times2$ game in which all values payoff values are drawn independently from a uniform $U[-0.5,0.5]$ distributions in Figure \ref{fig_lte_se_comparison}. The figure demonstrates how price of anarchy, a game theoretic measure of inefficiency, and the Stackelberg gap, difference between the value of the Stackelberg equilibrium and the socially optimal solution, both shrink as the players adopt increasing $\delta$ values.} 
This allows all players to benefit in the long run by avoiding inefficient equilibria which only benefit one player. 

% \begin{figure}[ht]
%     \centering
%     \includegraphics[width = \linewidth]{LTE_SE_comparison}
%     \caption{Game efficiency as a function of $\delta$}
%     \label{fig_lte_se_comparison}
% \end{figure}

In a 2-player leader-follower game the limited-trust best response of the follower to the leader playing $s_1\in\Sigma_1$ is
\begin{align*}
r_2(s_1,\delta_2) =\arg\max_{s_2\in \Sigma_2} \quad  &u_1(s_1,s_2)+u_2(s_1,s_2)\\
\text{s.t.} \quad & u_2(s_1,G_2(s_1)) - u_2(s_1,s_2)  \leq \delta_2,
\end{align*}
where $G_2(s_1) = \arg\max_{s_2\in \Sigma_2}u_2(s_1,s_2)$  is the follower's greedy best response. $r_2(s_1,\delta_2)$ is thus the strategy that maximizes net utility, subject to the constraint that the follower does not give up more than $\delta_2$ than it could have obtained from the greedy best response. The leader's limited-trust optimal strategy is 
\resizebox{\linewidth}{!}{
  \begin{minipage}{\linewidth}
\begin{align*}
    s_1^*(\delta_1,\delta_2) = \arg\max_{s_1\in \Sigma_1} \quad & u_1(s_1,r_2(s_1,\delta_2))+u_2(s_1,r_2(s_1,\delta_2)) \\
    \text{s.t.} \quad & u_1(G_1(\delta_2),r_2(G_1(\delta_2),\delta_2)) - u_1(s_1,r_2(s_1,\delta_2)) \leq \delta_1,
\end{align*}
\end{minipage}
}
where $G_1(\delta_2) = \arg\max_{s_1\in \Sigma_1}u_1(s_1,r_2(s_1,\delta_2))$ is the leader's greedy best strategy, given the limited-trust best response which will be made by the follower.

% \JGn{Do we need the next paragraph since we only use 2-player LTSE?}
% Next, we define the 2-player LTSE, which will be relevant for this paper, and the $N$-player extension can be easily intuited from this and the definition of a Stackelberg equilibrium.

\begin{definition}[Limited-Trust Stackelberg Equilibrium]
A strategy pair $(s_1,s_2)$ in a 2-player limited-trust Stackelberg game with trust levels $\delta_1,\delta_2$ is said to be a limited-trust Stackelberg equilibrium if and only if $s_1 \in s_1^*(\delta_1,\delta_2)$ and $s_2 \in r_2(s_1,\delta_2)$ (if each player plays its Stackelberg limited-trust best response or strategy to the other).
\end{definition}

Note that when $\delta_1=\delta_2=0$, the LTSE reduces to a Stackelberg equilibrium.

\subsection{System Model}

Having covered the preliminaries, we now introduce our model for interactions in a social network. We consider a system over a social network $G(V,E)$ with a set of vertices $V$ and edges $E$. Vertices represent agents in the network and an edge between vertices {implies} that the corresponding agents can interact. There are no self-loops. Each agent in $G$ is self-interested and seeks to maximize its own utility. However, direct interactions between agents occur only in a one-on-one setting through 2-player limited-trust Stackelberg games. As such, each agent $i$ has a trust level $\delta_i$ which serves as its reputation and governs its individual interactions with other agents. {As it is odd to assume that agent $i$ interacts with each of its neighbors with the same level of trust $\delta_i$, in Section \ref{subsec_personal_delta} we will extend the system so that agent $i$ may have an individual trust level $\delta_{i}(j)$ for each of its neighbors $j$.} The system as a whole can be considered as an $N$-player utility maximization game, where $N=|V|$.

Let $N^1_i$ represent the one-hop neighborhood of agent $i$ in $G$: $N^1_i$ is the set of all agents $j$ for which $(i,j)\in E$. {Note that as there are no self-loops $i$ cannot be its own neighbor and thus $i \notin N^1_i$.} Define the $2$-hop neighborhood of $i$, the set of agents (other than $i$) who are not in $N^1_i$ but have neighbors in $N^1_i$, as
$$N^2_i = \left(\bigcup_{j\in N^{1}_i}N^1_j\right)\setminus (N^1_i \cup \{i\}). $$
% represent the $2$-hop neighborhood of $i$, the set of agents (other than $i$) who are not in $N^1_i$ but have neighbors in $N^1_i$. 
For one time period in the system, agent $i$ may invite at most $k_i{\in \mathcal{Z}^+}$ of its neighbors in $N^1_i$ to interact. If agent $j\in N^1_i$ accepts an invitation from $i$, they engage in a leader-follower game with leader $i$ and follower $j$ over payoff matrices $A$ and $B$, respectively, such that $A\sim \mathcal{A}_{ij}$ and $B\sim \mathcal{B}_{ij}$ where $\mathcal{A}_{ij},\mathcal{B}_{ij}$ are probability distributions for interactions between $i$ and $j$ initiated by $i$. 

% \RN{Should provide some motivation for this type of model. Something like this.... We observe that in real-life collaborations, humans initiate and lead activities that require significant time and effort and there are limited such ``games'' that they can initiate. On the other hand they can join in as many games that others are willing to lead. With this in mind....  Think about Dunbar number or related concept}
While agent $i$ may issue at most $k_i \leq |N^1_i|$ invitations per time period, it may accept as many as it receives. This consideration is motivated by the fact that it is easy for an individual to take a supporting role in many endeavors, but it only has the time or resources to take a lead role in a small number. Further, while $i$ may both {issue and} receive an invite from a neighbor $j\neq i$, leading to two separate interactions, it may not issue more than one invitation to $j$ within a single time period. Agent $i$ may have up to a maximum of $k_i+|N^1_i|$ interactions per time period, if all invitations it issues are accepted and all neighbors issue it an invitation. An interaction in which $i$ invites $j$ can thus be fully characterized by $\theta_{ij}=\{\mathcal{A}_{ij},\mathcal{B}_{ij},\delta_i,\delta_j\}$, with expected utilities $u_i(\theta_{ij}),u_j(\theta_{ij})$ for each player. Agent $j$ accepts $i$'s invitation provided $u_j(\theta_{ij})\geq 0$. $i$ will choose to invite (at most) $k_i$ of its neighbors, selecting neighbor $j\neq i$ if it provides one of the $k_i$ highest values for $u_i(\theta_{ij})$ in $N_i^1$. {While it may appear contradictory for an agent to select partners to maximize its personal utility while then interacting in a limited-trust manner (i.e.,  focusing on maximizing net utility), the contradiction disappears when considering $\delta$ as a reputation parameter: behaving well is the cost paid by the agent to have access to more interactions}. Such coopetition settings arise naturally in evolutionary studies of personal interactions such as those mentioned in  \cite{barclay2004,barclay2007,sylwester2010,sylwester2013,barclay2013,barclay2016}. This is subject to $u_i(\theta_{ij}),u_j(\theta_{ij}) \geq 0$ as otherwise the interaction will cost at least one of $i$ or $j$.

% \RN{Should we say that one player cannot play more than 1 game with one person?}

% \RN{I believe our approach is motivated by and implicitly assumes that each game has positive utility. I understand that some math goes through, but we should discuss....}\TM{At present all of the math allows for general utility. I agree that as our use case we should only discuss positive utility or positive utility with fixed interaction cost, but it seems like such an easy extension to make as you may not wish to interact with someone who wants to work with you.}

Given how $i$ determines who to invite and which invitations to accept, we can characterize all behavior in the system if we know $\theta=\{G,\mathcal{A},\mathcal{B},\delta\}$ where $\mathcal{A}=\{\mathcal{A}_{ij}\}_{(i,j)\in E}$, $\mathcal{B}=\{\mathcal{B}_{ij}\}_{(i,j)\in E}$, $\delta = \{\delta_i\}_{i\in [N]}$. Let $K^1_i$ be the set of neighbors that $i$ invites to interact and $K^2_i$ be the set of neighbors that invite $i$ to interact. Agent $i$'s expected net utility is
$$\mathbf{u}_i(\theta) = v_i(\theta) + w_i(\theta),$$
where $v_i(\theta) =\sum_{j\in K_i^1}{u_i(\theta_{ij})}$, the value of the games $i$ initiates which are accepted, and $w_i(\theta) = \sum_{j \in K_i^2}{u_i(\theta_{ji})}$, the value of the games $i$ accepts invitations to. Note that $K^1_i$ can be determined from knowledge of $N_i^1$, and $K^2_i$ can be determined from knowledge of $N_i^1\cup N_i^2$, meaning that agent $i$'s interactions depend only on its 1- and 2-hop neighborhoods, not the network as a whole.

To illustrate these concepts more concretely, consider agent 1 in the network in Figure \ref{fig_example0}. Agent 1 has a 1-hop neighborhood $\{2,5,7\}$ and a 2-hop neighborhood of $\{3,4,6\}$. $v_1(\theta)$ will be determined by the 1-hop neighborhood, with agent 1 inviting the $k>0$ members who will accept the invite and provide the most utility. $w_1(\theta)$ will also be provided by the 1-hop neighborhood, but agent 1 must compete with the members  of $\{N_2^1,N^1_3,N^1_7\}$: if $k=1$ and agent 2 decides agent 4 provides more utility than agent 1, then 1 will not receive utility from agent 2 in $w_1(\theta)$. Thus agent 1's utility in $v_1$ is determined by its own 1-hop neighborhood $N^1_1$, and in $w_1$ by $N^1_2 \cup N^1_3 \cup N^1_7 = N^1_1 \cup N^2_1$.

\begin{lemma}\label{thm_monotonic_follower}
Given a 2-player limited-trust Stackelberg game between a leader $i$ and a follower $j$, $u_j(\theta_{ij})$ increases monotonically as $\delta_i$ increases.
\end{lemma}
\begin{proof}
Consider a follower $j$ with fixed $\delta_j$. For any action $s_i$ the leader $i$ takes, $j$ has a deterministic response $r_2(s_i,\delta_j)$. Note that $r_2$ is not a function of $\delta_i$, so $j$'s response is fixed for fixed $\delta_j$. Suppose that for given $\delta_i$, player $i$ takes action $a$ and that for $\delta_i' = \delta_i+\varepsilon$, $\varepsilon>0$, player $i$ takes action $b$. Given $r_2$ is not a function of $\delta_i$ it must be that the reason $i$ switches to $b$ when operating under $\delta_i'$ is that it increases the net utility, but results in a loss of more than $\delta_i$ from $i$'s greedy best strategy $G_1(\delta_j)$. Given $i$'s utility decreases and the net utility increases, it must be that $j$'s utility increases. %\RN{~~ Is it better to use equations rather than words?}
\end{proof}

\begin{corollary}\label{thm_monotinic_leader}
Given a 2-player limited-trust Stackelberg game between a leader $i$ and a follower $j$, $u_i(\theta_{ij})$ decreases monotonically as $\delta_i$ increases.
\end{corollary}

\begin{corollary}\label{cor_monotinic_leader_net}
Given a 2-player limited-trust Stackelberg game between a leader $i$ and a follower $j$, net utility increases monotonically as $\delta_i$ increases.
\end{corollary}

While Lemma \ref{thm_monotonic_follower} shows that for any fixed game the follower $j$ can only benefit if $\delta_i$ of the leader $i$ increases, the same is not true for $i$ if $\delta_j$ increases. However, \cite{murray2021} provides strong empirical evidence that in games randomly generated from several types of distributions, the utility of the leader has a strong positive correlation to the $\delta$ of the follower. With that in mind, we make the assumption that given two players $l$ and $j$ such that $\mathcal{A}_{il} = \mathcal{A}_{ij}$, player $i$ sends an invitation to whichever of the two has a higher $\delta$, and is indifferent between them if $\delta_j=\delta_l$.

% \begin{corollary}\label{thm_monotinic_leader}
% Given a 2-player limited-trust Stackelberg game between a leader $i$ and a follower $j$, $u_i(\theta_{ij})$ decreases monotonically as $\delta_i$ increases.
% \end{corollary}
% \begin{proof}{Proof}
% Consider a follower $j$ with fixed $\delta_j$. $j$ responds deterministically to any action taken by the leader $i$. Suppose that for given $\delta_i$, player $i$ takes action $a$ and that for $\delta_i' = \delta_i+\varepsilon$, $\varepsilon>0$, player $i$ takes action $b$. Given the increase from $\delta_i$ to $\delta_i'$, it must be that $b$ results in greater net utility than $a$. However, the fact $i$ took action $a$ when using $\delta_i$ implies that $u_i(a) > u_i(b)$, as $i$ would have taken $b$ when using $\delta_i$ if $u_i(a) \leq u_i(b)$, as the net utility of $b$ is higher. Thus, $u_i$ monotonically decreases as $\delta_i$ increases. 
% \end{proof} 

\begin{corollary}\label{cor_monotonic_v}
Given a network $G$ in which all games between any two players have nonnegative expected utilities and are independent, $v_i(\theta)$ is monotonically decreasing with $\delta_i$.
\end{corollary}
\begin{proof}%{Proof:}
Given all games have nonnegative expected utility for $i$ and a game between players $i$ and $j$ is independent of the utility in a later game between $i$ and $l$ or $j$ and $h$, all players accept any games they are invited to. Therefore, by Lemma \ref{thm_monotinic_leader} every term in the sum $v_i(\theta) =\sum_{j\in K_i^1}{u_i(\theta_{ij})}$ decreases monotonically with $\delta_i$ and so $v_i(\theta)$ decreases monotonically with $\delta_i$. 
\end{proof}

Note that the lemma and corollaries do not imply that the utility of the follower $j$ in a specific game decreases monotonically with $\delta_j$: games can be constructed where $j$'s utility increases with $\delta_j$. Intuitively, these games reflect situations in which the leader $i$ can trust $j$ not to take advantage of its strategy $s_i$, allowing both players to benefit. However, \cite{murray2021} again provides empirical evidence that $j$'s \textit{expected} utility decreases monotonically with $\delta_j$ for games generated from several distribution types.

\begin{theorem}\label{thm_cont_in_delta}
Given two agents $i,j$ with continuous distributions $\mathcal{A}_{ij},\mathcal{B}_{ij}$ in which any element $d$ which is dependent on any other set of elements $D$ has a continuous marginal distribution function $f_{d|D}$ for any realization of the elements of $D$, $u_i(\theta_{ij})$ and $u_j(\theta_{ij})$ are both continuous and have finite variance in $\delta_i$ and $\delta_j$ for all $\delta_i, \delta_j \geq 0$ provided $u_i(\theta_{ij}),u_j(\theta_{ij}) < \infty$.
\end{theorem}
\begin{proof}%{Proof:}

Consider the set of games $C\subseteq (\mathcal{A}_{ij},\mathcal{B}_{ij})$ for which $u_i(C,\delta_i,\delta_j)$ is discontinuous on the interval $\delta_i\in[x,x+\varepsilon)$ for $\varepsilon>0$. Note that in each of these games $(A,B)\in C$, $u_i(A,B,\delta_i,\delta_j)$ is a constant-valued step function where it is not discontinuous. 
By definition as an expected value, 
\begin{align*}
u_i(\theta_{ij}) &= u_i(\mathcal{A}_{ij},\mathcal{B}_{ij},\delta_i,\delta_j)\\ 
&= \int_{\mathcal{A}_{ij},\mathcal{B}_{ij}}f_{ij}(A,B)u_i(A,B,\delta_i,\delta_j)dAdB
\end{align*}
where $f_{ij}$ is the distribution function over $(\mathcal{A}_{ij},\mathcal{B}_{ij})$. Therefore, as $\varepsilon\rightarrow 0$ 
$$u_i(\mathcal{A}_{ij},\mathcal{B}_{ij},x+\varepsilon,\delta_j) -  u_i(\mathcal{A}_{ij},\mathcal{B}_{ij},x,\delta_j)$$
\begin{align*}
= \int_{C}f_{ij}(A,B)\left(u_i(A,B,x+\varepsilon,\delta_j) - u_i(A,B,x,\delta_j)\right)dAdB
\end{align*}
due to $u_i(A,B,\delta_i,\delta_j)$ being a constant-valued step function. Note the change in the limits due to $\left(u_i(A,B,x+\varepsilon,\delta_j) - u_i(A,B,x,\delta_j)\right)=0$ for $A,B\notin C$.

As $\varepsilon\rightarrow 0$, then for all $\delta_i \geq 0$ $C\rightarrow\emptyset$ by the continuity of the $\mathcal{A}_{ij},\mathcal{B}_{ij}$ and all marginal distributions therein. Therefore, $u_i(\mathcal{A}_{ij},\mathcal{B}_{ij},x+\varepsilon,\delta_j) -  u_i(\mathcal{A}_{ij},\mathcal{B}_{ij},x,\delta_j)$ goes to 0 because $u_i(\theta_{ij}),u_j(\theta_{ij}) < \infty$ and $u_i(\theta_{ij}),u_j(\theta_{ij})$ have finite variance. Therefore, $u_i(\theta_{ij})$ is continuous in $\delta_i$.

Identical arguments show that $u_i(\theta_{ij})$ is continuous in $\delta_j$, and that $u_j(\theta_{ij})$ is continuous in both $\delta_i$ and $\delta_j$. 
\end{proof}

In this paper we assume that $u_i(\theta_{ij}), u_j(\theta_{ij}) < \infty$ and have finite variance for all $\delta_i, \delta_j \geq 0$, all $i,j\in [N]$. We will use Theorem \ref{thm_cont_in_delta} in Section \ref{sec_var_delta} for games in which agents are able to change their value of $\delta$. Such a result is desirable in this setting because, unlike in \cite{fu2008} and \cite{song2019} where agents choose to shift their network structure between rounds of play and make a binary decision of whether or not to play cooperatively, $\delta$ is a continuous-valued variable, and so it is necessary that the utility of a random game distribution is continuous in $\delta$ in order for players to select the optimal value.

We now return to the example we gave in Figure \ref{fig_example0}. Suppose that $k_i=2$ for all agents $i$ and all games between any players are nonnegative and independent and identically distributed. This means $\mathcal{A}_{ij} = \mathcal{B}_{lh}$ for $i,j,l,h \in [N]$ where $[N]=\{1,2,...,N-1,N\}$. Further suppose that $\delta_i<\delta_{i+1}$ for $i\in [6]$; in particular $\delta_i= \frac{2(i-1)}{3}$ for $i\in [7]$. We can predict exactly who will invite whom to interact (since all outcomes are nonnegative, the expected utility for any interaction for both leader and follower is nonnegative and all invitations will be accepted). The behavior is fully characterized by Table \ref{tab_example1}.

\begin{table*}[ht]
\centering
\caption{Behavior of Network in Figure \ref{fig_example0} under Different $\delta$ Distributions}\label{tab_example1} 
\resizebox{0.7\textwidth}{!}{
\begin{tabular}{ll|l|l|l|l|l|l|}
\cline{3-8}
 &  & \multicolumn{3}{l|}{$\delta_i = 2(i-1)/3$} & \multicolumn{3}{l|}{$\delta_i = 2(7-i)/3$} \\ \hline
\multicolumn{1}{|l|}{Player} & Degree & Invites & Invited By & Utility per Round & Invites & Invited By & Utility per Round \\ \hline
\multicolumn{1}{|l|}{1} & 3 & 5,7 & $\emptyset$ & 7.750 & 2,5 & 2,5,7 & 15.559 \\
\multicolumn{1}{|l|}{2} & 4 & 5,4 & $\emptyset$ & 7.360 & 1,3 & 1,3,4,5 & 20.446 \\
\multicolumn{1}{|l|}{3} & 3 & 5,7 & 5 & 11.040 & 2,5 & 2,7 & 12.859 \\
\multicolumn{1}{|l|}{4} & 4 & 6,7 & 2,5,6,7 & 21.064 & 2,5 & 6 & 9.960 \\
\multicolumn{1}{|l|}{5} & 4 & 3,4 & 1,2,3 & 15.193 & 1,2 & 1,3,4 & 18.059 \\
\multicolumn{1}{|l|}{6} & 2 & 4,7 & 4,7 & 13.754 & 4,7 & $\emptyset$ & 6.607 \\
\multicolumn{1}{|l|}{7} & 4 & 4,6 & 1,3,4,6 & 18.931 & 1,3 & 6 & 11.244 \\ \hline
\multicolumn{1}{|l|}{Average} &  &  &  & 13.591 &  &  & 13.559 \\ \hline
\end{tabular}
}
\end{table*}

Table \ref{tab_example1} shows the interaction between $\delta$ and network structure. When $\delta_i = \frac{2(i-1)}{3}$, agent 7 is invited to play by each of its neighbors in the network. This is unsurprising as $\delta_7 > \delta_{j\neq 7}$. So is agent 6 as $k_i=2$ for all $i\in [7]$ and the only agent for which $\delta_j > \delta_6$ is $j=7$. What is more interesting is that agent 4 is being invited to play by all of its neighbors, and is engaging in as many games per round as agent 7. Further, it is engaging with agents that have a higher $\delta$ than agent 7's partners, both as a leader and as a follower, so we expect that it achieves a higher utility per round than agent 7, especially because it is behaving more selfishly. If we define $\mathcal{A}_{ij}=\mathcal{B}_{lh}$ to be a probability distribution over $2\times2$ matrices with all entries generated  iid from an exponential distribution with $\lambda=2$, we see that this expectation is confirmed. Column 5 shows the average utility each agent receives per round after 1000 rounds of play. Columns 6-8 consider the same setting, but when $\delta_i = \frac{2(7-i)}{3}$, reversing which agents are the most valuable partners. 
% \RN{Do you want to say something about the gain over 12.012? = 13\%}

We also remind the reader of the information in Table \ref{tab_uniform_deltas} which considers the same values of $k$ and $\mathcal{A}_{ij},\mathcal{B}_{ij}$. For the two cases in Table \ref{tab_example1} the average value of $\delta$ is 2, so it is unsurprising that they have roughly the same average utility as when $\delta$ is uniformly equal to 2 for all agents in Table \ref{tab_uniform_deltas}. With that being said, the more even distribution of a uniform $\delta=2$ produces higher average utility, an increase of approximately 14.5\% over $\delta=0$, compared to Table \ref{tab_example1} in which the first case shows an increase of 13.1\% and the second shows an increase of approximately 12.9\%. In Section \ref{sec_numerical_sng} we numerically examine the relationship between network structure and $\delta$.

We saw that the behavior of systems with fixed, known $\theta$ can be characterized and readily predicted. We now focus on when parts of $\theta$ are unknown or are not fixed. Section \ref{sec_unknown_delta} focuses heavily on the algorithmic methods agents use to learn the $\delta$ of their neighbors, and Section \ref{sec_var_delta} mathematically details how agents adjust their $\delta$ to maximize their utility in response to their 1- and 2-hop neighborhoods. For those who are concerned primarily with the results of our numerical {studies}, we recommend skipping ahead to Section \ref{sec_numerical_sng}.

\section{Learning under unknown $\delta$}\label{sec_unknown_delta}

In this section, we consider how agents behave with incomplete information about their neighbors' $\delta$ values. 
{In other words, an agent's reputation (its neighbors' perception of $\delta$) does not match the agents' true trustworthiness.}
The expected utility agent $i$ gains by interacting with agent $j$ is dependent on both the utility $j$ brings ($\mathcal{A}_{ij}$ and $\mathcal{B}_{ij}$) and its trustworthiness being aligned with its reputation ($i$'s perception of $\delta_j$ being accurate). These parameters are independent, so each may be estimated separately. $\mathcal{A}_{ij},\mathcal{B}_{ij}$ are multi-dimensional distributions and can be estimated using standard statistical methods if they are not known \textit{a priori}. Therefore we focus on how agent $i$ estimates $\delta_{-i}$, where $\delta_{-i} = \{\delta_j\}_{j\in [N]\setminus\{i\}}$ is the set of $\delta$ values for all agents $j\neq i$.

Because much of this section considers only interactions between two players in a game rather than agents in a larger network, we will use the terms ``player'' and ``agent'' interchangeably here.

\subsection{Learning $\delta_{-i}$ as Leader}\label{subsec_learn_leader}

Consider an $m\times n$ leader-follower game with leader player 1 and follower player 2. Assume that player 1 knows through past observations that $\delta_2\in [\delta_{21}^l,\delta_{21}^u)$, 
% \RN{There must be a reason for the half-open interval; worth mentioning? Place answer from later here} 
a pair of lower and upper bounds. The interval $[\delta_{21}^l,\delta_{21}^u)$ is half-open because when we observe $\delta^l_{ji}$ being given up we know $\delta_j \geq \delta^l_{ji}$, but when $\delta^u_{ji}$ is not given up all we know is $\delta_j < \delta^u_{ji}$. Suppose that player 1 has selected strategy $s_i$ to play. Then the game is equivalent to the $1\times n$ game given in Table \ref{tab_1ngame}. 
% \JGn{We are using players and agents interchangeably. Please mention this in the beginning.}

\begin{table}[ht]
\centering
\caption{$1\times n$ Leader-Follower Game}
\label{tab_1ngame}
% \resizebox{0.4\textwidth}{!}{
\begin{tabular}{llcll}
 &  & \multicolumn{3}{c}{Player 2} \\
 &  & $s_1$ & ... & $s_n$ \\ \cline{3-3} \cline{5-5} 
Player 1 & \multicolumn{1}{l|}{$s_i$} & \multicolumn{1}{l|}{$(a_1,b_1)$} & \multicolumn{1}{l|}{...} & \multicolumn{1}{l|}{$(a_n,b_n)$} \\ \cline{3-3} \cline{5-5} 
\end{tabular}
% }
\end{table}

% Without loss of generality assume that $b_1 \geq b_l\forall l\in [n]$.
Player 1 can refine its knowledge of $\delta_{2}$ based on player 2's response by considering the Pareto frontier of player 2's strategies measured in the values of $u_2$ and $u_1+u_2$. Assume that there are $k$ strategies on the frontier and they are relabeled $\{s_1,s_2,...,s_k\}$ such that $u_2(s_1)>u_2(s_2)>...>u_2(s_k)$ and $u_1(s_1)+u_2(s_1)<u_1(s_2)+u_2(s_2)<...<u_1(s_k)+u_2(s_k)$. Figure \ref{fig_pareto_leader} depicts such a frontier with $k=5$. If player 2 plays $s_j$ in response, then it must be that $b_1-b_j \leq \delta_2 < b_1-b_{j+1}$. Let $b_{k+1} = -\infty$ for the case where $j=k$. %$\delta_{21}^l$ and $\delta_{21}^u$ can then be updated if better bounds have been found. 

\begin{figure}[ht]
    \centering
    \includegraphics[width=0.85\linewidth]{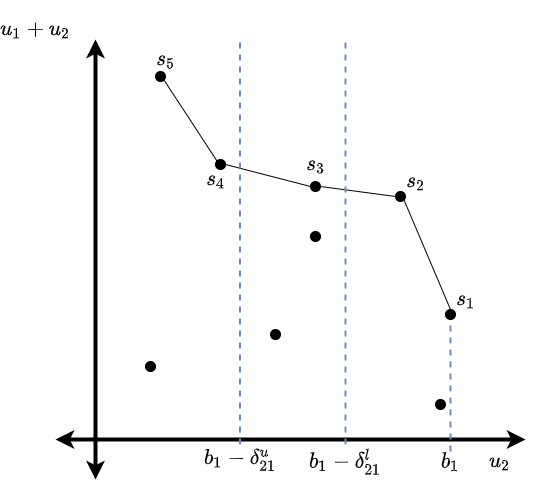}
    \caption{Pareto frontier of game/strategy in Table \ref{tab_1ngame} for the Leader.}
    \label{fig_pareto_leader}
\end{figure}

Analyzing the Pareto frontier allows player 1 to determine whether a better bound for $\delta_2$ is found and whether $\delta_2$ has changed. Consider Figure \ref{fig_pareto_leader} again: based on player 1's previous bounds for $\delta_2$, it expects player 2 to select either $s_2$ or $s_3$, depending on whether $b_1-b_3 > \delta_{2}$. Thus if player 2 responds with $s_j\in\{s_1,s_4,s_5\}$, there has been a change to $\delta_2$. Algorithm 1 how the leader player 1 notices detectable changes and updates its bounds for $\delta_2$.

For the interested reader, Appendix B contains additional work deriving the expected time for the leader to learn the follower's $\delta$ to arbitrary precision levels.

\begin{center}
\begin{algorithm}[ht]
\footnotesize
    \caption{Leader: Update $\delta_{21}^l,\delta_{21}^u$}
    \label{alg_update_leader}
    % \resizebox{\linewdth}{!}{
    \begin{algorithmic}
        \REQUIRE $\delta_{21}^l,\delta_{21}^u,S=\{s_1,s_2,...,s_k\},j$
            \STATE $l \leftarrow b_1 - b_j$
            \STATE $u \leftarrow b_1 - b_{j+1}$
            \STATE $change \leftarrow False$
            \IF{$l \geq \delta_{21}^u$ \textbf{or} $u \leq \delta_{21}^l$}
                \STATE $\delta_{21}^l \leftarrow l$
                \STATE $\delta_{21}^u \leftarrow u$
                \STATE $change \leftarrow True$
            \ELSE 
                \STATE $\delta_{21}^l \leftarrow \max\{l,\delta_{21}^l\}$
                \STATE $\delta_{21}^u \leftarrow \min\{u,\delta_{21}^u\}$
            \ENDIF
            \STATE \Return $\delta_{21}^l,\delta_{21}^2, change$
    \end{algorithmic}
\end{algorithm}
\end{center}

\subsection{Learning $\delta_{-i}$ as Follower} \label{subsec_learn_follower}

We now show how a follower can learn the $\delta$ value of a leader. Consider an $m\times n$ game with leader player 1 and follower player 2. From past observations player 2 knows $\delta_1\in [\delta_{12}^l,\delta_{12}^u)$. Recall that player 2's best response to player 1 selecting $s_i$ is $s_j^* = r_2(s_i,\delta_2)$. While player 1 does not know $\delta_2$, it does have an estimate $\delta_{21}$. Unless otherwise specified, $\delta_{21} = \frac{\delta_{21}^u+\delta_{21}^l}{2}$. {This estimate is based on player 1 assuming a uniform distribution for $\delta_2$; in the event, it has further distributional knowledge of $\delta_2$ it can of course use that to estimate $\delta_{21}$ on that distribution truncated to the range $[\delta_{21}^l,\delta_{21}^u)$} Thus from player 1's perspective, the game can be rewritten as the $m\times1$ game in Table \ref{tab_m1game}.

\begin{table}[ht]\centering
\caption{$m\times 1$ Leader-Follower Game}
\label{tab_m1game}
% \resizebox{0.3\textwidth}{!}{
\begin{tabular}{lll}
         &                            & Player 2                         \\
         &                            & $r_2(s_i,\delta_{21})$           \\ \cline{3-3} 
         & \multicolumn{1}{l|}{$s_1$} & \multicolumn{1}{l|}{$(a_1,b_1)$} \\ \cline{3-3} 
Player 1 & $\vdots$                   & $\vdots$                         \\ \cline{3-3} 
         & \multicolumn{1}{l|}{$s_m$} & \multicolumn{1}{l|}{$(a_m,b_m)$} \\ \cline{3-3} 
\end{tabular}
% }
\end{table}

If player $2$ knows $\delta_{21}$, player 1's estimate of $\delta_2$, it can construct the $m\times 1$ game that player 1 is considering. Based on past interactions, player 2 can compute $\delta_{21}^l,\delta_{21}^u$ from its own past actions, and therefore compute $\delta_{21}$. Player 2 can then construct a Pareto frontier of player 1's strategies similar to Figure \ref{fig_pareto_leader}, but measured in the values of $u_1$ and $u_1+u_2$. Without loss of generality, if there are $k$ strategies on the frontier assume that $u_1(s_1)>u_1(s_2)>...>u_1(s_k)$ and $u_1(s_1)+u_2(s_1)<u_1(s_2)+u_2(s_2)<...<u_1(s_k)+u_2(s_k)$. Figure \ref{fig_pareto_follower} gives an example with $k=5$. If player 1 selects $s_i$ while anticipating $r_2(s_i,\delta_{21})$ in response, it must be that $a_1-a_i \leq \delta_1 < a_1-a_{i+1}$, where $a_{k+1}=-\infty$ in the case that $i=k$. $\delta_{12}^l$ and $\delta_{12}^u$ are updated if this implies a better bound.

\begin{figure}
    \centering
    \includegraphics[width=0.85\linewidth]{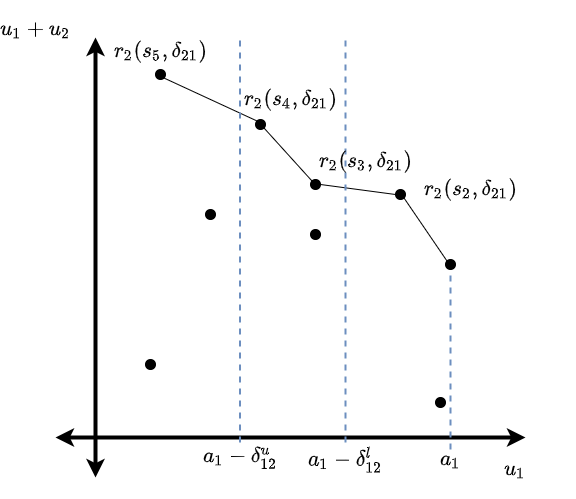}
    \caption{Pareto frontier of game/strategy in Table \ref{tab_m1game} for the Follower}
    \label{fig_pareto_follower}
\end{figure}

Similar to the leader in the previous section, the follower can use the Pareto frontier to determine if $\delta_1$ has changed. For the frontier in Figure \ref{fig_pareto_follower}, player 2's previously derived bounds for $\delta_1$ indicate that if $\delta_1$ hasn't changed, player 1 will select $s_2,s_3,$ or $s_4$, with $s_2$ occurring if $\delta_1 < a_1-a_3$. Thus if player 1 selects $s_1$ or $s_5$, $\delta_1$ has changed. The follower can then use Algorithm 1 with slight modifications (consider $b_1,b_i,b_{i+1}$ rather than $a_1,a_i,a_{i+1}$) to update its knowledge of the leader player 1, while noticing any detectable changes in $\delta_1$.

\subsection{Network Dynamics Under Unknown $\delta$}\label{subsec_network_unknown_delta}

So far we have focused on the learning of unknown $\delta$ between two players. Now we consider the broader network. Recall that each agent $i$ can initiate at most $k_i$ interactions per round. Therefore, any agent $i$ with neighborhood $|N_i^1| > k_i$ faces an exploration-exploitation dilemma: invite neighbor $j\neq i$ where $u_i(\theta_{ij})$ is maximized or invite neighbor $h\neq i$ where $\delta_{hi}^u - \delta_{hi}^l$ is large. 
% \JGn{why not consider both?}

We consider agents which address this dilemma in the following manner: at the beginning of round $t$, agent $i$ invites at most $k_i$ of its neighbors $N_i^1$ to interact. Agent $i$ selects some of these neighbors for the purpose of exploration and some for exploitation. For each neighbor $j$, although $i$ does not know $\delta_j$ it has an estimate $\delta_{ji}$ from past interactions. Unless otherwise specified, $\delta_{ji} = \frac{\delta_{ji}^u + \delta_{ji}^l}{2}$. If $i$ decides to exploit $h(t)$ of its interactions in round $t$, then it selects the set of $S$ neighbors such that $S = \arg\max_{S}\sum_{j\in S}u_i(\theta_{ij}),\text{ }$ subject to $|S|\leq h(t)$ and $u_j(\theta_{ij}) \geq 0 \text{ }\forall{j\in S}$, as otherwise the invitation will be refused. Based on past interactions, agent $i$ is capable of determining the value $\delta_{ij}$ that agent $j$ estimates for $\delta_i$, and can avoid sending an invitation which will be rejected. $h(t)$ is determined according to a multi-armed bandit scheduling policy, such as uniform $\varepsilon$-greedy, -first, or -decreasing, or a more sophisticated policy such as Thompson sampling. {Which policy is appropriate depends on the distribution of game utilities and agent $i$'s observations thus far.} Agent $i$ then randomly samples $k_i-|S|$ of its remaining neighbors for exploration, according to a discrete probability distribution $f_t(N_i^1\setminus S,\delta'_i)$, where $\delta'_i = \{(\delta_{ji}^l,\delta_{ji}^u)\}_{j\in N_i^1}$ is $i$'s estimates of the $\delta$ values of its neighbors. For each invited neighbor $j$ agent $i$ plays according to $\delta_i$ and its estimate $\delta_{ji}$: 
agent $i$ should not forego a large opportunity or accept a large cost while exploring agent $j$. Instead, it will take that into account in determining whether to issue $j$ future invitations.

\section{Network Games with Variable $\delta$}\label{sec_var_delta}

In Sections \ref{sec_network_model} and \ref{sec_unknown_delta}  we defined the mechanics under which network games function. {A} natural extension to consider is how a player $i$ might change $\delta_i$ in order to take advantage of $\delta_{-i}=\{\delta_1,\delta_2,...,\delta_{i-1},\delta_{i+1},...,\delta_N\}$. Therefore, if agent $i$ can change $\delta_{i}$ between rounds, it should set it to 
$$\delta_i^*=\arg\max_{\delta_i\in\Delta_i} \mathbf{u}_i(\delta_i,\theta_{-i}),$$
where $\theta_{-i} = \{G,\mathcal{A},\mathcal{B},\delta_{-i}\}$ and where $\Delta_{i}=[0,\delta_{max}]$ and $\delta_{max}$ is an arbitrarily enforced maximum value of $\delta$ for the system. Note that while $\delta_{max}$ can be arbitrarily large, we require $\delta_{max}\in \mathcal{R^+}$.

Finding $\delta_{i}^*$ is complicated by two factors. First, agent $i$ does not know $\delta_{-i}$. Instead, it has estimates of $\delta_{ji}$ for its neighbors $j\in N_{i}^1$ based on past interactions. This prevents it from accurately determining whether it will receive an invitation from $j$ for a given value of $\delta_i$, and it must estimate the expected value of the game if it does receive the invitation. 
Second, agent $i$'s neighbor $j$ does not know $\delta_{-j}$: agent $j$ will not immediately notice a change in $\delta_i$, and so agent $i$ will not immediately receive the expected utility from shifting to $\delta_i^*$. This comes from $j$ not correctly deciding whether or not to issue an invitation to $i$, as well as not estimating agent $i$'s trust level correctly if they do interact.

With mild re-use of notation, we let $\delta'_{i} = \{\delta_{ji}\}_{j\in N_i^1}$ be agent $i$'s estimates of $\delta_j$ for each of its neighbors $j$. Because we consider social networks, we assume that agents share impressions of their own neighbors during their interactions via gossip. {Although this assumption is made for model tractability it is a theoretically reasonable assumption for social groups.} This leads to agent $i$ knowing $\delta'_j$ for each of its neighbors $j\in N_i^1$. This allows $i$ to address the first complication in finding $\delta_i^*$: by knowing $\delta'_j$, agent $i$ can predict whether or not it will receive an invitation from agent $j$ if it shifts the value of $\delta_i$. 
% \JGn{Are we saying that agent $i$ knows $\delta'_j$'s of each of its neighbours, but not the actual $\delta_j$? This seems strange! Also, since everyone is shifting, then isn't it going in a loop?} \TM{Have we sufficiently addressed/discussed this here and in our last meeting? Possibly admit as a limitation.} 
% \TM{While we realize that the fact that agent $i$ needs to know $\delta'_{j}$ or at least have some idea of it could be considered a weakness in the model, as we have already stated we believe that the fact we are considering a social network makes agent $i$'s possession of this knowledge a reasonable assumption.} 
As noted in Section \ref{sec_network_model}, agent $i$'s expected utility can be determined entirely from its 2-hop neighborhood. $i$ therefore does not need any additional information from other agents $l\notin N_i^1$. While it must estimate the expected value of the game agent $j$ initiates as a function of $\delta_{ji}$ and $\delta_i$, after several interactions, it is likely that $\delta_{ji} \approx \delta_j$. By the continuity implied by Theorem \ref{thm_cont_in_delta}, this means that as $\delta_{ji}\rightarrow \delta_j$, $u_i(\mathcal{A}_{ji},\mathcal{B}_{ji},\delta_{ji},\delta_i)\rightarrow u_i(\theta_{ji})$ and $u_i(\mathcal{A}_{ij},\mathcal{B}_{ij},\delta_{i},\delta_{ji})\rightarrow u_i(\theta_{ij})$. 
% \JGn{do you mean $u_j(.)$ in the previous sentence?} 

We address the second complication heuristically. Because a change in $\delta_i$ takes time to become apparent to agent $i$'s neighbors $j\in N_i^1$, $i$ does not benefit by changing $\delta_i$ frequently. Doing so results in it never gaining the expected utility it computed when determining $\delta_i^*$. Therefore, agent $i$ adjusts $\delta_i$ in a given round with arbitrary probability $p_i$. The number of rounds that each agent $i$ commits to a given $\delta_i$ before reevaluating is thus a geometric random variable. Another heuristic option is for agents to update on an epoch schedule, every $t$ rounds. This allows agent $i$ to better estimate $\delta'_i$, and allows agent $j\in N^1_i$ to better determine $\delta_i$ so that $i$ realizes the value it expected when it set $\delta_i$. It is natural to assume that in this setting without an immediately observable $\delta$ a myopic agent $i$ will select $\delta_i=0$ to take advantage of its neighbors in their next interactions. This is correct, but under both heuristic mechanisms agent $i$ chooses $\delta_i$ which maximizes its utility until the next update. This does mean that agent $i$ is non-myopic when selecting $\delta_i$, as it makes the assumption that its reputation will eventually match the new value of $\delta_i$. Once $\delta_i$ is selected though, each subsequent interaction can continue to be evaluated myopically until $\delta_i$ is updated. We consider both mechanisms in Section \ref{sec_numerical_sng}.

Having covered the mechanics by which $\delta_i$ varies for a given agent $i$, we now address changes to the exploration-exploitation methodology from Section \ref{subsec_network_unknown_delta}. This framework functions well when $\delta$ is fixed for all agents, as the value of exploration decreased with accumulated knowledge. However, whenever $\delta_j$ shifts, agent $i$'s past knowledge of $\delta_j$ becomes obsolete and additional exploration is beneficial. Thus, instead of using $h(t)$ to determine how many of its neighbors to exploit and explore, agent $i$ should also take into account how recently each of its neighbors $j\in N_i^1$ changed $\delta_j$. Algorithm1 1 already reports whether or not a change in $\delta_j$ has been detected, so a vector $t_i$ of how many rounds ago each neighbor $j\in N^1_i$ changed $\delta_j$ is easily maintained and exploration and exploitation can instead be determined by $h(t_i)$.
% \JGn{$x_i$ is used here only. There are too many notations and not very precise. Maybe you want to rephrase it to not use $x_i$.}

\subsection{Personalized $\delta$}\label{subsec_personal_delta}

Before our numerical studies in the next section, we consider the setting in which agent $i$ uses different values of $\delta_i$ depending on which of its neighbors it is interacting with. This reflects when individuals may prefer specific trusted partners or be more willing to help them than they would other acquaintances.  

Let $\delta_i(j)$ be the value of $\delta_i$ agent $i$ uses when interacting with agent $j$. Similarly, let $\delta_{ji}(i)$ be agent $i$'s estimate of $\delta_j(i)$. In many ways this will make the problem of selecting $\delta_i^*$ simpler, despite the fact that agent $i$ now needs to select a vector rather than a single value. This is because $i$ can determine $\delta_i^*(j)$ while only considering $j$ and $N_j^1$, and only taking $N_i^1$ into account at the end. 

The procedure for agent $i$ to determine $\delta_i^*(j)$ is straightforward. First, $i$ determines $\delta_i^F(j)$, where
\begin{align*}
\delta_i^F(j) &= \arg\max_{\delta_i(j)\geq 0} I_1(\mathcal{A}_{ji},\mathcal{B}_{ji},\delta_j(i),\delta_i(j))\\
&*\max\{0,u_i(\mathcal{A}_{ji},\mathcal{B}_{ji},\delta_j(i),\delta_i(j))\}
\end{align*}
and $I_1(\mathcal{A}_{ji},\mathcal{B}_{ji},\delta_j(i),\delta_i(j))$ is an indicator function which is 1 if agent $j$ will issue $i$ an invitation and 0 otherwise. $\delta_i^F(j)$ is the optimal value of $\delta_i^*(j)$ 
% \JGn{do you want $\delta^*_i(j)$? If yes, then please check at other places too.} 
if agent $i$ does not intend to issue $j$ an invitation, but would still benefit from receiving one from $j$. Note that the term $\max\{0,u_i(\mathcal{A}_{ji},\mathcal{B}_{ji},\delta_j(i),\delta_i(j))\}$ indicates that $i$ will decline the invitation if $u_i < 0$.
Next, agent $i$ determines $\delta_i^L(j)$ such that 
\resizebox{\linewidth}{!}{
\begin{minipage}{\linewidth}
\begin{align*}
  \delta_i^L(j) = \arg\max_{\delta_i(j)\geq 0} & I_1(\mathcal{A}_{ji},\mathcal{B}_{ji},\delta_j(i),\delta_i(j)) \\
  &*\max\{0,u_i(\mathcal{A}_{ji},\mathcal{B}_{ji},\delta_j(i),\delta_i(j))\}\\ 
  &+I_2(u_j(\mathcal{A}_{ij},\mathcal{B}_{ij},\delta_i(j),\delta_j(i)) \geq 0)\\ &*u_i(\mathcal{A}_{ij},\mathcal{B}_{ij},\delta_i(j),\delta_j(i)),  
\end{align*}
\end{minipage}
}
% $$\delta_i^L(j) = \arg\max_{\delta_i(j)\geq 0}  I_1(\mathcal{A}_{ji},\mathcal{B}_{ji},\delta_j(i),\delta_i(j))\max\{0,u_i(\mathcal{A}_{ji},\mathcal{B}_{ji},\delta_j(i),\delta_i(j))\} + I_2(u_j(\mathcal{A}_{ij},\mathcal{B}_{ij},\delta_i(j),\delta_j(i)) \geq 0) u_i(\mathcal{A}_{ij},\mathcal{B}_{ij},\delta_i(j),\delta_j(i))$$
where $I_2$ is an indicator variable which is 1 if agent $j$ would accept an invitation from agent $i$ at the specified $\delta_i(j)$.
$\delta_i^L(j)$ represents the optimal value of $\delta_i^*(j)$ if agent $i$ would like to issue an invitation to $j$, as well as potentially receive one.

Having compiled a pair $(\delta_i^F(j),\delta_i^L(j))$ for each neighbor $j\in N_i^1$, agent $i$ now can select $\delta_i(j)$ as one of the two values from each pair. This is subject only to the constraint that $i$ may select $\delta_i(j) =  \delta_i^L(j)$ for at most $k_i$ of its neighbors, as it can only issue at most $k$ invitations. By choosing at most $k_i$ neighbors for which this difference is highest, agent $i$ can determine $\delta_i^*$.

Note that while we assumed $\delta_{-i}$ to be known for simplicity, the procedure for determining $\delta_i^*(j)$ when $\delta_{-i}$ is unknown is analogous, utilizing the techniques from earlier in Section \ref{sec_var_delta}.

\subsection{Variable Known $\delta$}

Before moving on we again pause to consider the case when $\delta_{-i}$ is known to agent $i$. While agents in the network interact with each other in 2-player Stackelberg games, allowing agents to modify their values of $\delta$ between rounds gives them a second, indirect interaction. When an agent selects a value for $\delta$ it does not directly result in utility, but it influences which agents will interact with it as well as how they will interact, which in turn results in utility. For this reason, selecting $\delta$ represents an $N$-player game on top of the system. Because the strategies for this game indirectly influence utility, instead influencing the system that will determine utility, we refer to this as an $N$-player continuous strategy ``metagame'' over the system. As each player displays trust only in its own self-interest, we can show that under certain settings the metagame displays mixed Nash equilibria despite its continuous strategy sets.

\begin{definition}[Mixed Nash Equilibrium]
Given an $N$-player game with strategy profiles $\sigma = (\sigma_1,\sigma_2,...,\sigma_N)$ for each player where for a given player $i$, $\sigma_{-i}$ is the set of strategies played by all other players, $\sigma$ is a mixed Nash equilibrium (MNE) if and only if for any other valid strategy profiles $\sigma_i'$, $u_i(\sigma_i,\sigma_{-i}) \geq u_i(\sigma'_i,\sigma_{-i})$ for all $i \in [N]$, and $u_i(\sigma_i,\sigma_{-i})$ is the expected utility of the game for player $i$ when it plays strategy $\sigma_i$.
\end{definition}

\begin{theorem}\label{thm_sng_uniform_mne}
Consider a social network $G$ with uniform interactions $\mathcal{A}_{ij} = \mathcal{B}_{lh}$ for all $l,i,j,h\in [N]$ such that all payoffs are nonnegative and for agent $i$ with neighbors $j$ and $l$, $\delta_j \leq \delta_l \rightarrow u_i(\theta_{ij}) \leq u_i(\theta_{il})$. Then the $N$-player metagame with closed interval strategy space $\Delta_i\subseteq \mathcal{R}$ and utility function $\mathbf{u}_i$ for $i\in[N]$ possesses a mixed Nash equilibrium.
\end{theorem}
The proof to Theorem \ref{thm_sng_uniform_mne} is given in our e-companion \cite{murray2021trust} in Appendix B. To the best of our knowledge, this is the first such equilibrium result for a continuous metagame over a network system without non-myopic controls. The closest similar result is from \cite{song2019} which explores the requirements in their system for network structure and interaction convergence where portions of the strategy sets are continuous but assumes that agents are non-myopic and punish defectors via grim-trigger strategies. This result also requires that all agent interactions be a specific variation of public goods games, while our metagame allows for agents who interact in a more generalized game without external controls.

\section{Numerical Studies}\label{sec_numerical_sng}
In this section, we empirically examine the relationship between network position and trust level $\delta$ in realistic social networks. We will study Zachary's Karate Club network from \cite{zachary1977}, as well as the ego-Facebook network curated by SNAP \cite{snapnets}. 
The network is a social network representing friendships in a university karate club studied by Wayne Zachary from 1970-72, and is commonly employed as an example of a real-world social network with a community structure.  It is visually represented in Figure \ref{fig_zkc_graph}. 
We will consider separately when $\delta_{-i}$ is known and unknown for agent $i$ in each network. When agent $i$ determines which of its neighbors $j\in N_i^1$ to issue invitations to, $u_i(\theta_{ij})$ is estimated as the mean of $i$'s utility in 1000 games drawn iid at random from $\mathcal{A}_{ij},\mathcal{B}_{ij}$ with $\delta_i$, $\delta_j$. We will confine our attention to distributions where $u_j(\theta_{ij})\geq 0$, so that all invitations will be accepted.

\begin{figure}[ht]
    \centering
    \includegraphics[width=0.85\linewidth]{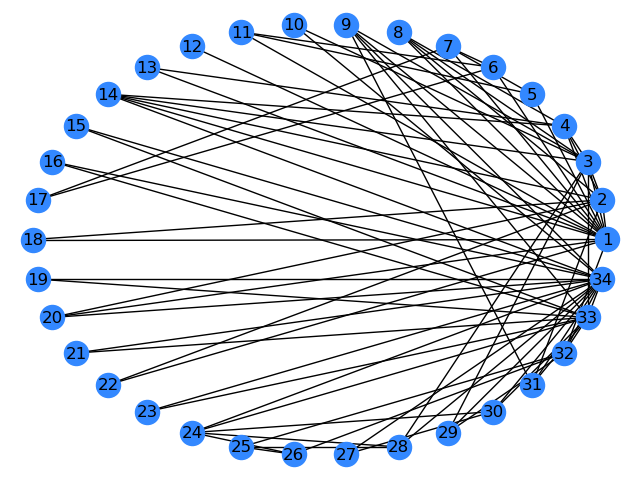}
    \caption{Zachary's Karate Club}
    \label{fig_zkc_graph}
\end{figure}

In addition to the two networks mentioned above, we conduct numerical studies on five other networks from the Konect \cite{konect} network collection. We have omitted them from the main paper due to space constraints, but the results reinforce the conclusions we reach here and can be viewed in the appendices of our e-companion \cite{murray2021trust}.

\subsection{Known $\delta$}
\label{subsec_network_known_delta}

We consider the Karate Club network with the following parameters:

\begin{itemize}
    \item $\mathcal{A}_{ij} = \mathcal{B}_{hl}$ for all $h,i,j,l\in [N]$. $A_{ij} \sim \mathcal{A}_{ij}$ is a $2\times2$ matrix with entries generated iid from the exponential distribution with $\lambda=4$.
    \item $\delta_i\in [0,30]$ $\forall i\in [N]$.
    \item $\delta_i$ is known to all agents.
    \item $\delta_i$ updates between rounds. $\delta_i$ at time $t$ is a greedy best response to $\delta_{-i}$ at time $t-1$. For $t=0$, $\delta_i=0$, $\forall i\in[N]$.
    \item $k_i=2$ for each agent $i$. In the event that $|N_i^1| < 2$, we set $k_i = |N_i^1|$.
\end{itemize}

$u_i(\theta_{ij})$ and $u_j(\theta_{ij})$ as functions of $\delta_i,\delta_j$ are estimated by taking the sample mean utilities of 1000 simulated games generated iid according to $\mathcal{A}_{ij}, \mathcal{B}_{ij}$. 

Figure \ref{fig_zkc_n0_1} illustrates $\mathbf{u}_{1}(\theta_{-1},\delta_1)$ as a function of $\delta_1$ for fixed (random) $\delta_{-i}$, for vertex 1 in the karate club network. The left plot gives the value of $\mathbf{u}_1(\theta_{-1},\delta_1)$ and the right plot shows the number of games the agent at vertex 1 in the Karate Club network participates in per round. We see in the left plot that $\mathbf{u}_1(\theta_{-1},\delta_1)$ is flat or monotonically decreasing with $\delta_1$ at all but a handful of points where there is a sharp increase. These jumps occur when $\delta_1$ is large enough to attract a new agent to interact with agent 1, shown in the plot on the right. This is unsurprising: the utility of a single game for the follower may not monotonically decrease as $\delta$ increases, but as stated in Section \ref{sec_network_model} we strongly suspect that for many distributions it monotonically decreases \textit{in expectation}. Based on Corollary \ref{thm_monotinic_leader} we also expect the flat areas to monotonically decrease over the true distribution, as this figure comes from a sample of 2000 rounds.
% \RN{I don't understand this.} \TM{Rewrote, is it more clear?}

\begin{figure}
\centering
  \centering
  \includegraphics[width=0.85\linewidth]{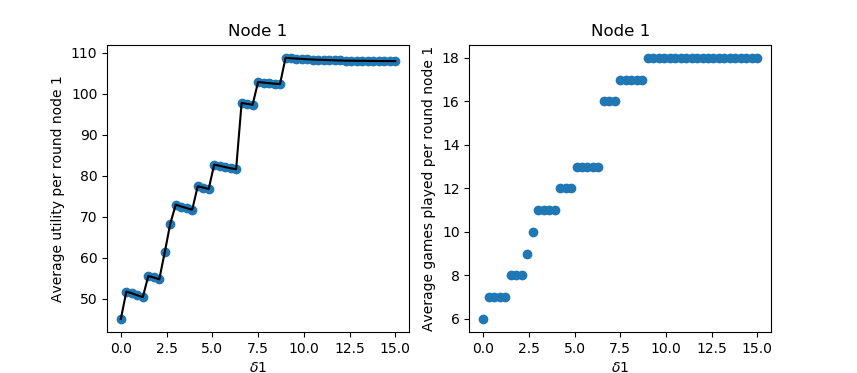}
  \caption{Utility for Vertex 1, Known $\delta$ }%\JG{In this figure, we are saying ``vertex''. Maybe change it to ``vertex'' as every other place.}} 
  \label{fig_zkc_n0_1}
 \end{figure}

Now we examine how players behave under these parameters when all adjust $\delta$ together between rounds. As noted, $u_i(\theta_{ij})$ is estimated numerically. Between rounds, it is sampled at a number of points in the interval $[0,30]$ then set to the one which maximizes $\mathbf{u}_i(\theta_{-i},\delta_i)$. Each curve in Figure \ref{fig_zkc_known_random} plots $\delta_i$ for an individual agent as it varies over time. Because each curve is an individual agent whose label is an arbitrary number, we have not included a legend. The same is true of the majority of the other figures in this section. The majority of players reach $\delta_{max}=30$ and stay there in order to be competitive in attracting partners, with only occasional decreases to $\delta=0$ when they are not competitive enough. {In this context, non-competitive means that the agents are not attracting a sufficiently high number of additional interactions to justify selecting $\delta > 0$.} There are generally between 1-3 agents selecting $\delta \approx 0$ in each of these rounds. The curves in Figure \ref{fig_zkc_known_random_conjoined} display the mean value of $\delta$ for all vertices of the same degree. For example, the yellow curve is the mean value of $\delta$ across all vertices of degree 4. Figure \ref{fig_zkc_known_random_conjoined} shows what is occurring: at any given time there are generally between 0 and 2 vertices of degree 2 (out of 11) with low values of $\delta$. The variance of which ones are not at $\delta=30$ may be due to best response dynamics between them, or to insufficient sample size when estimating utilities.

% \begin{figure}[ht]
%     \centering
%     \includegraphics[width=0.65\linewidth]{random_tie.PNG}
%     \caption{$\delta$ in karate club network with $k=2$ invitations per round with ties broken uniformly at random}
%     \label{fig_zkc_known_random}
% \end{figure}

% \begin{figure}[ht]
%     \centering
%     \includegraphics[width=0.65\linewidth]{random_tie_conjoined.PNG}
%     \caption{Mean $\delta$ by vertex degree in karate club network with $k=2$ invitations per round with ties broken uniformly at random}
%     \label{fig_zkc_known_random_conjoined}
% \end{figure}

\begin{figure}[ht]
\centering
    \includegraphics[width=.85\linewidth]{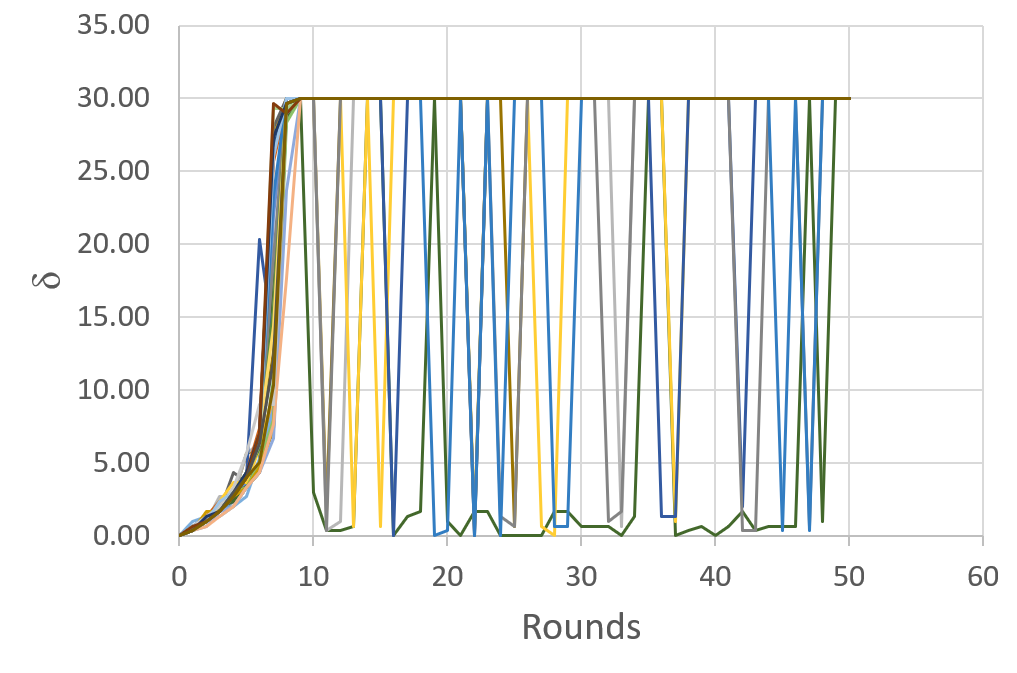}
    \captionsetup{width=.85\linewidth}
    \caption{$\delta$ in karate club network with $k=2$ invitations per round with ties broken uniformly at random}
    \label{fig_zkc_known_random}
\end{figure}
\begin{figure}[ht]
    \centering
    \includegraphics[width=.85\linewidth]{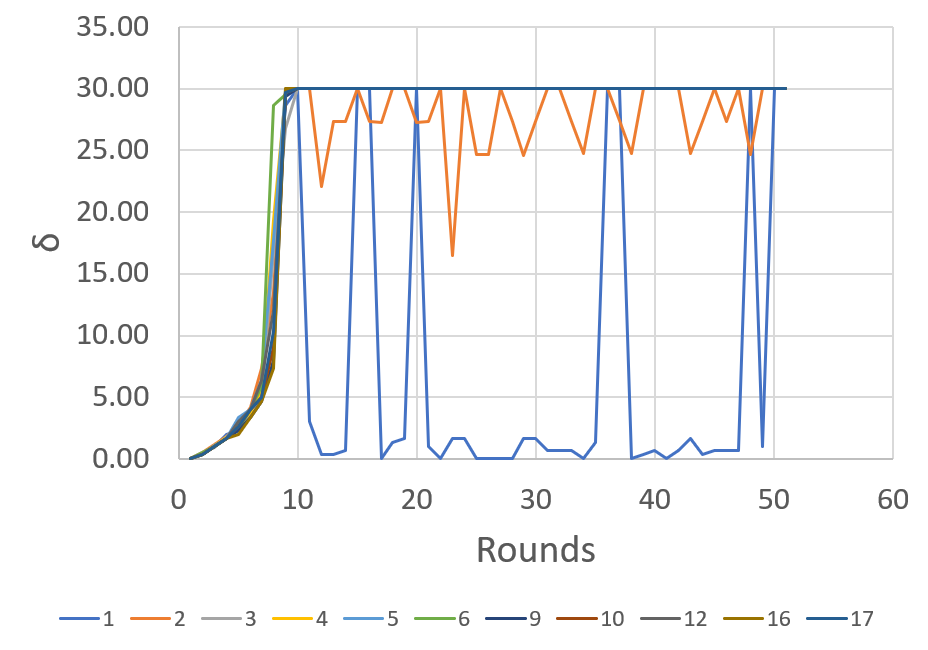}
    \captionsetup{width=.85\linewidth}
    \caption{Mean $\delta$ by vertex degree in karate club network with $k=2$ invitations per round with ties broken uniformly at random}
    \label{fig_zkc_known_random_conjoined}
\end{figure}

Vertex 12, the only vertex with degree 1, is easier to examine. Its only neighbor is vertex 1, which has degree 16. If all of vertex 1's other neighbors have $\delta=30$, then if $\delta_{12}=30$, vertex 12 can expect to be invited to $\frac{1}{8}$ games per round. Figure \ref{fig_zkc_known_random_conjoined} indicates that $u_{12}(\theta_{12,1},\delta_{12}=0) > u_{12}(\theta_{12,1},\delta_{12}=30) + \frac{1}{8}u_{12}(\theta_{1,12},\delta_{12}=30)$. The occasional jumps of $\delta_{12}=30$ can be attributed to the changes in the behavior of the vertices of degree 2. The points at which $\delta_{12}$ is low but not 0 are attributable to insufficient sample size in 12's estimates. 

We previously focused on breaking ties randomly when deciding whom to issue invitations. Now, we examine tie-breaking according to a lexicographic ordering. Figure \ref{fig_zkc_known_arbitrary} is analogous to Figure \ref{fig_zkc_known_random} with this change. It shows a strong, consistent, and repeating pattern in the values of player $\delta$s. When $\delta$ is low for most players, there is a general pattern of one-upmanship between players: each tries to slightly outdo its competitors, which progresses toward large jumps as the costs of increasing $\delta$ relative to a player's current $\delta$ value shrink. The result is a very clear S-curve. However, once many players reach $\delta=30$, the ``losers'' of the tie-breaking drop to $\delta=0$ starting a cascade of all players back to low $\delta$ values. The process then repeats. This echoes the findings of \cite{barclay2007}, in which players ``fake'' generosity to attract partners: once they have either failed to attract partners or their competitors have given up, each player returns to selfish behavior until competition again forces it to behave in a trustworthy manner. 
In contrast, random tie-breaking maintains a constant state of competition, preventing backsliding. Further, it is supported by \cite{fu2008} which finds that semi-frequent partner changes are necessary to motivate generous behavior, as otherwise partners become complacent and take advantage of each other.
% \begin{figure}[ht]
%     \centering
%     \includegraphics[width=.65\linewidth]{arbitrary_tie.PNG}
%     \caption{$\delta$ in karate club network with $k=2$ invitations per round with lexicographic tie-breaking.}
%     \label{fig_zkc_known_arbitrary}
% \end{figure}

\begin{figure}[ht]
\centering
    \centering
    \includegraphics[width=.85\linewidth]{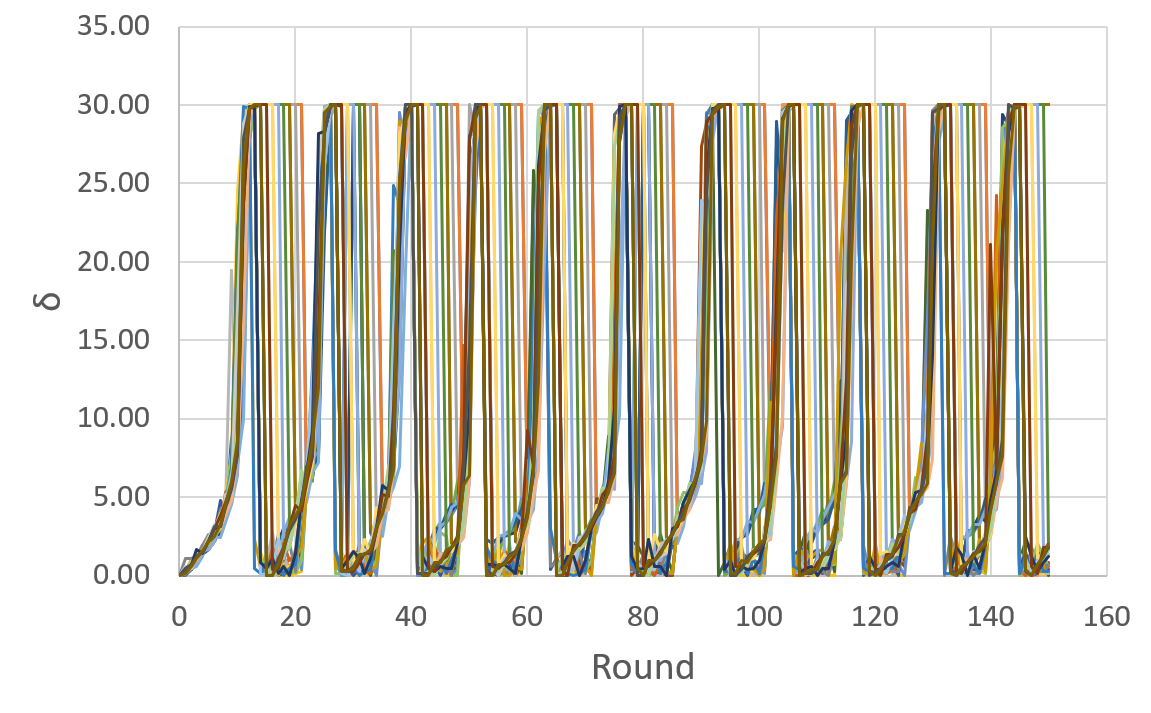}
    \captionsetup{width=.85\linewidth}
    \caption{$\delta$ in karate club network with $k=2$ invitations per round with lexicographic tie-breaking.}
    \label{fig_zkc_known_arbitrary}
\end{figure}
\begin{figure}[ht]
    \centering
    \includegraphics[width=.85\linewidth]{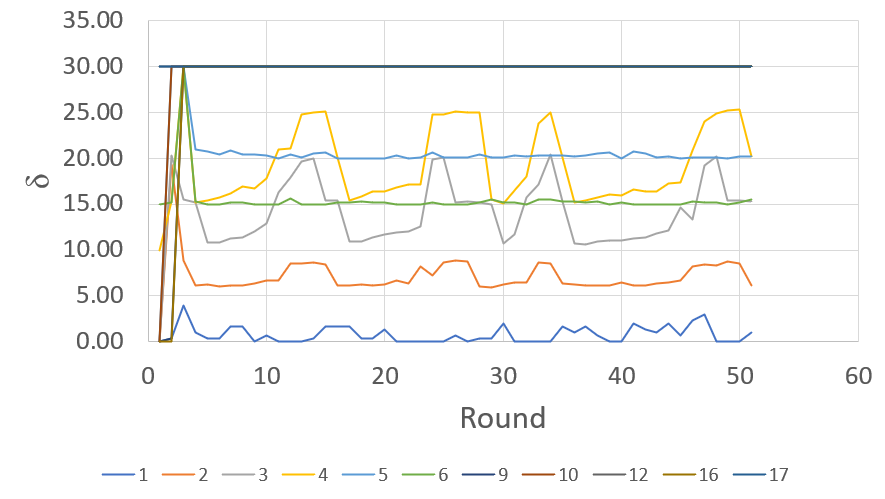}
    \captionsetup{width=.85\linewidth}
    \caption{Mean $\delta$ by vertex degree in karate club network with $k=2$ invitations per round with ties broken lexicographically, leaders constrained to $\delta=30$}
    \label{fig_zkc_leader}
\end{figure}
Interestingly, this behavior in which all agents cyclically return to low values of $\delta$ does not persist when certain agents are willing to take a ``leadership'' role and maintain a high value for $\delta$ at all times. In Figure \ref{fig_zkc_leader} we consider the same setting as in Figure \ref{fig_zkc_known_arbitrary}, except that we constrain $\delta_i=30$ for $i\in \{30,31,...,34\}$. As ties are broken lexicographically, these agents are losers in terms of tie-breaking, so their drop down to lower values of $\delta$ triggers a cascade of other agents doing the same. By keeping $\delta=30$ for these agents, many other agents never drop their $\delta$ values back to near 0. In addition to causing more stable $\delta$ values, curtailing the cycling also improves total utility in the network: the average utility per agent per round is 26.535 in the setting in Figure \ref{fig_zkc_known_arbitrary}, but improves to 26.884 in Figure \ref{fig_zkc_leader}.
A broader implication of this result is that careful subsidization of key agents can result in an overall increase in $\delta$ for the whole system. While we would like to devote more study to this fascinating topic in the future, the problem seed selection for trusting behavior, such as is considered in \cite{scata2016}, is beyond the scope of this paper.

% \begin{figure}[ht]
%     \centering
%     \includegraphics[width=0.65\linewidth]{zkc_arbitrary_leader.PNG}
%     \caption{Mean $\delta$ by vertex degree in karate club network with $k=2$ invitations per round with ties broken lexicographically, leaders constrained to $\delta=30$}
%     \label{fig_zkc_leader}
% \end{figure}

We also consider a second, larger social network where $N=333$. This network is a subset of the ego-Facebook network curated by SNAP. We let $k_i=3$ invitations for each agent $i$ and keep all other conditions identical to our previous setting, breaking ties uniformly at random. Figure \ref{fig_facebook_known} displays the similar behavior which occurred in the karate club network in Figure \ref{fig_zkc_known_random} but aggregated as a mean due to the large number of players. This contrasts with the repeated S-curves seen in Figure \ref{fig_zkc_known_arbitrary}. From the raw data, we found that most agents stayed at $\delta_{max} =30$ with a small number playing very small values of $\delta \approx 0$. This divide, with approximately $\frac{5}{6}$ agents using a high $\delta$ while $\frac{1}{6}$ use a low $\delta$, leads to a  stable average $\delta \approx 25$ which is similar to the $\frac{9}{10}$ to $\frac{1}{10}$ split seen in the karate club network. The exact differences are attributable to network structure and how many agents had limited opportunities. Also, as in Figure \ref{fig_zkc_known_random}, a handful of agents drop briefly to low $\delta$ values at times when they judge it too competitive before returning to $\delta=30$. This suggests that the pattern of an S-curve increase in $\delta$ followed by a plateau may be characteristic of naturally occurring social networks engaged in partner selection.

\begin{figure}
    \centering
    \includegraphics[width=.85\linewidth]{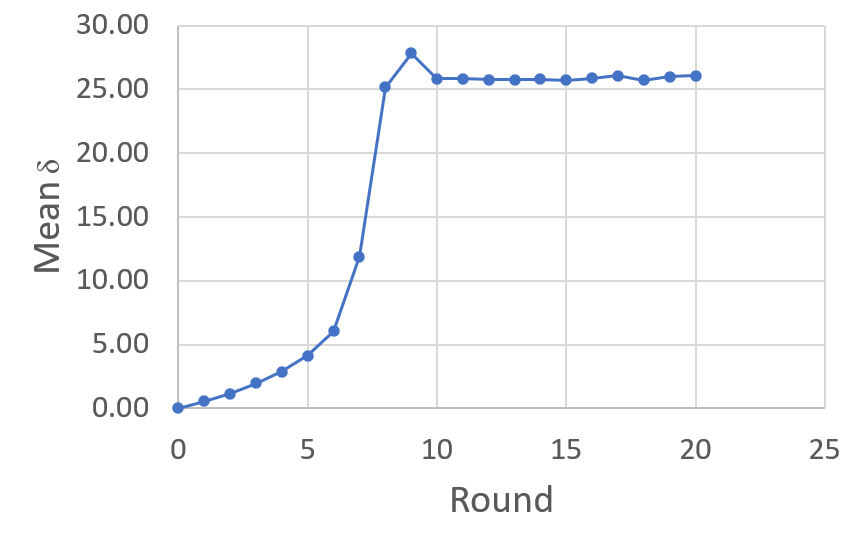}
    \caption{Mean $\delta$ in Facebook Ego Network}
    \label{fig_facebook_known}
\end{figure}

\subsection{Unknown $\delta$}

In Section \ref{subsec_network_known_delta} we considered networks in which each player $i$ knew the value of $\delta_{-i}$, and was immediately aware of any changes in it. Now we consider behavior when $\delta_{-i}$ is unknown, and agent $i$ estimates it using Algorithm 1. Ties are broken lexicographically in this setting. However, as games are drawn from a continuous distribution, the probability of a tie in perceived $\delta$ values is 0, so the results of random tie-breaking would be identical. We again consider the Karate Club Network and ego-Facebook Network. For easy comparison, all parameters will be identical to those in the previous subsection. The points in the curve in Figure \ref{fig_zkc_n0_2} are computed as the average of 1000 independent identically distributed games and provide a direct comparison to Figure \ref{fig_zkc_n0_1}. We see that the estimates in the unknown case approximate those in the known case. This suggests that as knowledge of $\delta_{-i}$ improves, players will approach the same behaviors they display when $\delta_{-i}$ is known to agent $i$. 
% \RN{Wondering if we need to say more -- this might appear simplistic. More when we discuss.}

\begin{figure}
  \centering
  \includegraphics[width=.85\linewidth]{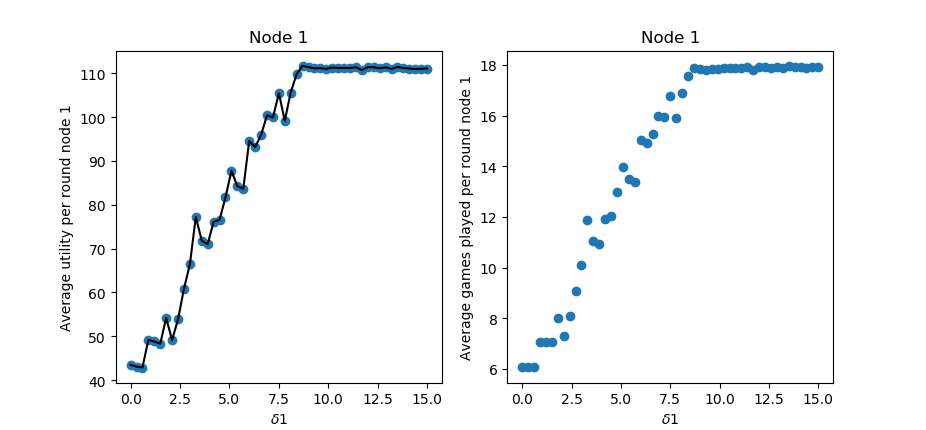}
  \caption{Utility for vertex 1, unknown $\delta_{-1}$}% \JG{Change vertex to vertex}}
  \label{fig_zkc_n0_2}
\end{figure}

We consider agents who update their values of $\delta$ in the two ways we discussed in Section \ref{sec_var_delta}: all agents either update their $\delta$ value probabilistically between rounds, or all agents work on an epoch system, updating their $\delta$ values every $t$ rounds after they and their neighbors have learned about each other. 

Figure \ref{fig_zkc_unknown_epoch} illustrates how delta shifts when players update according to an epoch system, every $t=100$ round. Agents move from $\delta= 0$ to $\delta=30$ in fewer value changes than in the known $\delta$ setting due to agents overestimating their competitors' trustworthiness and overcompensating in response. Since players only have estimates of each other's $\delta$, we see some gaming occur in Figure \ref{fig_zkc_known_arbitrary}: agents attempt to lower their $\delta$ value once they believe they've discouraged their competitors from raising $\delta$ values. Upon learning otherwise, they increase again. However, the time these agents spent with decreased $\delta$ values acts as a signal to their competitors, who now believe that they can lower their $\delta$ value in the same way the original agent did. This then acts as a new signal for the original agent leading to an oscillation between points in the range $[23,30]$ for $\delta$. However, agents' estimates are accurate enough that they rarely underestimate so badly as to drop below this range. 
% \RN{Revisit language.} \RN{Wondering if we need to say more -- this might appear simplistic. More when we discuss.}

% \begin{figure}
%     \centering
%     \includegraphics[width=.65\linewidth]{zkc_epochs}
%     \caption{$\delta$ in Karate Club Network, epoch = 100 rounds, $k=2$ invitations per round}
%     \label{fig_zkc_unknown_epoch}
% \end{figure}

\begin{figure}[ht]
\centering
    \centering
    \includegraphics[width=\linewidth]{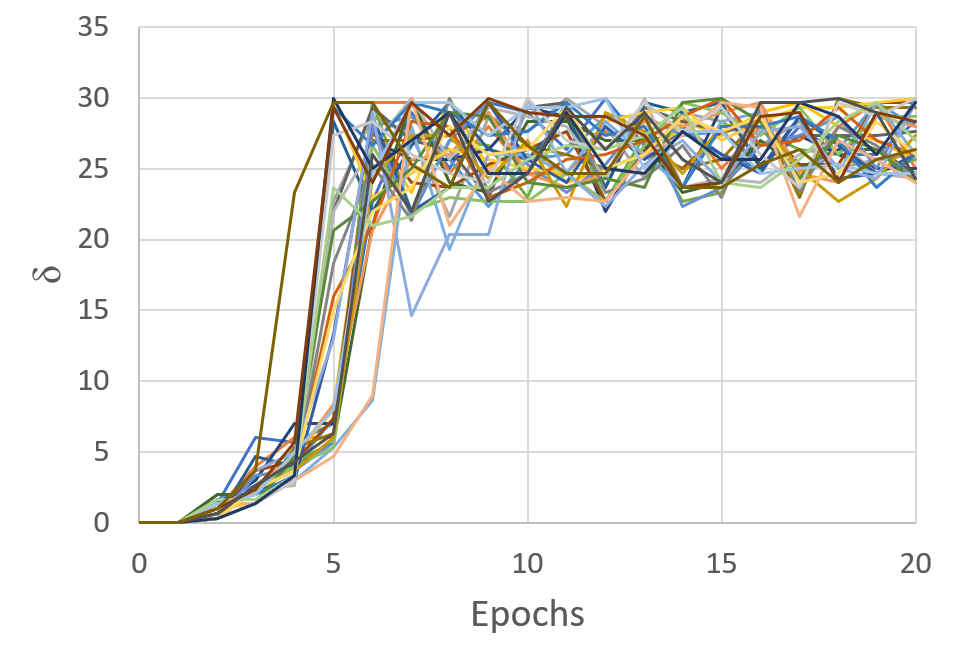}
    \captionsetup{width=.85\linewidth}
    \caption{$\delta$ in Karate Club Network, epoch = 100 rounds, $k=2$ invitations per round}
    \label{fig_zkc_unknown_epoch}
\end{figure}

\begin{figure}[ht]
    \centering
    \includegraphics[width=\linewidth]{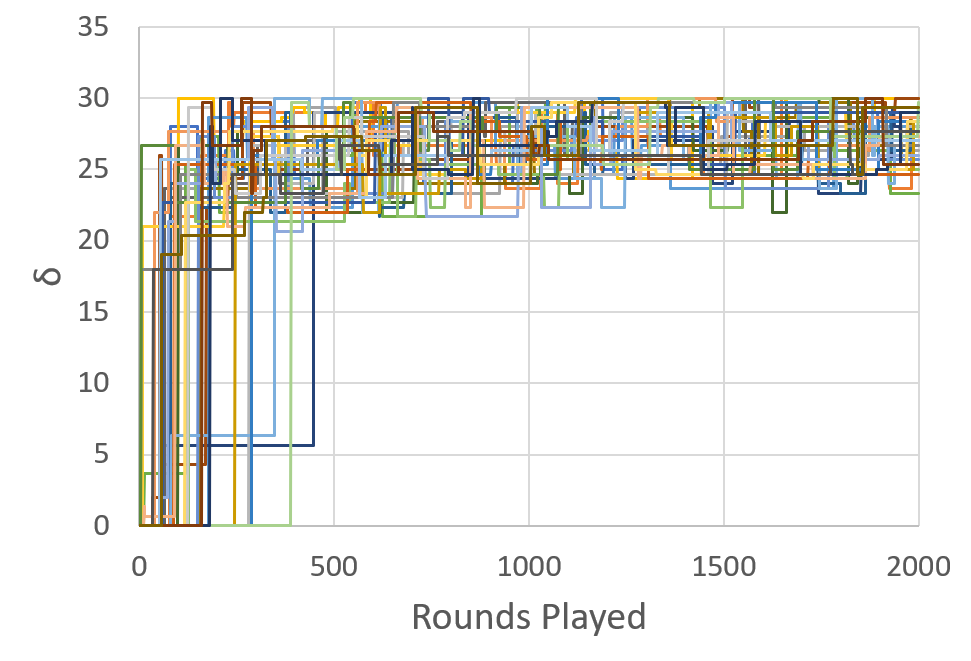}
    \captionsetup{width=.85\linewidth}
    \caption{$\delta$ in Karate Club Network with update probability $=\frac{1}{100}$, $k=2$ invitations per round}
    \label{fig_zkc_unknown_random}
\end{figure}

Figure \ref{fig_zkc_unknown_random} illustrates the same setting when players independently update their $\delta$ value with probability $\frac{1}{100}$ after each round. The behavior is very similar to that in the previous figure, again showing the attempts at gamesmanship where agents lower $\delta$ once they feel they have discouraged competition. They also oscillate within a similar range of $\delta$ values. It is worth noting that for both heuristic update schemes, lexicographic tie-breaking results in identical behavior to that pictured in Figures \ref{fig_zkc_unknown_epoch} and \ref{fig_zkc_unknown_random} rather than the sharply cyclical behavior seen in Figure \ref{fig_zkc_known_arbitrary}. This due to agents not having any ties to break, as the continuous payoff distribution leads to $\delta_{ji} \neq \delta_{li}$ for $l,j\in N_i^1$, even if $\delta_{j}=\delta_l$.

% \begin{figure}
%     \centering
%     \includegraphics[width=0.65\linewidth]{zkc_random}
%     \caption{$\delta$ in Karate Club Network with update probability $=\frac{1}{100}$, $k=2$ invitations per round}
%     \label{fig_zkc_unknown_random}
% \end{figure}

Finally, we consider the ego-Facebook network with $k_i=3$ invitations for each agent $i$. The mean $\delta$ value is displayed in Figure \ref{fig_facebook_unknown}, where players update under the epoch system. The comparison between Figures \ref{fig_facebook_known} and \ref{fig_facebook_unknown} appears identical to that of Figures \ref{fig_zkc_known_random} and \ref{fig_zkc_unknown_epoch}. As in Figure \ref{fig_facebook_known}, the network quickly settles to an average $\delta$ value of $\approx 25$. Player behavior is closer to that of Figure \ref{fig_zkc_unknown_epoch} than Figure \ref{fig_facebook_known}: the majority of players oscillating in the range [23,30], with some few consistently playing $\delta \approx 0$. The similarities are unsurprising as both are in the same setting of epoch updates with unknown $\delta$.

We see in both networks that competition between agents strongly pushes them to maintain increasing levels of trustworthiness; this stops only at the point they are no longer competitive. This increase in trustworthy behavior is healthy for the system as a whole both when $\delta$ is known and when it is not. {When $\delta=0$ for all players in the karate club network under the setting considered, the average utility per player per round is 22.792.} When $\delta$ is known, it is 27.578; when $\delta$ is unknown, and updates are on an epoch schedule, it is 27.540; and when $\delta$ is unknown and updates probabilistically it is 27.747, which represents increases of approximately 21.0\%, 20.8\%, and 21.7\%, respectively. Similarly, for the ego-Facebook network, {when $\delta=0$ the average utility per player per round is 30.905;} when $\delta$ is known, it is 37.750 and when $\delta$ is unknown and updates on an epoch schedule, it is 38.297 which represent increases of approximately 22.1\% and 23.9\%, respectively!
% \RN{Should we report (increased) utilities (percentage gain) as well?}
% \TM{Consider an influenced by neighbors evolutionary process, similar to the PRE paper. Don't do it numerically, but put in discussion (low priority)}

\begin{figure}
    \centering
    \includegraphics[width=0.9\linewidth]{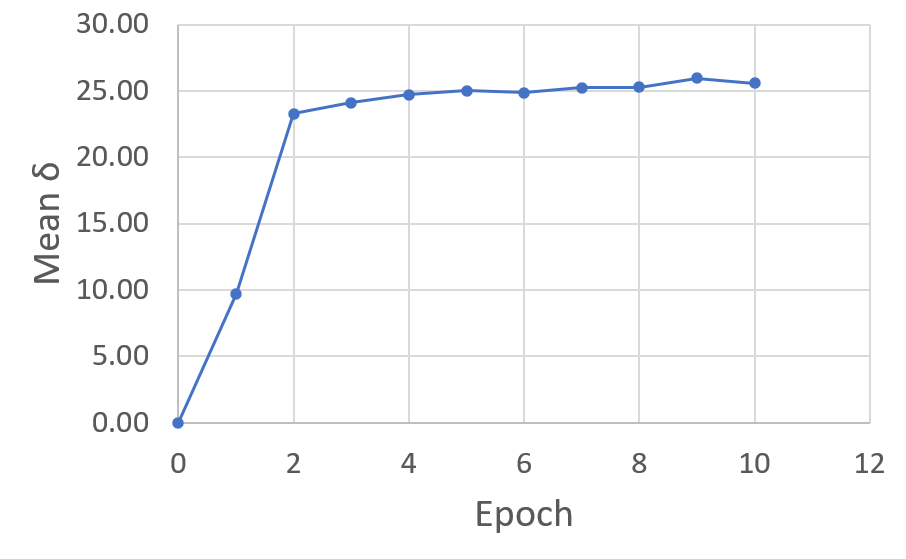}
    \caption{Mean $\delta$ in Facebook Ego Network, epoch = 100 rounds, $k=2$ invitations per round}
    \label{fig_facebook_unknown}
\end{figure}

\section{Discussion}\label{sec_discussion_sng}

In both social networks in the previous section, there was a naturally occurring increase in trustworthiness. However, we can construct ``artificial'' networks where this does not occur. Consider the 5-star graph in Figure \ref{fig_star}: the central vertex does not need to compete to receive invitations from the other vertices, as it is their only option. Conversely, for $k_i \leq 4$ for the central vertex $i$, the other vertices must compete to attract invitations and thus maintain a high level of trustworthiness. This competition vanishes for $k_i \geq 5$: each non-central vertex is invited, provided the expected utility for the central vertex is positive, causing all $\delta$ values to drop to 0. This dynamic is present in all networks: by restricting a resource and the number of invitations that may be issued, agents become more trustworthy as behaving otherwise causes a loss of access to the resource. Counter-intuitively, they become more selfish when the resource is abundant rather than less. This is apparent in the diad graph in Figure \ref{fig_diade}: each player knows that it is the other players' only option and need not compete for an invitation. However, a lack of competition does not mean that both do not stand to benefit from trustworthy behavior in the long term, merely that it is less beneficial for myopic agents in the short term. Traditional mechanisms for repeated games, such as grim trigger and discounted horizon analysis, have been shown to help players avoid such short-sighted behavior; we hope to incorporate these techniques into future versions of our system. 

\begin{figure}[ht]
    \centering
    \begin{minipage}{.5\linewidth}
      \centering
      \includegraphics[width=.45\linewidth]{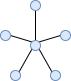}
      \caption{A 5-star graph}
      \label{fig_star}
    \end{minipage}%
    \begin{minipage}{.5\linewidth}
      \centering
      \includegraphics[width=.45\linewidth]{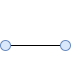}
      \caption{A diad graph}
      \label{fig_diade}
    \end{minipage}
\end{figure}

Nonetheless, as we noted previously, this behavior does not seem to exist in naturally occurring social networks. While there are ``singleton" vertices with only one neighbor who does not need to compete for their invitation, the desire of these neighbors to attract other invitations keeps the singleton vertices from being taken advantage of. We thus conjecture that this is why humans in social settings generally behave in a trustworthy manner even when they have the chance to take advantage of each other. It is only when the model is taken to extremes that we observe this behavior in social networks. 
% Figure \ref{fig_delta_highk} illustrates this phenomenon in Zachary's karate club network when we let $k_i=10$ for each agent $i$, with all other parameters the same as in Section \ref{sec_numerical_sng}. The figure displays the mean value of $\delta$ for all vertices of the same degree, and we see that while the vertices with fewer neighbors continue to compete for invitations, the vertices with more neighbors stop competing in order to take advantage of the invitations they receive, similar to the 5-star graph.

% \begin{figure}
%     \centering
%     \includegraphics{zkc_delta_highk.PNG}
%     \caption{Mean $\delta$ by vertex degree in karate club network when $k=10$ invitations per round}
%     \label{fig_delta_highk}
% \end{figure}

In addition to the numerical and theoretical results discussed in this paper, there are still areas to explore within the model. Empirically, the interactions we examined were strictly nonnegative between to avoid the computational cost. If this is changed, we expect to see different behavior so that the interactions are only nonnegative \textit{in expectation}. This is a reasonable avenue of exploration, as sometimes partnerships may not work out despite positive expectations. {Another area of exploration} non-uniform interactions between agents, for example if agent $i$ provides a better partnership than agent $j$. How will $\delta_i$ and $\delta_j$ change in response? In this case, we expect to see ``diva''-like behavior, with $i$ displaying a low $\delta_i$ and still attracting many more partners than $j$. 

% \TM{Following increased numerical results, look at more of the broader considerations which we are overlooking (things like biased decision making, influence of neighbors, etc.)}

\section{Conclusion and Future Directions}
\label{sec_conclusion_sng}

In this paper, we considered pairwise interactions between agents arranged in a social network. We modeled how agents behaved when competing with each other for interaction opportunities, using the limited-trust equilibrium of \cite{murray2021} to define player interactions. The agents in the network evolved to display behaviors that mirrored various empirical findings on human interaction, particularly \cite{barclay2007} and \cite{fu2008} from the field of evolutionary biology. This is particularly notable as agents within this model were not forward thinking as real humans are: each updated to display a level of trustworthiness which was a best response to that currently displayed by other players. Yet despite their myopia, agents are motivated to behave in a trustworthy manner without foresight or historically-based mechanisms such as Grim Trigger or Tit-for-Tat.

In addition to the empirical results of the model, a thorough mathematical analysis was presented. Simple learning algorithms were derived so that agents can learn about their neighbors through interactions. A process for agents to update their trustworthiness metric $\delta$ based solely on their two-hop neighborhood within the network was  presented, and a Nash equilibrium was shown to exist in the $\delta$-selection metagame that agents engage in. Along with the algorithms and processes which define the evolution of agent behavior over time, mathematical bounds were developed for the expected time for agents to learn each others' $\delta$s to a given level of precision. These bounds can be viewed in our e-companion \cite{murray2021trust} in Appendix B. %Appendix \ref{app_time_bounds}.

% \subsection{Future Directions}
% \label{subsec_future}

The model and experiments performed present many directions for future research. While agents need not utilize mechanisms such as Grim Trigger or Tit-for-Tat in the social networks we considered, we also saw that limited-trust is not enough to encourage trustworthy behavior in the networks in Figures \ref{fig_star} and \ref{fig_diade}. Therefore, {it will be interesting to incorporate} these concepts into the social network model proposed here. We believe that doing so will allow our model to more accurately represent interactions in smaller communities which run the risk of exhibiting the behaviors discussed, such as diads.

We are also interested in incorporating historical data into the decision-making process. Individuals who interact frequently are more likely to have an established relationship. They are thus less likely to switch partners if the utility increase is minor. Another way to incorporate historical data is  based on past utility earned. Agents who do well are able to increase the number of interactions they can initiate per round. Changes to network structure based on past behavior is a related topic: edges may wither if unused, or new edges may appear from an agent $i$ to an agent $j$ if the two have a mutual neighbor $l\neq i,j$ whom both interact with frequently. 

{Another potential topic for study is how agents behave when they may only accept a limited number of invitations. This brings with it a host of considerations: are invitations received in order, do they need to be accepted or declined sequentially, and do some agents decide who to invite only after receiving their invitations?}

One final area of interest is how agents behave with unknown network structures. Consider the case in which an agent $i$ is aware of its one-hop neighborhood $N_1^i$ but is uncertain of its two-hop neighborhood when updating $\delta_i$. We wish to develop methods for $i$ to estimate the structure of $N_2^i$ as this will remove agent $i$'s need to know $\delta'_j$ of its neighbor $j$.

We hope the wider research community is as excited by this work as we are.
To make it easier for interested researchers to explore this setting, we have made a portion of our code publicly available on GitHub at \cite{murray2020}. This repository is written in Python3 and contains the class file, used to generate a network and track player behavior, as well as demo scripts to show its use. We hope that making it available will enable further exploration by interested researchers.

% Appendix here
% Options are (1) APPENDIX (with or without general title) or
%             (2) APPENDICES (if it has more than one unrelated sections)
% Outcomment the appropriate case if necessary
%
% \begin{APPENDIX}{<Title of the Appendix>}
% \end{APPENDIX}
%
%   or
%

% or


% \appendices

% % use section* for acknowledgment
% \ifCLASSOPTIONcompsoc
%   % The Computer Society usually uses the plural form
  \section*{Acknowledgments}
% \else
% \fi

{The work of J. Garg was supported by NSF Grant 1942321. T. Murray and R. Nagi were supported in part by ONR through the Program Management of Drs. D. Wagner and W. Adams under Award N000014-16-1-2245.}

\section*{Declaration of Interest}
{No interests to report.}

% \ifCLASSOPTIONcaptionsoff
%   \newpage
% \fi



%\bibliographystyle{IEEEtran}
\bibliographystyle{apalike}
\bibliography{references}

\begin{thebibliography}{}

\bibitem[Abramson and Kuperman, 2001]{abramson2001}
Abramson, G. and Kuperman, M. (2001).
\newblock Social games in a social network.
\newblock {\em Phys. Rev. E}, 63:030901.

\bibitem[Aral and Walker, 2014]{aral2014}
Aral, S. and Walker, D. (2014).
\newblock Tie strength, embeddedness, and social influence: A large-scale networked experiment.
\newblock {\em Management Science}, 60:1352--1370.

\bibitem[Barclay, 2004]{barclay2004}
Barclay, P. (2004).
\newblock Trustworthiness and competitive altruism can also solve the “tragedy of the commons”.
\newblock {\em Evolution and Human Behavior}, 25(4):209 -- 220.

\bibitem[Barclay, 2013]{barclay2013}
Barclay, P. (2013).
\newblock Strategies for cooperation in biological markets, especially for humans.
\newblock {\em Evolution and Human Behavior}, 34(3):164 -- 175.

\bibitem[Barclay, 2016]{barclay2016}
Barclay, P. (2016).
\newblock Biological markets and the effects of partner choice on cooperation and friendship.
\newblock {\em Current Opinion in Psychology}, 7:33 -- 38.
\newblock Evolutionary psychology.

\bibitem[Barclay and Willer, 2007]{barclay2007}
Barclay, P. and Willer, R. (2007).
\newblock Partner choice creates competitive altruism in humans.
\newblock {\em Proceedings. Biological sciences / The Royal Society}, 274:749--53.

\bibitem[Bolouki et~al., 2018]{bolouki2018}
Bolouki, S., Nedi{\'{c}}, A., and Ba{\c{s}}ar, T. (2018).
\newblock {\em Social Networks}, pages 907--949.
\newblock Springer International Publishing, Cham.

\bibitem[Chen et~al., 2014]{chen2014}
Chen, P.-A., Keijzer, B.~D., Kempe, D., and Sch\"{a}fer, G. (2014).
\newblock Altruism and its impact on the price of anarchy.
\newblock {\em ACM Trans. Econ. Comput.}, 2(4):17:1--17:45.

\bibitem[Coleman et~al., 1964]{coleman1964}
Coleman, J.~S. et~al. (1964).
\newblock Introduction to mathematical sociology.
\newblock {\em Introduction to mathematical sociology.}

\bibitem[Dall{\textquoteright}Asta et~al., 2012]{dall2012}
Dall{\textquoteright}Asta, L., Marsili, M., and Pin, P. (2012).
\newblock Collaboration in social networks.
\newblock {\em Proceedings of the National Academy of Sciences}, 109(12):4395--4400.

\bibitem[Davis et~al., 2009]{davis2009}
Davis, A., Gardner, B.~B., and Gardner, M.~R. (2009).
\newblock {\em Deep South: A social anthropological study of caste and class}.
\newblock Univ of South Carolina Press.

\bibitem[Debove et~al., 2015]{debove2015}
Debove, S., André, J.-B., and Baumard, N. (2015).
\newblock Partner choice creates fairness in humans.
\newblock {\em Proceedings. Biological sciences / The Royal Society}, 282.

\bibitem[Eisenbruch et~al., 2019]{eisenbruch2019}
Eisenbruch, A., Grillot, R., and Roney, J. (2019).
\newblock Why be generous? tests of the partner choice and threat premium models of resource division.
\newblock {\em Adaptive Human Behavior and Physiology}, 5.

\bibitem[Freeman et~al., 1998]{freeman1998}
Freeman, L.~C., Webster, C.~M., and Kirke, D.~M. (1998).
\newblock Exploring social structure using dynamic three-dimensional color images.
\newblock {\em Social networks}, 20(2):109--118.

\bibitem[Fu et~al., 2008]{fu2008}
Fu, F., Hauert, C., Nowak, M.~A., and Wang, L. (2008).
\newblock Reputation-based partner choice promotes cooperation in social networks.
\newblock {\em Phys. Rev. E}, 78:026117.

\bibitem[Hamilton, 1963]{Hamilton1963}
Hamilton, W.~D. (1963).
\newblock The evolution of altruistic behavior.
\newblock {\em The American Naturalist}, 97(896):354--356.

\bibitem[Hanaki et~al., 2007]{hanaki2007}
Hanaki, N., Peterhansl, A., Dodds, P., and Watts, D. (2007).
\newblock Cooperation in evolving social networks.
\newblock {\em Management Science}, 53:1036--1050.

\bibitem[Jackson and Zenou, 2015]{jackson2015}
Jackson, M.~O. and Zenou, Y. (2015).
\newblock Chapter 3 - games on networks.
\newblock In Young, H.~P. and Zamir, S., editors, {\em Handbook of Game Theory with Economic Applications}, volume~4, pages 95 -- 163. Elsevier.

\bibitem[Jain et~al., 2020]{jain2020}
Jain, L., Katarya, R., and Sachdeva, S. (2020).
\newblock Recognition of opinion leaders coalitions in online social network using game theory.
\newblock {\em Knowledge-Based Systems}, 203:106158.

\bibitem[Kreps et~al., 1982]{kreppsmilgrom1982}
Kreps, D.~M., Milgrom, P., Roberts, J., and Wilson, R. (1982).
\newblock Rational cooperation in the finitely repeated prisoners' dilemma.
\newblock {\em Journal of Economic Theory}, 27(2):245--252.

\bibitem[Kreps and Wilson, 1982]{kreps1982}
Kreps, D.~M. and Wilson, R. (1982).
\newblock Reputation and imperfect information.
\newblock {\em Journal of Economic Theory}, 27(2):253--279.

\bibitem[Kunegis, 2013]{konect}
Kunegis, J. (2013).
\newblock {KONECT} -- {The} {Koblenz} {Network} {Collection}.
\newblock In {\em Proc. Int. Conf. on World Wide Web Companion}, pages 1343--1350.

\bibitem[Ledyard, 1994]{ledyard1994}
Ledyard, J.~O. (1994).
\newblock {Public Goods: A Survey of Experimental Research}.
\newblock Public Economics 9405003, University Library of Munich, Germany.

\bibitem[Leskovec and Krevl, 2014]{snapnets}
Leskovec, J. and Krevl, A. (2014).
\newblock {SNAP Datasets}: {Stanford Large Network Dataset Collection}.
\newblock \url{http://snap.stanford.edu/data}.

\bibitem[Mui et~al., 2002]{mui2002}
Mui, L., Mohtashemi, M., and Halberstadt, A. (2002).
\newblock A computational model of trust and reputation.
\newblock In {\em Proceedings of the 35th annual Hawaii international conference on system sciences}, pages 2431--2439. IEEE.

\bibitem[Murray, 2020]{murray2020}
Murray, T. (2020).
\newblock {Social Network Games}.
\newblock github.com.
\newblock \url{https://github.com/kzr-soze/SocialNetworkGames}.

\bibitem[Murray et~al., 2021]{murray2021}
Murray, T., Garg, J., and Nagi, R. (2021).
\newblock Limited-trust equilibria.
\newblock {\em European Journal of Operational Research}, 289(1):364 -- 380.

\bibitem[{Murray} et~al., 2023]{murray2021trust}
{Murray}, T., {Garg}, J., and {Nagi}, R. (2023).
\newblock Trust in social network games.
\newblock ArXiv.org.
\newblock \url{https://arxiv.org/abs/2103.01460}.

\bibitem[Murray, 2021]{murray_dissertation}
Murray, T.~S. (2021).
\newblock {\em Modeling trustworthy behavior and limiting the impact of selfishness}.
\newblock PhD thesis, University of Illinois at Urbana-Champaign.

\bibitem[{Naghizadeh} and {Liu}, 2018]{naghizadeh2018}
{Naghizadeh}, P. and {Liu}, M. (2018).
\newblock Provision of public goods on networks: On existence, uniqueness, and centralities.
\newblock {\em IEEE Transactions on Network Science and Engineering}, 5(3):225--236.

\bibitem[Nash, 1950]{nash1950}
Nash, J.~F. (1950).
\newblock Equilibrium points in n-person games.
\newblock {\em Proceedings of the National Academy of Sciences}, 36(1):48--49.

\bibitem[Ozkan-Canbolat and Beraha, 2016]{ozkancanbolat2016}
Ozkan-Canbolat, E. and Beraha, A. (2016).
\newblock Evolutionary knowledge games in social networks.
\newblock {\em Journal of Business Research}, 69(5):1807 -- 1811.
\newblock Designing implementable innovative realities.

\bibitem[Read, 1954]{kenneth54}
Read, K.~E. (1954).
\newblock Cultures of the {Central} {Highlands}, {New} {Guinea}.
\newblock {\em Southwestern J. of Anthropol.}, 10(1):1--43.

\bibitem[Scatà et~al., 2016]{scata2016}
Scatà, M., {Di Stefano}, A., {La Corte}, A., Liò, P., Catania, E., Guardo, E., and Pagano, S. (2016).
\newblock Combining evolutionary game theory and network theory to analyze human cooperation patterns.
\newblock {\em Chaos, Solitons \& Fractals}, 91:17 -- 24.

\bibitem[Schwimmer, 1973]{schwimmer1973}
Schwimmer, E.~G. (1973).
\newblock {\em Exchange in the social structure of the Orokaiva: traditional and emergent ideologies in the Northern District of Papua}.
\newblock C. Hurst.

\bibitem[{Song} and {van der Schaar}, 2019]{song2019}
{Song}, Y. and {van der Schaar}, M. (2019).
\newblock Repeated network games with dominant actions and individual rationality.
\newblock {\em IEEE Transactions on Network Science and Engineering}, 6(4):812--823.

\bibitem[Sylwester and Roberts, 2010]{sylwester2010}
Sylwester, K. and Roberts, G. (2010).
\newblock Cooperators benefit through reputation-based partner choice in economic games.
\newblock {\em Biology letters}, 6:659--62.

\bibitem[Sylwester and Roberts, 2013]{sylwester2013}
Sylwester, K. and Roberts, G. (2013).
\newblock Reputation-based partner choice is an effective alternative to indirect reciprocity in solving social dilemmas.
\newblock {\em Evolution and Human Behavior}, 34(3):201 -- 206.

\bibitem[Szabó and Fáth, 2007]{szabo2007}
Szabó, G. and Fáth, G. (2007).
\newblock Evolutionary games on graphs.
\newblock {\em Physics Reports}, 446(4):97 -- 216.

\bibitem[Zachary, 1977]{zachary1977}
Zachary, W.~W. (1977).
\newblock An information flow model for conflict and fission in small groups.
\newblock {\em Journal of Anthropological Research}, 33(4):452--473.

\end{thebibliography}

\fi

\onecolumngrid
\newpage

%%%%%%%%%%%%%%%%%%%%%%%%%%%%%%%%%%%%%%%%%%%%%%%%%%%%%%%%%%%%%%%%%%%%%%%%%%%%%%%%
% EVERYTHING FROM THIS POINT ON IS AN APPENDIX APPEARS ONLY IN THE E-COMPANION %
%%%%%%%%%%%%%%%%%%%%%%%%%%%%%%%%%%%%%%%%%%%%%%%%%%%%%%%%%%%%%%%%%%%%%%%%%%%%%%%%
\iftrue
\begin{center}
    \Huge{Limited-Trust in Social Network Games \\E-Companion}
\end{center}

\appendix
\section{Additional Notes: Rate of learning $\delta_{-i}$}\label{app_time_bounds}

In Section III we looked at how a player $i$ can learn $\delta_j$ through interactions with player $j$ as both a leader and a follower. In this appendix we provide bounds on how quickly $i$ can learn a fixed $\delta_j$. We focus on doing so from $i$'s perspective as a leader, rather than as a follower: although we saw in Section III-B that information about $\delta_j$ can be inferred when $j$ is the leader, $i$ cannot guarantee that $j$ will ever invite it to interact. Further, the likelihood of gaining this information changes with $\delta_{ij}$, player $j$'s estimate of $\delta_i$. We will assume that $j$ believes that $u_j(\theta_{ij}) \geq 0$, as otherwise it will never accept an invitation from $i$ to interact. We now consider how many games are needed for $i$ to determine $\delta_j$ to within an error of $\varepsilon >0$.

First, we note that in order to guarantee that $|\delta_j - \delta_{ji}| \leq \varepsilon$, where $\delta_{ji}$ is player $i$'s estimate of $\delta_j$, both of the following must be true:
\begin{enumerate}
    \item $\delta_j-\delta_{ji} \leq \varepsilon$ 
    \item $\delta_{ji}-\delta_j \leq \varepsilon$
\end{enumerate}
While trivial, this implies that if player $i$ can determine an interval $[\delta^l_{ji},\delta^u_{ji})$ such that $\delta_j\in [\delta^l_{ji},\delta^u_{ji})$ and $\delta^u_{ji}-\delta^l_{ji}<2\varepsilon$, 
then $\delta_j$ must be within $\varepsilon$ of at least one of the lower bound $\delta^l_{ji}$ or the upper bound $\delta^u_{ji}$. Consider the expected number of observations required for player $i$ {to determine} both an upper and lower bound within $\varepsilon$ of $\delta_{j}$. As this is a stricter condition than is required to estimate $\delta_j$ to within an error of $\varepsilon$, determining it serves as an upper bound on the expected number of observations to estimate $\delta_j$ to within $\varepsilon$.

Suppose that player 1 and player 2 are participating in a randomly generated $m\times n$ leader-follower game generated via $\mathcal{A}_{12},\mathcal{B}_{12}$ in which 
all entries of each matrix are generated iid to all other entries in the matrix.
% $a_{ij}\sim a_{kl}$ and $b_{ij} \sim b_{kl}$ and all $a_{ij},a_{kl},b_{ij},b_{kl}$ are independent. 
Without loss of generality, we will assume that player 1 chooses to play $s_i$ as the leader. We want to determine the probability that the follower player 2 will choose a strategy $s_j$ which reveals a lower bound $\delta_{21}^l$ such that $\delta_2-\delta_{21}^l \leq \varepsilon$. Note that a leader-follower game after the leader has chosen its strategy is equivalent to a $1 \times n$ game, as in Table IV in Section III-A. 
In order for player 1 to observe a lower bound $\delta^l_{21} > 0$, it is necessary for player 2 to not pick the $s_j$ in Table IV that maximizes its utility. Without loss of generality, assume $b_1 \geq b_j$ $ \forall i\in[n]$, where $[n] = \{1,2,...,n\}$. If player 2 instead chooses to play $s_j$, then player 1 can deduce that $\delta_2 \geq b_1 - b_j$ and set $\delta^l_{21} = b_1-b_j$ using Algorithm 1. Based on the fact that all such values are randomly generated from known distributions $\mathcal{A}_{12},\mathcal{B}_{12}$, we can compute the probability of finding an acceptable lower bound $\delta_{21}^l$ by taking the cumulative distribution functions of the distributions over the relevant areas in $\mathcal{R}^2$.

\textbf{Lemma 2.}
\textit{For the game in Table IV, the probability of the game revealing $\delta^l_{21}$ such that $\delta_2-\delta^l_{21} \leq \varepsilon$ is 
\begin{equation*}
\resizebox{\linewidth}{!}{%
$P^{l} (\varepsilon) = n(n-1)\int_{-\infty}^{\infty}{\left(
\int_{-\infty}^{\infty}{\left(
\int_{b_1-\delta_2}^{b_1-\delta_2+\varepsilon}{\left(
\int_{a_1+b_1-b_2}^{\infty}{
P(E|s_1,s_2)f_a(a_2)da_2
}\right)f_b(b_2)db_2
}\right)f_a(a_1)da_1
}\right)f_b(b_1)db_1
}$%
}
\end{equation*}
where  
% \begin{align*}
\begin{equation*}
\resizebox{\linewidth}{!}{%
$P(E|s_1,s_2) = \left(\int_{b_1-\delta}^{b_1}{\left(
\int_{\infty}^{a_2+b_2-b_j}{
f(a_j)da_j}
\right)f(b_j)db_j} +\int_{-\infty}^{\infty}{\left(
\int_{-\infty}^{b_1-\delta_2}{
f_b(b_j)db_j
}\right)f_a(a_j)da_j
}\right)^{n-2}.$}
\end{equation*}}
% \end{lemma}

\begin{proof}
We begin by considering the probability of an event $E$ occurring where $E$ is the event that $b_1$ is the greedy choice for player 2, $s_2$ is chosen, and this results in a lower bound $\delta_{21}^l$ within $\varepsilon$ of the true $\delta_2$. In order for this to occur, the relationship between $s_1$ and $s_2$ must be as in Figure \ref{fig_lb1}, where both $s_2$ gives better net utility than $s_1$ and $b_2$ is within $\varepsilon$ of $\delta_2$ but not greater than it. The diagonal line which passes through $s_1$ is the set of all points which provide equal net utility to $s_1$: points above the line provide better net utility and points below it provide worse net utility. Additionally, given $s_1$ and $s_2$ it is necessary for another strategy $s_j$ where $j\neq 1,2$ to not both give better net utility than $s_2$ and have $b_1 - b_j \leq \delta_2$. If both of these conditions are met, $s_j$ will be picked instead of $s_2$. Therefore, for a given $s_1,s_2$ that fulfill the necessary conditions, $s_j$ must lie in the shaded region of Figure \ref{fig_lb2} for event $E$ to occur. Thus if $E_j$ is the event that $s_j$ is in an acceptable position
\begin{equation*}
\resizebox{\linewidth}{!}{%
$P(E_j|s_1,s_2) = \int_{b_1-\delta}^{b_1}{\left(
\int_{\infty}^{a_2+b_2-b_j}{
f(a_j)da_j}
\right)f(b_j)db_j} + \int_{-\infty}^{\infty}{\left(
\int_{-\infty}^{b_1-\delta_2}{
f_b(b_j)db_j
}\right)f_a(a_j)da_j
}.$}
\end{equation*}
This must be true for every $s_j$ with $j\neq 1,2$, so $P(E|s_1,s_2)=P(E_j|s_1,s_2)^{n-2}$. 
Next we note that for event $E$ to occur for a given $s_1$, we require that $s_2$ lie in the shaded area in Figure \ref{fig_lb1} which occurs with probability 
\begin{equation*}
P(E|s_1) = \int_{b_1-\delta_2}^{b_1-\delta_2+\varepsilon}{\left(
\int_{a_1+b_1-b_2}^{\infty}{
P(E|s_1,s_2)f_a(a_2)da_2
}\right)f_b(b_2)db_2
}.
\end{equation*}
This allows us to integrate over all values of $s_1$ to get that 
$$
P(E) = \int_{-\infty}^{\infty}{\left(
\int_{-\infty}^{\infty}{
P(E|s_1)f_a(a_1)da_1
}\right)f_b(b_1)db_1
}.
$$
Given that any $s_i$ could be the greedy response and any $s_j$ for $j\neq i$ could provide the bound upper bound with equal probability, but each of these events is mutually exclusive, we finally get that $P^l(\varepsilon)=n(n-1)P(E)$, which completes the proof.
% \Halmos
\end{proof}

\begin{figure}[ht]
\centering
\begin{minipage}{.45\linewidth}
  \centering
  \includegraphics[width=.8\linewidth]{Prob_delta_learning_1}
  \caption[16]{Area for $s_2$ given $s_1$ to establish lower bound}
  \label{fig_lb1}
\end{minipage}%
\begin{minipage}{.45\linewidth}
  \centering
  \includegraphics[width=.8\linewidth]{Prob_delta_learning_2.png}
  \caption{Area for $s_i$ given $s_1,s_2$ to establish lower bound}
  \label{fig_lb2}
\end{minipage}
\end{figure}

The role of $n$ as indicated by Lemma 2 is somewhat counterintuitive. The Lemma shows that the probability of the first player establishing an acceptable lower bound $\delta^l_{21}$ on $\delta_2$ goes to 0 as the second player's number of strategies $n\rightarrow \infty$: if the second player is presented with more choices, shouldn't they have a higher probability of getting one which allows them to precisely play near $\delta_2$ and maximize utility? This may be the case in smaller values of $n$, but it does not occur in general. 

Suppose $f_a=f_b=U[0,1]$, the uniform distribution. For high values of $n$, when the first player chooses strategy $s_i$ the second player will have a greedy best response $s_j = G_2(s_i)$ with $b_{ij}\approx 1$ with high probability, as for a high number of independent samples of $U[0,1]$, the expected value of the maximum sample goes to 1. Also due to the high value of $n$, there will be another response $s_l$ with $b_{il} < b_{ij}$ but $b_{il}\approx a_{il} \approx 1$ with high probability due to the same independent sampling. Consider this game with $n=1000$ and suppose the first player selects strategy $s_i$. The second player's strategy set now resembles Figure \ref{fig_random_1m_game}, which is a randomly generated row of the $m\times 1000$ game matrix. In the figure, a $\delta_2$ of approximately $0.02$  appears to be sufficient for the socially optimal strategy near the top right corner to be played. This means that if, for example, $\delta_2=0.2$ it will be nearly impossible to ever establish an acceptable lower bound within $\varepsilon=0.05$ of $\delta_2$; it is similarly difficult to establish any upper bound at all.

\begin{figure}
    \centering
    \includegraphics[width=0.5\linewidth]{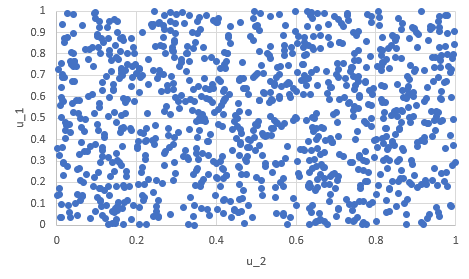}
    \caption{Example game from Table IV, $n=1000$ $f_a=f_b=U[0,1]$}
    \label{fig_random_1m_game}
\end{figure}

Next we derive the probability player 1 observes an upper bound $\delta_{21}^u < \infty$ on $\delta_2$ in the game in Table IV. For this to occur, player 2 must choose $s_i$ over $s_j$ because $b_i > b_j$ despite the fact that $a_i+b_i < a_j + b_j$. If, without loss of generality, $s_1$ is the greedy best response for player 2, this allows player 1 to conclude that $\delta_2 < b_1-b_j$ and set $\delta_{21}^u = b_1-b_j$. As we have already seen, it would also allow player 1 to set $\delta_{21}^l = b_1 - b_j$, thus showing that both upper and lower bounds can be observed within a single game. 

\textbf{Lemma 3.}\textit{
For the game in Table IV, the probability of the game revealing $\delta_{21}^u$ such that $\delta_{21}^l - \delta_2 \leq \varepsilon$ for $\varepsilon > 0$ is 
\begin{equation*}
\resizebox{\linewidth}{!}{$
P^u(\varepsilon) = n\int_{-\infty}^{\infty}{\left(
\int_{-\infty}^{\infty}{
\left( P(E\cap E_1|s_1) + (n-1)P(E\cap E_2|s_1)
\right)}f_a(a_1)da_1
\right)f_b(b_1)db_1}
$}
\end{equation*}
where
\begin{equation*}
% \resizebox{\linewidth}{!}{
P(E\cap E_1|s_1) = P(s_i\in A_1|s_1)^{n-1} - (P(s_i\in A_1|s_1)-P(s_i\in A_2|s_1))^{n-1},
\end{equation*}
% \begin{equation*}
% \resizebox{\linewidth}{!}{
\begin{equation*}
\resizebox{\linewidth}{!}{$
P(s_i\in A_1|s_1) = 
\int_{-\infty}^{b_1-\delta_2}{\left(
\int_{-\infty}^{\infty}{
f_a(a_i)da_i
}
\right)f_b(b_i)db_i}+
\int_{b_1-\delta_2}^{b_1}{\left(
\int_{-\infty}^{a_1+b_1-b_i}{
f_a(a_i)da_i
}
\right)f_b(b_i)db_i,}
$}
\end{equation*}
% }
% \end{equation*}
\begin{equation*}
P(s_i\in A_2|s_1) = \int_{b_1-\delta_2-\varepsilon}^{b_1-\delta_2}{\left(
\int_{a_1+b_1-b_i}^{\infty}{
f_a(a_i)da_i
}\right)f_b(b_i)db_i,
}
\end{equation*}
and
\begin{equation*}
\resizebox{\linewidth}{!}{$
P(E\cap E_2|s_1) = \int_{b_1-\delta_2}^{b_1}{\left(
\int_{a_1+b_1-b_2}^{\infty}{
\left(P(A_3|s_1,s_2)^{n-2}
-(P(A_3|s_1,s_2)-P(A_4|s_1,s_2))^{n-2}\right)f_a(a_2)da_2
}
\right)f_b(b_2)db_2}
$}
\end{equation*}
\begin{equation*}
\resizebox{\linewidth}{!}{
$P(s_i \in A_3|s_1,s_2) = \int_{b_1-\delta}^{b_1}{\left(
\int_{\infty}^{a_2+b_2-b_i}{
f(a_i)da_i}
\right)fb_idb_i}+\int_{-\infty}^{\infty}{\left(
\int_{-\infty}^{b_1-\delta_2}{
f_b(b_i)db_i
}\right)f_a(a_i)da_i
},$}
\end{equation*}
\begin{equation*}
% \resizebox{\linewidth}{!}{
% $
P(s_i\in A_4|s_1,s_2) = \int_{b_1-\delta_2-\varepsilon}^{b_1-\delta_2}{\left(
\int_{a_2+b_2-b_i}^{\infty}{
f_a(a_i)da_i
}\right)f_b(b_i)db_i.
}
%$}
\end{equation*}}
% \end{lemma}

\begin{proof}
Despite its complicated appearance, this lemma is simply the result of integrating probability distribution functions over $\mathcal{R}^2$. Consider the event $E$ that $s_1$ is the greedy choice for player 2. There are two possible outcomes: player 2 chooses $s_1$ (event $E_1$) or player 2 chooses another strategy (without loss of generality we assume that strategy to be $s_2$) with better net utility such that $b_1-b_2 \leq \delta_2$ (event $E_2$). We refer to these two outcomes as (respectively) case 1 and case 2.

Case 1 occurs if all strategies $s_i$ are in the shaded region in Figure \ref{fig_ub1}. As a function of $s_1$, a strategy $s_i$ is in this region $A_1$ with probability 
\begin{equation*}
\resizebox{\linewidth}{!}{$
P(s_i\in A_1|s_1) = 
\int_{-\infty}^{b_1-\delta_2}{\left(
\int_{-\infty}^{\infty}{
f_a(a_i)da_i
}
\right)f_b(b_i)db_i}+
\int_{b_1-\delta_2}^{b_1}{\left(
\int_{-\infty}^{a_1+b_1-b_i}{
f_a(a_i)da_i
}
\right)f_b(b_i)db_i}.
$}
\end{equation*}
Next, for fixed $s_1$, we need to consider the probability that case 1 occurs and at least one strategy is able to provide an upper bound on $\delta_2$. This occurs if all strategies lie within the shaded area in Figure \ref{fig_ub1} and at least one strategy lies in the shaded area $A_2$ in Figure \ref{fig_ub2}. The probability of a strategy lying in $A_2$ is 
$$
P(s_i\in A_2|s_1) = \int_{b_1-\delta_2-\varepsilon}^{b_1-\delta_2}{\left(
\int_{a_1+b_1-b_i}^{\infty}{
f_a(a_i)da_i
}\right)f_b(b_i)db_i,
}
$$
which means that case 1 occurs and an acceptable bound is established with probability $P(E\cap E_1|s_1) = P(s_i\in A_1|s_1)^{n-1} - (P(s_i\in A_1|s_1)-P(s_i\in A_2|s_1))^{n-1}$.

\begin{figure}[ht]
\centering
\begin{minipage}{.45\linewidth}
  \centering
  \includegraphics[width=.8\linewidth]{Prob_delta_learning_ub1.png}
  \caption{Area for $s_2$ given $s_1$ to establish upper bound, case 1}
  \label{fig_ub1}
\end{minipage}%
\begin{minipage}{.45\linewidth}
  \centering
  \includegraphics[width=.8\linewidth]{Prob_delta_learning_ub2.png}
  \caption{Area for $s_i$ given $s_1$ to establish upper bound, case 1}
  \label{fig_ub2}
\end{minipage}
\end{figure}

We now consider case 2, that another strategy is played, and assume without loss of generality that the strategy played is $s_2$. Case 2 requires that  occurs if all strategies $s_i$ are in the shaded region in Figure \ref{fig_ub3}, for $i\neq 1,2$. As a function of $s_1$ and $s_2$, a strategy $s_i$ is in this region with probability
\begin{equation*}
\resizebox{\linewidth}{!}{$
P(A_2|s_1,s_2) = 
\int_{-\infty}^{b_1-\delta_2}{\left(
\int_{-\infty}^{\infty}{
f_a(a_i)da_i
}
\right)f_b(b_i)db_i}
+
\int_{b_1-\delta_2}^{b_1}{\left(
\int_{-\infty}^{a_2+b_2-b_i}{
f_a(a_i)da_i
}
\right)f_b(b_i)db_i}.
$}
\end{equation*}
Next, we want to consider the probability that case 2 occurs and an acceptable upper bound is observed. For a fixed $s_1,s_2$, $s_2$ is played if all other $s_i$ lie in the shaded region $A_3$ in Figure \ref{fig_lb2}, the probability of which we know from the proof of Lemma 2 is
\begin{equation*}
\resizebox{\linewidth}{!}{$
P(s_i \in A_3|s_1,s_2) = \int_{b_1-\delta}^{b_1}{\left(
\int_{\infty}^{a_2+b_2-b_i}{
f(a_i)da_i}
\right)fb_idb_i}+\int_{-\infty}^{\infty}{\left(
\int_{-\infty}^{b_1-\delta_2}{
f_b(b_i)db_i
}\right)f_a(a_i)da_i
}.
$}
\end{equation*}
In order for one of these $s_i$ to provide an acceptable upper bound, at least one of them must be in the shaded region $A_4$ in Figure \ref{fig_ub3}, providing better net utility than $s_2$ and providing utility $b_i$ such that $b_1-\delta_2-\varepsilon \leq b_i \leq b_1 - \delta_2$. The probability of an $s_i$ being in this region is
\begin{equation*}
% \resizebox{\linewidth}{!}{$
P(s_i\in A_4|s_1,s_2) = \int_{b_1-\delta_2-\varepsilon}^{b_1-\delta_2}{\left(
\int_{a_2+b_2-b_i}^{\infty}{
f_a(a_i)da_i
}\right)f_b(b_i)db_i,
}
% $}
\end{equation*}
which means that case 2 occurs and establishes an acceptable upper bound with probability $P(A_3|s_1,s_2)^{n-2}-(P(A_3|s_1,s_2)-P(A_4|s_1,s_2))^{n-2}$. Note that this value is zero for $n=2$, indicating that in a two strategy game an upper bound cannot be established if case 2 occurs. Therefore, as a function of $s_1$ the probability case 2 occurs and an acceptable upper bound is established is 
\begin{equation*}
\resizebox{\linewidth}{!}{$
P(E\cap E_2|s_1) = \int_{b_1-\delta_2}^{b_1}{\left(
\int_{a_1+b_1-b_2}^{\infty}{
\left(P(A_3|s_1,s_2)^{n-2}-(P(A_3|s_1,s_2)-P(A_4|s_1,s_2))^{n-2}\right)f_a(a_2)da_2
}
\right)f_b(b_2)db_2}.
$}
\end{equation*}
Finally, this gives us that the probability of establishing an acceptable upper bound is 
\begin{equation*}
\resizebox{\linewidth}{!}{$
P^u(\varepsilon) = n\int_{-\infty}^{\infty}{\left(
\int_{-\infty}^{\infty}{
\left( P(E\cap E_1|s_1) + (n-1)P(E\cap E_2|s_1)
\right)}f_a(a_1)da_1
\right)f_b(b_1)db_1}.
$}
\end{equation*}
\end{proof}

\begin{figure}
    \centering
    \includegraphics[width=0.5\linewidth]{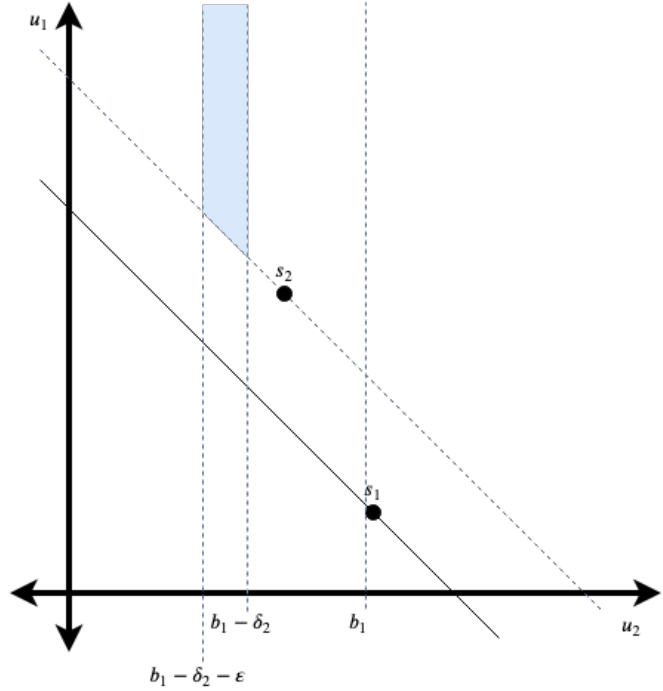}
    \caption{Area for $s_i$ given $s_1,s_2$ to establish upper bound, case 2}
    \label{fig_ub3}
\end{figure}

As we noted before the lemma, when $n\geq3$, it is possible for both upper and lower bounds to be established in a single game. The derivation of Lemma 3 allows us to do so directly through our derivation of $P(E\cap E_2|s_1)$, which was the probability that the player 2 did not select the greedy best response but still revealed an upper bound within $\varepsilon$ of $\delta_2$. 

\textbf{Lemma 4.}\textit{
The probability of player 1 observing both an upper and lower bound within $\varepsilon$ of $\delta_2$ is given by 
\begin{equation*}
% \resizebox{\linewidth}{!}{$
Q(\varepsilon) = n(n-1)\int_{-\infty}^{\infty}{\left(
\int_{-\infty}^{\infty}{
\left(P'(E\cap E_2|s_1)
\right)}f_a(a_1)da_1
\right)f_b(b_1)db_1}.
% $}
\end{equation*}
where
\begin{equation*}
\resizebox{\linewidth}{!}{$
P'(E\cap E_2|s_1) = 
\int_{b_1-\delta_2}^{b_1-\delta_2+\varepsilon}{\left(
\int_{a_1+b_1-b_2}^{\infty}{
\left(P(A_3|s_1,s_2)^{n-2}-(P(A_3|s_1,s_2)-P(A_4|s_1,s_2))^{n-2}\right)f_a(a_2)da_2
}
\right)f_b(b_2)db_2}.
$}
\end{equation*}
}

\textbf{Theorem 3.}
% \begin{theorem}\label{thm_simple_time_m}
\textit{
For the game in Table IV the expected number of games for player 1 to get an estimate $\delta_{21}$ of $\delta_2$ guaranteed to have error most $\varepsilon$ from the true value is less than or equal to 
\begin{align*}
E[\mathcal{T}(\varepsilon)] \leq T(\varepsilon) = \frac{1}{P^u(\varepsilon)+P^l(\varepsilon) - Q(\varepsilon)}
\left( 1+ \frac{P^u(\varepsilon)-Q(\varepsilon)}{P^l(\varepsilon)} + \frac{P^l(\varepsilon)-Q(\varepsilon)}{P^u(\varepsilon)}\right).
\end{align*}}
% \end{theorem}
% Let $Q(\varepsilon) = P(u,l|\varepsilon)$ be the probability of a game which reveals both an acceptable lower and upper bound on $\varepsilon$. 
\begin{proof}
We note that the only way for player 1 to make an estimate of $\delta_2$ which is guaranteed to be within at most $\varepsilon$ of the true value is to find an interval $[\delta_{21}^l,\delta_{21}^u)$ such that $\delta_{21}^u-\delta_{21}^l \leq 2\varepsilon$. Next, we note that in order to obtain this interval at least one of $\delta_{21}^l$ and $\delta_{21}^u$ must be within $\varepsilon$ of $\delta_2$. Therefore we can find an upper bound on the expected number of games required by finding the expected number of games required to observe both upper and lower bounds within $\varepsilon$ of $\delta_2$.

% To aid in our analysis, we turn to the weighted coupon collector problem. While computing or even estimating the value for this problem is generally unwieldy, we only need to collect two ``coupons": an acceptable upper bound and an acceptable lower bound. 
Beginning from the first game, the expected time to discover one or more bounds is $\frac{1}{P^u(\varepsilon)+P^l(\varepsilon) - Q(\varepsilon)}$. This event can occur in any of three ways: an acceptable upper bound is found, an acceptable lower bound is found, or both are found. The probabilities of these events are proportional to $P^u(\varepsilon)-Q(\varepsilon)$, $P^l(\varepsilon)-Q(\varepsilon)$, and $Q(\varepsilon)$ respectively. The expected total time to discover both is therefore 
\begin{equation*}
\resizebox{.95\linewidth}{!}{$
\frac{1}{P^u(\varepsilon)+P^l(\varepsilon) - Q(\varepsilon)} + \frac{P^u(\varepsilon)-Q(\varepsilon)}{P^u(\varepsilon)+P^l(\varepsilon) - Q(\varepsilon)}\frac{1}{P^l(\varepsilon)} + \frac{P^l(\varepsilon)-Q(\varepsilon)}{P^u(\varepsilon)+P^l(\varepsilon) - Q(\varepsilon)}\frac{1}{P^u(\varepsilon)}
$}
\end{equation*}
\begin{equation*}
% \resizebox{\linewidth}{!}{$
=\frac{1}{P^u(\varepsilon)+P^l(\varepsilon) - Q(\varepsilon)}\left( 1+ \frac{P^u(\varepsilon)-Q(\varepsilon)}{P^l(\varepsilon)} + \frac{P^l(\varepsilon)-Q(\varepsilon)}{P^u(\varepsilon)}\right).
% $}
\end{equation*}
% \Halmos
\end{proof}

With the completion of Theorem 3, we now have an upper bound on the expected time to discover both an upper and lower bound within $\varepsilon$ of the game in Table IV. However, we are interested in that $1\times n$ game because it is equivalent to an $m\times n$ game in which the leader has made its decision and is waiting for the follower. Now we return to our original goal, estimating the expected time for the leader in an $m \times n$ game to estimate the follower's $\delta$ value to within $\varepsilon$. 

If player 1's goal is to learn $\delta_2$, rather than to play according to $\delta_1$, the challenge it faces is deciding which of its $m$ strategies to select each game. 
Note also that in pursuing this behavior player 1 has decided to focus purely on exploration and has abandoned any interest in its own utility, which means that it is impossible for its neighbors to learn anything about $\delta_1$ based on its actions as a leader.
Theorem 3 implies that the expected time for strategies which reveal an upper and lower bound within $\varepsilon$ of $\delta_2$ to occur in a randomly generated game is $\frac{1}{m}T(\varepsilon)$, which provides an upper bound on the expected time of player 1's theoretical optimum strategy. 
Similarly, it implies that if the leader chooses its strategy $s_i$ randomly, it has an upper bound of $T(\varepsilon)$ on the expected time to achieve this estimate. However, choosing in such a random manner ignores what the leader has already learned from the follower: while it may not have $\delta_{21}^u-\delta_{21}^l \leq 2\varepsilon$, it will still over time gain some $[\delta_{21}^l,\delta_{21}^u)$ interval in which $\delta_2$ is located through the use of Algorithm 1. 
This allows the leader to determine whether or not it will refine its knowledge of $\delta_2$ by selecting $s_i$. Assume without loss of generality that $s_1$ is the follower's greedy best response to the leader selecting $s_i$: if the Pareto frontier of player 2's responses to $s_i$ contains an $s_j$ such that $\delta_{21}^l < b_1-b_j < \delta_{21}^u$, selecting $s_i$ will result in the player 1 refining at least one of the bounds on $\delta_2$. 
If the frontier does not contain such an $s_j$, there will be no improvement in the bounds by selecting $s_i$.
Figure 1 in Section III-A gives an example of this with player 2's potential response $s_3$: if player 2 responds with $s_3$ then the lower bound $\delta_{21}^l$ will be raised, and if it responds with $s_2$ the upper bound $\delta_{21}^u$ will be lowered.

\section{Derivations}

\subsection{Proof of Theorem 2\label{app_sng_mne_proof}}
{\textbf{Theorem 2} }
\textit{Consider a social network $G$ with uniform interactions $\mathcal{A}_{ij} = \mathcal{B}_{lh}$ for all $l,i,j,h\in [N]$ such that all payoffs are nonnegative and for agent $i$ with neighbors $j,l \neq i$, $\delta_j \leq \delta_l \rightarrow u_i(\theta_{ij}) \leq u_i(\theta_{il})$. Then the $N$-player metagame with closed interval strategy space $\Delta_i\subseteq \mathcal{R}$ and utility function $\mathbf{u}_i$ for $i\in[N]$ possesses a mixed Nash equilibrium.}

For simplicity of notation, in the this proof we will use $\mathbf{u}_i(\delta_i,\delta_{-i})$ instead of $\mathbf{u}_i(\theta_{-i},\delta_i)$. Also, before formally starting the proof we first state the following definition and result: 

\begin{definition}[Weak lower semi-continuity]
$U_i(\sigma_i,\sigma_{-i})$ is weakly lower semi-continuous in $\sigma_i$ if $\forall \sigma_i' \in \Sigma_i^{**}(i)$, $\exists \lambda \in [0,1]$ such that $\forall \sigma_{-i}\in \Sigma_{-i}^{**}(\sigma_i')$, 
\begin{equation*}
% \resizebox{\linewidth}{!}{
\lambda \lim \inf_{\sigma_i\stackrel{-}\rightarrow \sigma_i'}U_i(\sigma_i,\sigma_{-i}) + (1-\lambda)\lim \inf_{\sigma_i\stackrel{+}\rightarrow \sigma_i'}U_i(\sigma_i,\sigma_{-i}) \geq U_i(\sigma_i',\sigma_{-i}).
\end{equation*}
\end{definition}

\textbf{Theorem 4} \textbf{(Dasgupta \& Maskin, 1986)}\label{thm_dasgupta}
\textit{
Let $\Sigma_i\subseteq\mathcal{R}$ for $i\in [N]$ be a closed interval and let $U_i:\Sigma \rightarrow \mathcal{R}$ be continuous except on a subset $\Sigma^{**}(i) \subseteq\Sigma^*(i)$. If $\sum_{i=1}^NU_i(\sigma)$ is upper semi-continuous and $U_i(\sigma_i,\sigma_{-i})$ is bounded and weakly lower semi-continuous in $\sigma_i$ then the $N$-player game with closed interval strategy space $\Sigma_i\subseteq \mathcal{R}$ and utility function $U_i$ for $i\in[N]$ possesses a mixed-strategy Nash equilibrium.}
% \end{theorem}

We are now ready to begin the proof.

\begin{proof}[Proof of Theorem 2]
This proof will make use of Theorem 4. As such, we need to show three things: a set $\Delta^*(i) = \{(\delta_1,...,\delta_N)\in\Delta|\exists j\neq i, \exists d, 1\leq d\leq D(i) \text{ such that }\delta_j=f_{ij}^d(\delta_i)\}$ which appropriately captures discontinuities in $\mathbf{u}_i(\delta)$, that $\mathbf{u}_i(\delta_i,\delta_{-i})$ is bounded and weakly lower semi-continuous in $\delta_i$, and that $\sum_{i\in [N]}\mathbf{u}_i(\delta)$ is upper semi-continuous.

We begin by determining the set $\Delta^*(i)$. We earlier noted that $\mathbf{u}_i(\delta_i,\delta_{-i})$ has at most $|N_i^1|$ discontinuities in $\delta_i$ when all interactions are uniform and nonnegative and $\delta_j \leq \delta_l \rightarrow u_i(\theta_{ij}) \leq u_i(\theta_{il})$, and that all of them occur in $w_i(\delta_i,\delta_{-i})$. $w_i$ is the utility gained by agent $i$ receiving invitations. Therefore, if $i$ receives an invitation from agent $j$ when $\delta_i$ changes, another agent $l$ that previously received an invitation from $j$ now loses it. Given that $\delta_i \leq \delta_l \rightarrow u_j(\theta_{ji}) \leq u_j(\theta_{jl})$, this discontinuity occurs when $\delta_i=\delta_l$. Therefore, we can let $D(i)=1$, $f_{ij}^d(x)=x$, the identity function, and $$\Delta^*(i) = \{(\delta_1,...,\delta_N)\in\Delta|\exists j\neq i, \text{ such that }\delta_j=\delta_i\}$$
will contain all potential discontinuities.

Next we will show that $\mathbf{u}_i(\delta_i,\delta_{-i})$ is bounded and weakly lower semi-continuous in $\delta_i$. From Theorem 1 we observe that $u_i(\theta_{ji})$ and $u_i(\theta_{ij})$ are both continuous in $\delta_i$. Therefore all discontinuities in $w_i(\delta_i,\delta_{-i})$ occur due to $K^2_i$ changing. Consider one such discontinuity point $\delta'$: We know that there is an agent $l$ such that $\delta_i' = \delta_l'$ and there is another agent $j$ which $i$ and $l$ both neighbor who is now indifferent between sending an invitation to agent $i$ and agent $l$. Let $K^2_i$ be the set of invitations $i$ receives from other agents $h\neq j$.
% $$\lambda \lim \inf_{\sigma_i\stackrel{-}\rightarrow \sigma_i'}U_i(\sigma_i,\sigma_{-i}) + (1-\lambda)\lim \inf_{\sigma_i\stackrel{+}\rightarrow \sigma_i'}U_i(\sigma_i,\sigma_{-i}) \geq U_i(\sigma_i',\sigma_{-i})$$
\begin{align*}
\lim_{\delta_i\stackrel{-}\rightarrow \delta_i'} \inf \mathbf{u}_i(\delta_i,\delta_{-i}') =& \lim_{\delta_i\stackrel{-}\rightarrow \delta_i'} \inf v_i(\delta_i,\delta_{-i}') + \sum_{h\in K^2_i}u_i(\delta_h',\delta_i)= v_i(\delta_i',\delta_{-i}') \\
&+ \sum_{h\in K^2_i}u_i(\delta_h',\delta_i')\\
\lim_{\delta_i\stackrel{+}\rightarrow \delta_i'} \inf\mathbf{u}_i(\delta_i,\delta_{-i}') =& \lim_{\delta_i\stackrel{+}\rightarrow \delta_i'} \inf v_i(\delta_i,\delta_{-i}')+\sum_{h\in K^2_i}u_i(\delta_h',\delta_i) + u_i(\delta_j',\delta_i)\\
 =& v_i(\delta_i',\delta_{-i}')+\sum_{h\in K^2_i}u_i(\delta_h',\delta_i') + u_i(\delta_j',\delta_i').
\end{align*}
% $$\lim_{\delta_i\stackrel{-}\rightarrow \delta_i'} \inf \mathbf{u}_i(\delta_i,\delta_{-i}') = \lim_{\delta_i\stackrel{-}\rightarrow \delta_i'} \inf v_i(\delta_i,\delta_{-i}')+\sum_{h\in K^2_i}u_i(\delta_h',\delta_i)= v_i(\delta_i',\delta_{-i}')+\sum_{h\in K^2_i}u_i(\delta_h',\delta_i')$$
% $$\lim_{\delta_i\stackrel{+}\rightarrow \delta_i'} \inf\mathbf{u}_i(\delta_i,\delta_{-i}') = \lim_{\delta_i\stackrel{+}\rightarrow \delta_i'} \inf v_i(\delta_i,\delta_{-i}')+\sum_{h\in K^2_i}u_i(\delta_h',\delta_i) + u_i(\delta_j',\delta_i)= v_i(\delta_i',\delta_{-i}')+\sum_{h\in K^2_i}u_i(\delta_h',\delta_i') + u_i(\delta_j',\delta_i')$$
% Let $L$ be the set of agents $l\neq i$ which agent $j$ is now indifferent between sending invitations to. Then
% $$\mathbf{u}_i(\delta_i',\delta_{-i}) = v_i(\delta_i',\delta_{-i}') + \sum_{h\in K^2_i}u_i(\delta_h',\delta_i') + \frac{1}{|L|+1}u_i(\delta_j',\delta_i').$$
Therefore $\mathbf{u}_i(\delta_i,\delta_{-i})$ is weakly lower semi-continuous by selecting $\lambda=1$.

Finally, we prove $\mathbf{u}(\delta) = \sum_{i\in[N]}\mathbf{u}_i(\delta)$ is upper semi-continuous. Actually, we will prove the stronger condition that it is continuous. Let $\delta'$ be a point of discontinuity for some $\mathbf{u}_i$. As we have discussed, this occurs due to some other agent $j$ shifting on whether or not to issue an invitation to agent $i$ or another of its neighbors agent $l$ (the case where it is actually a set of neighbors $L$ follows naturally). As a consequence, $\delta'$ is also a point of discontinuity for $\mathbf{u}_l$. However, it is not a point of discontinuity for $\mathbf{u}_j$. Theorem 1 shows $u_j(\theta_{ji})$ and $u_j(\theta_{jl})$ are continuous in $\delta_i$ and $\delta_l$, respectively. Therefore $w_j(\delta)$ is continuous in both. While the set of invites $j$ issues, $K_j^1$, is subject to change, $v_j(\delta)$ the sum of the $k_j$ highest values in the set of functions $\{u_j(\theta_{jh})\}_{h\in N_j^1}$, all of which are continuous in $\delta$ and is therefore continuous as well. This means $\mathbf{u}_j(\delta)$ is continuous, leaving us to show that while $\mathbf{u}_i$ and $\mathbf{u}_l$ are discontinuous at $\delta'$, the sum of the two functions is not. Again, we only need concern ourselves with showing $w_i+w_l$ is continuous.
\begin{align*}
\lim_{\delta_i \stackrel{-}\rightarrow \delta_i'} w_i(\delta_i,\delta_{-i}') + w_l(\delta_i,\delta_{-i}') &= \lim_{\delta_i \stackrel{-}\rightarrow \delta_i'}w_i(\delta_i,\delta_{-i}') + \lim_{\delta_i \stackrel{-}\rightarrow \delta_i'}w_l(\delta_i,\delta_{-i}')\\
&= \sum_{h\in K^1_i\setminus j}u_i(\delta_h',\delta_i) + \sum_{h\in K^1_l\setminus j}u_l(\delta_h',\delta_l') + u_l(\delta_j',\delta_l')\\
\lim_{\delta_i \stackrel{+}\rightarrow \delta_i'} w_i(\delta_i,\delta_{-i}') + w_l(\delta_i,\delta_{-i}') &= \lim_{\delta_i \stackrel{+}\rightarrow \delta_i'}w_i(\delta_i,\delta_{-i}') + \lim_{\delta_i \stackrel{+}\rightarrow \delta_i'}w_l(\delta_i,\delta_{-i}')\\
&= \sum_{h\in K^1_i\setminus j}u_i(\delta_h',\delta_i) + \sum_{h\in K^1_l\setminus j}u_l(\delta_h',\delta_l') + u_i(\delta_j',\delta_i)\\
\end{align*}

The fact that $u_l(\delta_j',\delta_l') = \lim_{\delta_i\rightarrow \delta_i'} u_i(\delta_j',\delta_i)$ implies that the left- and right-side limits are equal, and that they are equal to the actual value
\begin{equation*}
\resizebox{\linewidth}{!}{$
w_i(\delta_i',\delta_{-i}') + w_l(\delta_i',\delta_{-i}') = \sum_{h\in K^1_i\setminus j}u_i(\delta_h',\delta_i') + \sum_{h\in K^1_l\setminus j}u_l(\delta_h',\delta_l') + \frac{1}{2}u_i(\delta_j',\delta_i') + \frac{1}{2}u_l(\delta_j',\delta_l')
$}
\end{equation*}
at $\delta'$. This implies that $\mathbf{u}_i+\mathbf{u}_l$ is continuous at $\delta'$ and hence that $u(\delta)$ is continuous. Therefore, the metagame of selecting $\delta_i\in \Delta_i$ possesses a mixed Nash equilibrium in this setting. 

\end{proof}

\section{{Additional Social Network Validation}}\label{app_additional_numeric}

In Section 5, we presented a numerical analysis of the LTE in the Zachary's Karate Club and the Facebook-ego social networks. To reinforce the results of this analysis, we have repeated it on five additional networks from the Konect Project \cite{konect}. Four of these networks are human social networks, and the fifth we have constructed from an event-based human contact network; more details will be given on this later. By repeating the numerical analysis from Section 5, we find that the behaviors we observed in Zachary's Karate Club and the Facebook-ego network are echoed in these networks and provide further experimental validation for our model.

For all testing in this appendix, the following parameters are used:
\begin{itemize}
    \item $\mathcal{A}_{ij} = \mathcal{B}_{hl}$ for all $h,i,j,l\in [N]$. $A_{ij} \sim \mathcal{A}_{ij}$ is a $2\times2$ matrix with entries generated iid from the exponential distribution with $\lambda=4$.
    \item $\delta_i\in [0,30]$ $\forall i\in [N]$.
    \item For tests in the epoch update system, epochs last 100 rounds.
    \item For tests in the probabilistic update system, each agent updates its $\delta$ with independent probability $\frac{1}{100}$.
    \item Unless computational concerns are present, each trial runs for either 50 epochs or 3000 rounds depending on which update scheme is being used.
\end{itemize}

\subsection{Highland Tribes}

The first social network we will consider is the Gahuku-Gama tribal alliance network from the Eastern Central Highlands of New Guinea, from \cite{kenneth54}. The network consists of 16 vertices, each representing a different tribe. Edges between vertices are undirected and signed, with a sign of 1 denoting allies and a sign of -1 denoting enemies. Because we are interested in cooperative activities within social networks, we removed all -1 signed edges from the network prior to performing our numerical analysis. The resulting network contains 16 vertices and 29 edges. Interestingly, though perhaps not unexpectedly for an alliance network, this resulted in two entirely separate components within the network with one 4-clique and one less densely connected 12-vertex component with 23 edges. This is the only network we analyzed with multiple components.

\begin{figure}[ht]
    \centering
    \begin{minipage}{.5\linewidth}
      \centering
      \captionsetup{width=.85\linewidth}
      \includegraphics[width=\linewidth]{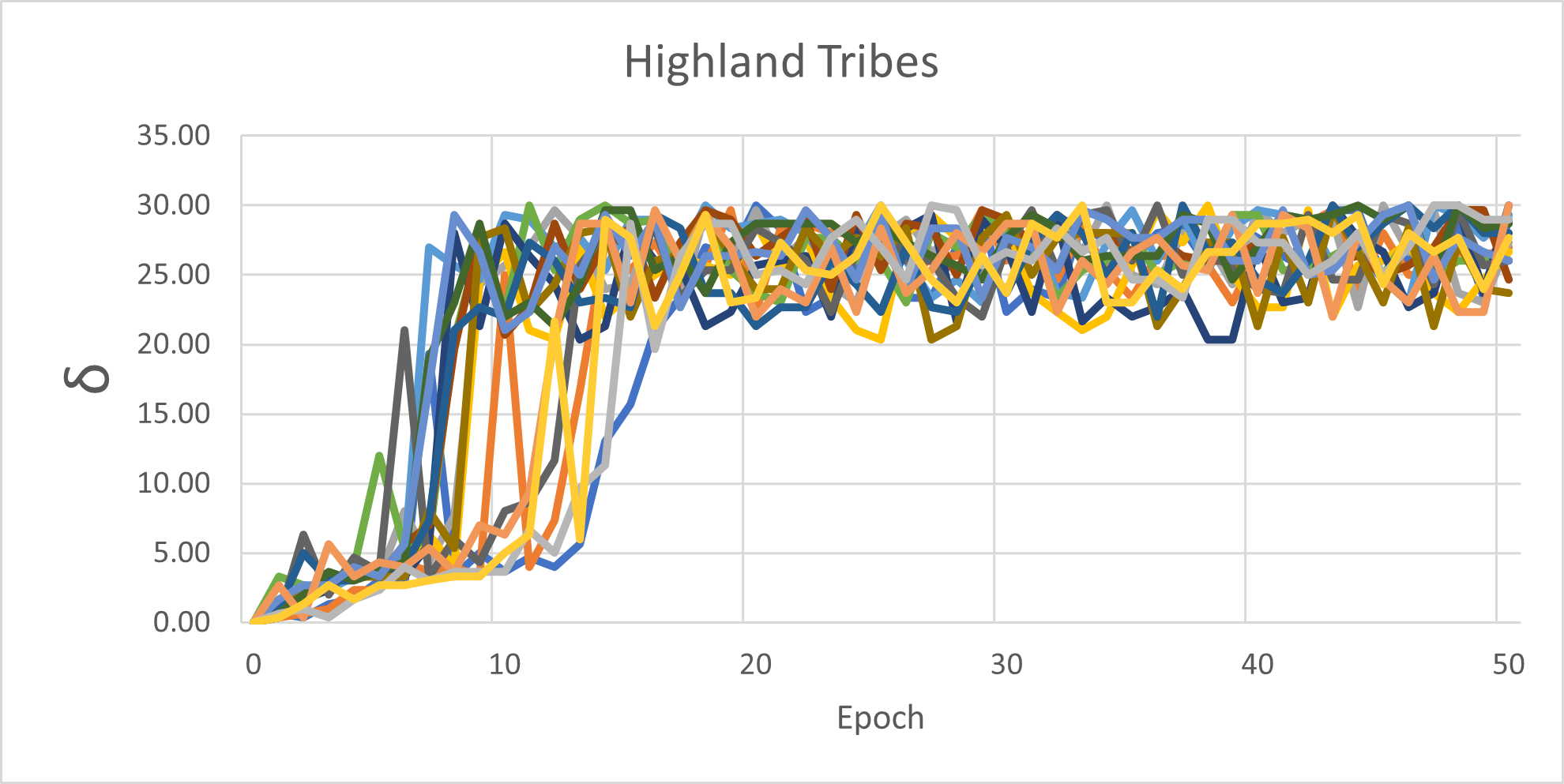}
      \caption{$\delta$ in Gahuku-Gama Tribal Alliance Network, $\delta$ unknown, epoch = 100}
      \label{fig_ht_epoch_unknown}
    \end{minipage}%
    \begin{minipage}{.5\linewidth}
      \centering
      \captionsetup{width=.85\linewidth}
      \includegraphics[width=\linewidth]{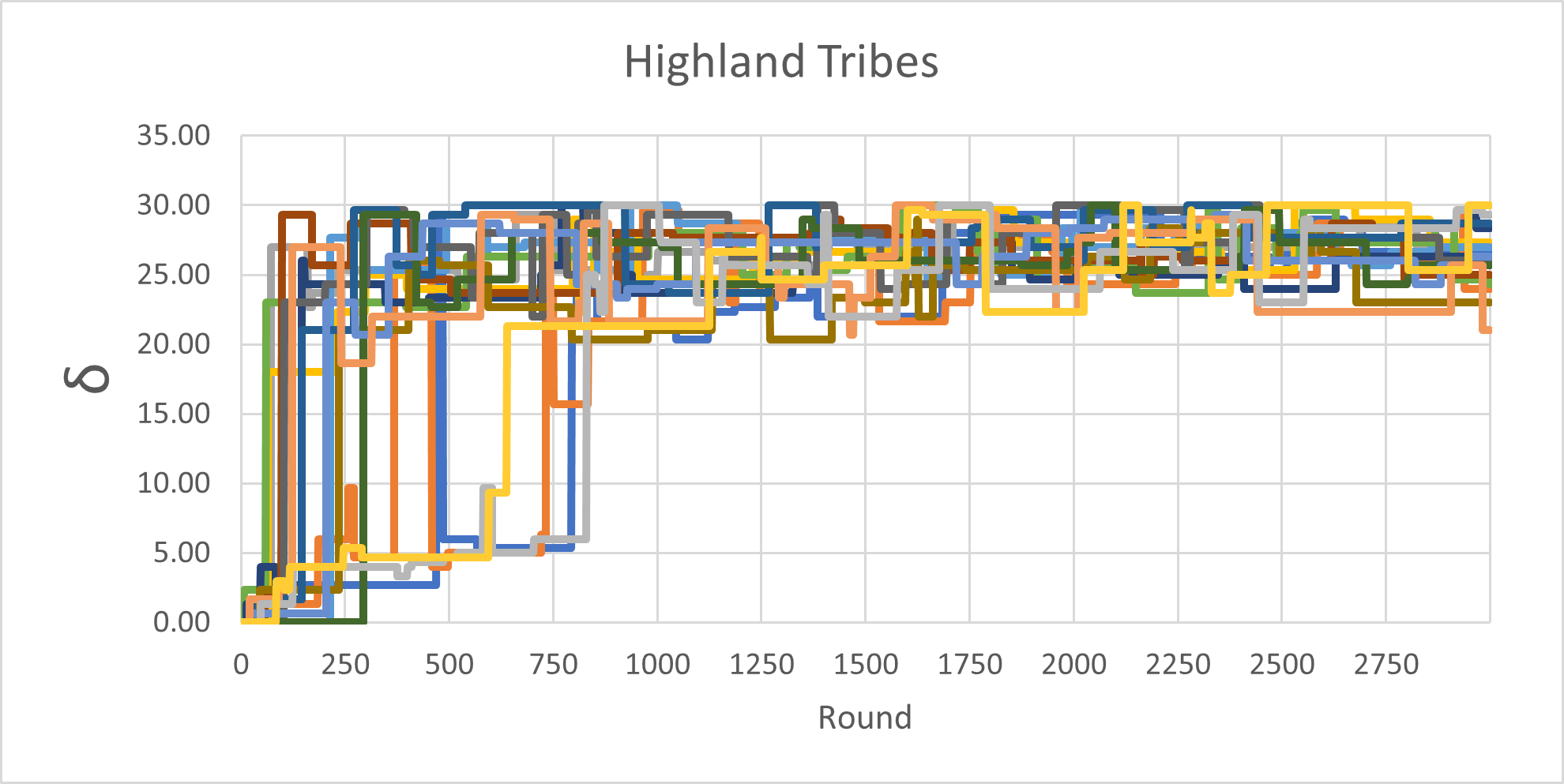}
      \caption{$\delta$ in Gahuku-Gama Tribal Alliance Network, $\delta$ unknown, update probability = $\frac{1}{100}$}
      \label{fig_ht_prob_unknown}
    \end{minipage}
\end{figure}

\begin{figure}[ht]
    \centering
    \begin{minipage}{.5\linewidth}
      \centering
      \captionsetup{width=.85\linewidth}
      \includegraphics[width=\linewidth]{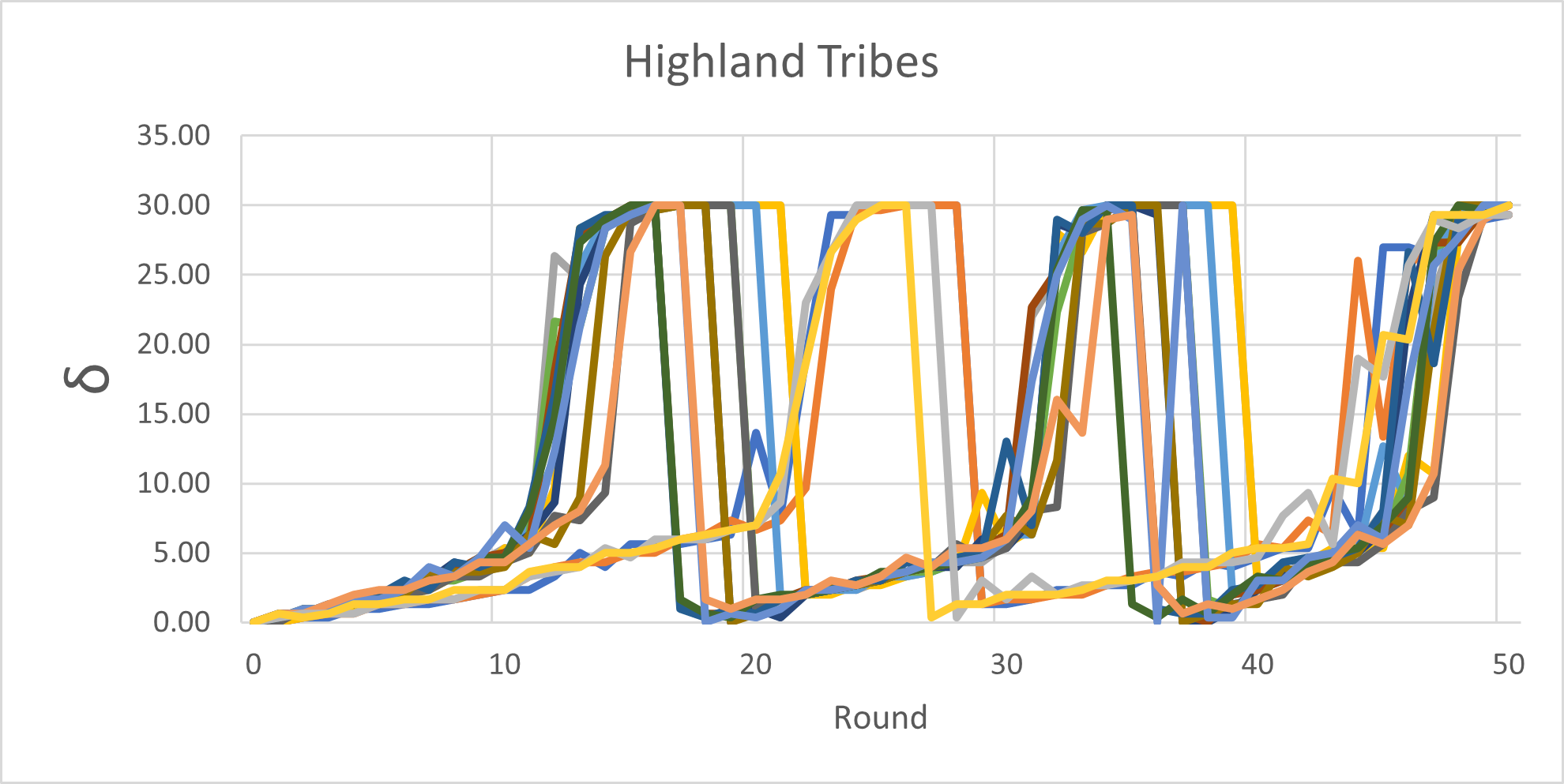}
      \caption{$\delta$ in Gahuku-Gama Tribal Alliance Network, $\delta$ known, lexicographic tie-breaking}
      \label{fig_ht_epoch_known_lex}
    \end{minipage}%
    \begin{minipage}{.5\linewidth}
      \centering
      \captionsetup{width=.85\linewidth}
      \includegraphics[width=\linewidth]{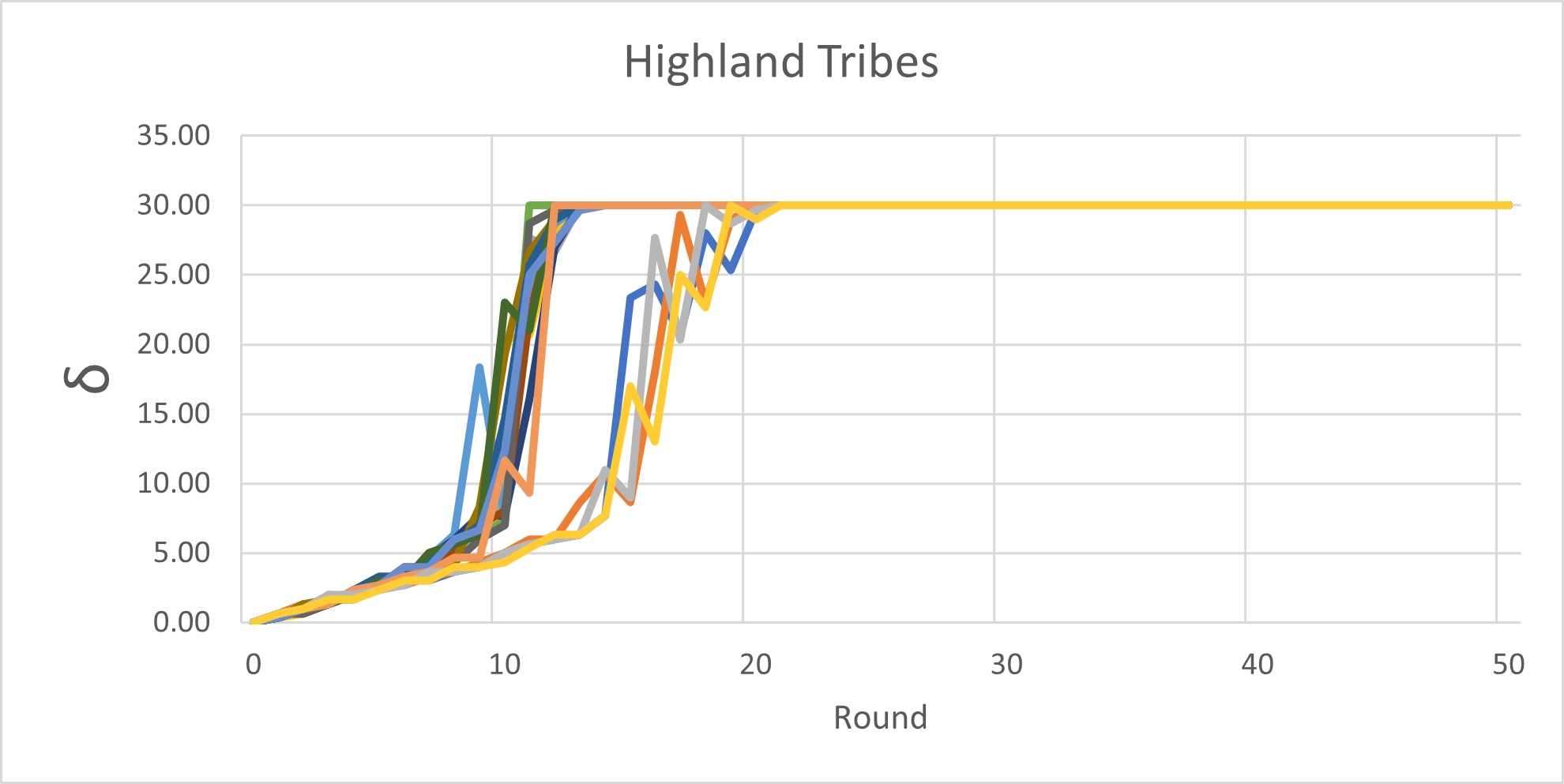}
      \caption{$\delta$ in Gahuku-Gama Tribal Alliance Network, $\delta$ known, random tie-breaking}
      \label{fig_ht_known_rand}
    \end{minipage}
\end{figure}

 The plots from the numerical analysis appear to strongly reinforce our conclusions from Section 5: In Figures \ref{fig_ht_epoch_unknown} and \ref{fig_ht_prob_unknown} we can see behavior which strongly resemblance that in Figures \ref{fig_zkc_unknown_epoch} and \ref{fig_zkc_unknown_random} respectively, with a sharp s-curve to $\delta=30$ followed by oscillating values as agents work to understand each others' $\delta$ values and change their own accordingly. Figures \ref{fig_ht_epoch_known_lex} displays the same repeating s-curves found in Figure \ref{fig_zkc_known_arbitrary} as a result of lexicographic tie-breaking. Figure \ref{fig_ht_known_rand} strongly resembles Figure \ref{fig_zkc_known_random}, with the primary difference being that unlike in the Karate Club network, no agents deviate from $\delta=30$ after reaching it. The two separate s-curves seen in the figure are due to the Alliance network containing two separate non-connected components, with each curve corresponding to one of the components. While not as clearly visible, this same separate behavior from the two components can also be seen in Figures \ref{fig_ht_epoch_unknown} and \ref{fig_ht_epoch_known_lex}.

The average utility per round for the settings displayed in Figures \ref{fig_ht_epoch_unknown}-\ref{fig_ht_known_rand} is 28.104, 28.221, 26.988, and 27.687 which correspond to percentage increases of 21.049\%, 21.553\%, 16.242\% and 19.253\% over the average utility of 23.217 per round produced by agents engaging in selfish behavior.

\subsection{Southern Women}

We will next consider the Southern Women (large) network from \cite{davis2009}. Unlike the other networks, this is an interaction network rather than a social network. It is a bipartite graph consisting of 32 vertices, with $n_1=18$ vertices in one set and $n_2=14$ vertices in the other, with 89 edges between the sets. The left set of vertices represents a group of 18 white women from the Southern United States in the 1930's. The right set of vertices represents a set of 14 social events which took place over a nine-month period, and an edge between two vertices indicates that the woman corresponding to the vertex from the left set attended the event corresponding to the vertex from the right set.

This network cannot be considered a social network, as no individuals are connected to other individuals only to events. However, we are able to construct a social network from it by doing the following:
\begin{enumerate}
    \item Create a new network with 18 vertices and no edges. 
    \item Assign each woman from the original network to a vertex in the new network.
    \item For each pair of vertices in the new network, add an edge between them if the corresponding vertices in the original network had at least three neighbors in common (i.e. the corresponding women attended at least three of the same events). 
\end{enumerate}
The new network can be considered an underlying social network of the original. The decision to only include an edge between women if they attended at least three of the same events is arbitrary, however it is unreasonable to assume that any two women who attended one common event are friend. It contains 15 vertices and 42 edges; while the original network contained 18 women, two of their vertices became isolated after as they had only attended two events in total and a third had not attended three events with another woman and so all were dropped.

\begin{figure}[ht]
    \centering
    \begin{minipage}{.5\linewidth}
      \centering
      \captionsetup{width=.85\linewidth}
      \includegraphics[width=\linewidth]{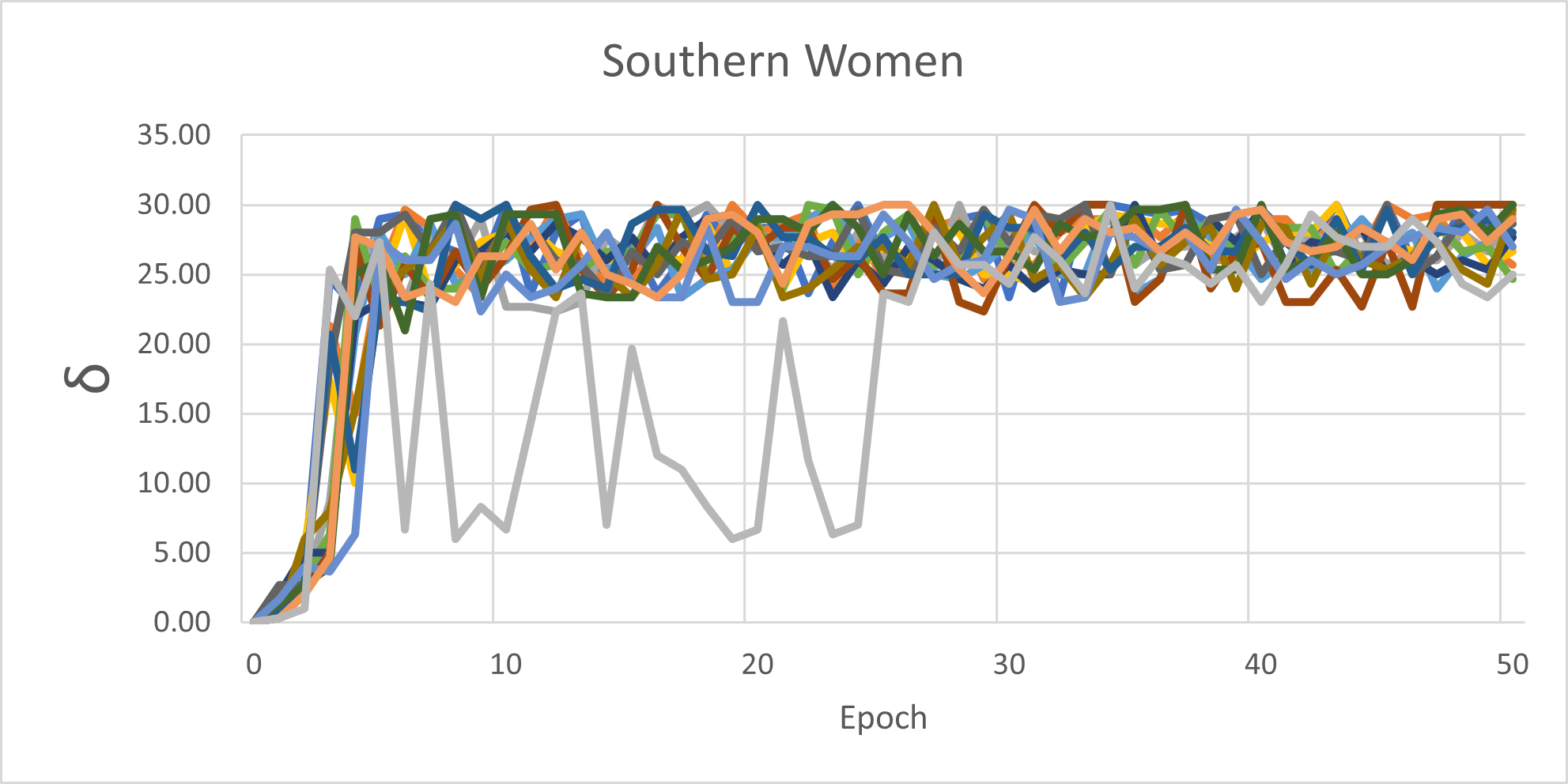}
      \caption{$\delta$ in Southern Women Interaction Network, $\delta$ unknown, epoch = 100}
      \label{fig_sw_epoch_unknown}
    \end{minipage}%
    \begin{minipage}{.5\linewidth}
      \centering
      \captionsetup{width=.85\linewidth}
      \includegraphics[width=\linewidth]{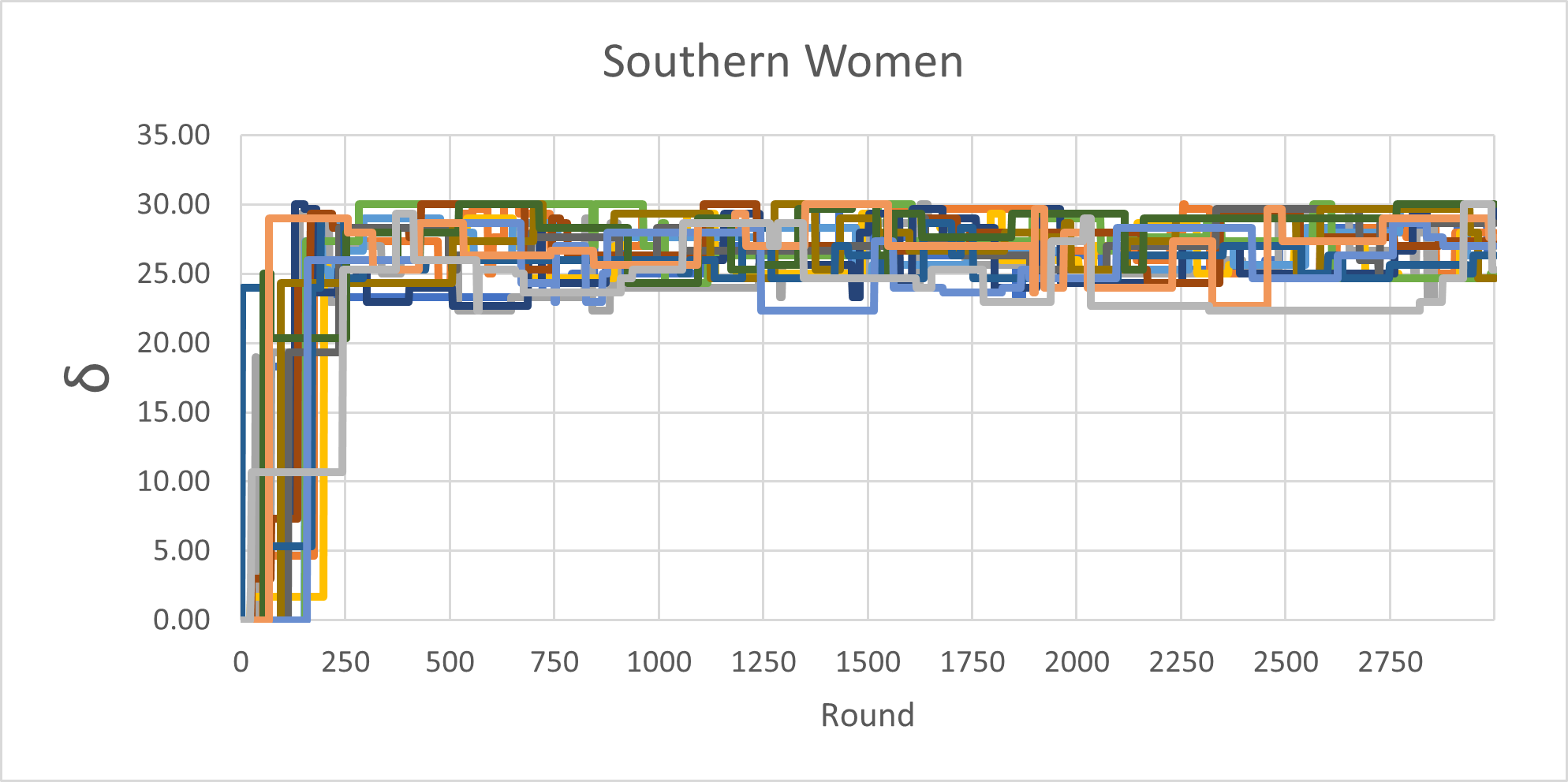}
      \caption{$\delta$ in Southern Women Interaction Network, $\delta$ unknown, update probability = $\frac{1}{100}$}
      \label{fig_sw_prob_unknown}
    \end{minipage}
\end{figure}

\begin{figure}[ht]
    \centering
    \begin{minipage}{.5\linewidth}
      \centering
      \captionsetup{width=.85\linewidth}
      \includegraphics[width=\linewidth]{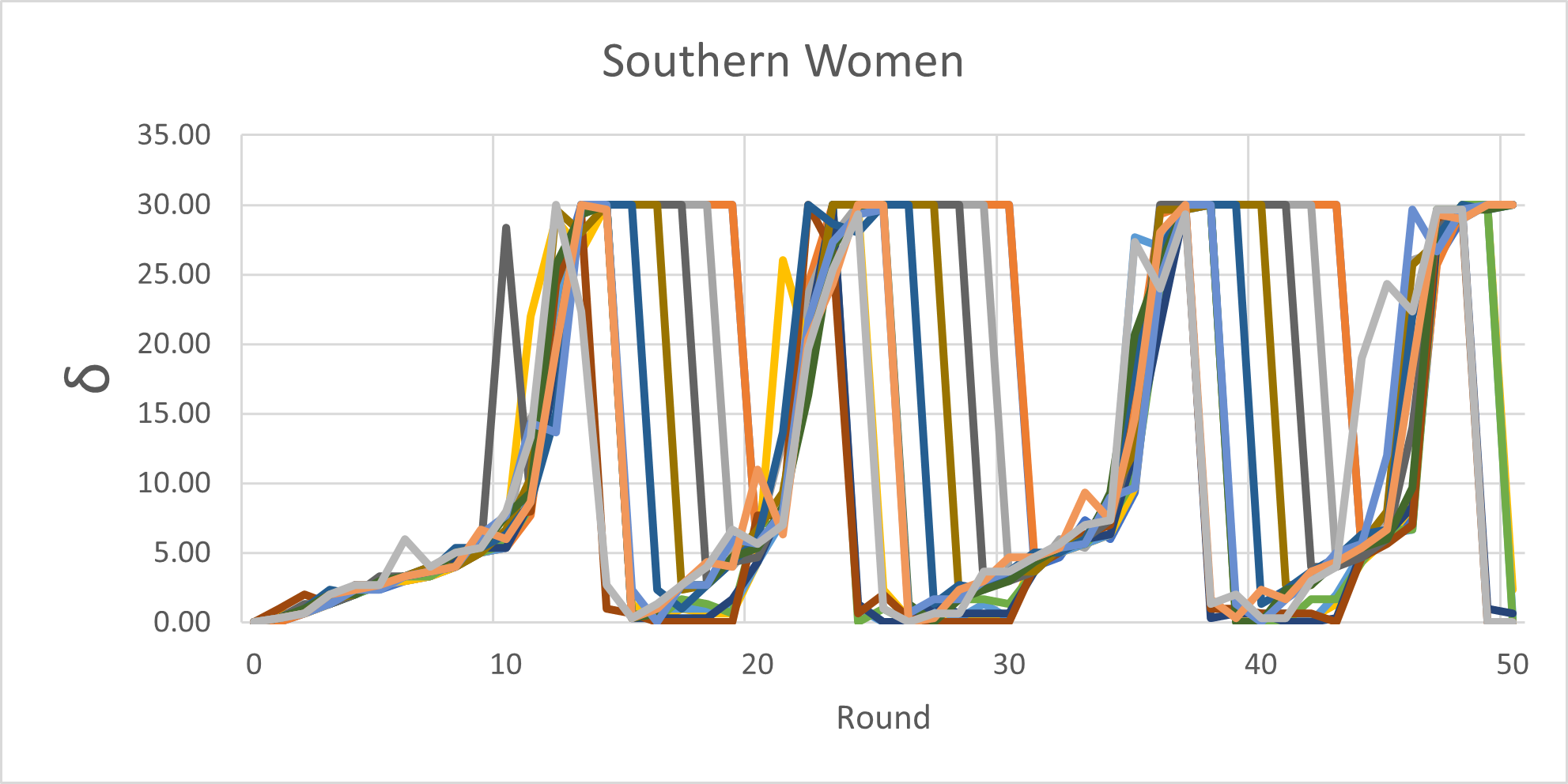}
      \caption{$\delta$ in Southern Women Interaction Network, $\delta$ known, lexicographic tie-breaking}
      \label{fig_sw_epoch_known_lex}
    \end{minipage}%
    \begin{minipage}{.5\linewidth}
      \centering
      \captionsetup{width=.85\linewidth}
      \includegraphics[width=\linewidth]{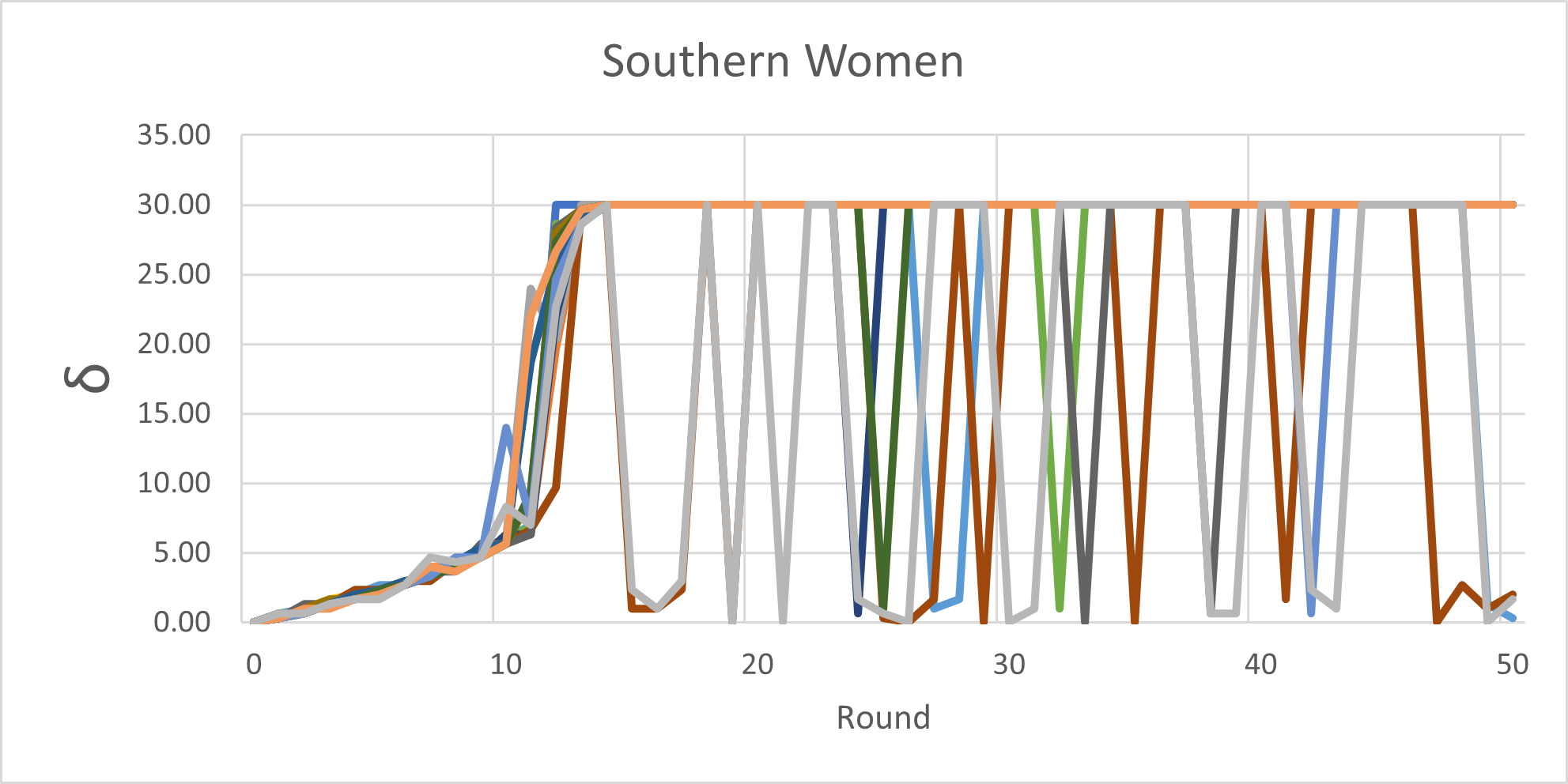}
      \caption{$\delta$ in Southern Women Interaction Network, $\delta$ known, random tie-breaking}
      \label{fig_sw_known_rand}
    \end{minipage}
\end{figure}

Figures \ref{fig_sw_epoch_unknown} and \ref{fig_sw_prob_unknown} display a strong resemblance to Figures \ref{fig_zkc_unknown_epoch} and \ref{fig_zkc_unknown_random} respectively, exhibiting the same behavior of a quick jump to $\delta=30$ followed by oscillating values in the mid-to-high 20's as agents work to understand each others' $\delta$ values and change their own accordingly. Figure \ref{fig_sw_epoch_known_lex} displays the repeating s-curves associated with lexicographic tie-breaking as seen in Figure \ref{fig_zkc_known_arbitrary}. Likewise, Figure \ref{fig_sw_known_rand} displays similar behavior to Figure \ref{fig_zkc_known_random}.

The average utility per round for the settings displayed in Figures \ref{fig_sw_epoch_unknown}-\ref{fig_sw_known_rand} is 27.149, 27.237, 26.298, and 26.844 which correspond to percentage increases of 40.348\%, 40.803\%, 35.949\% and 38.772\% over the average utility of 19.344 per round produced by agents engaging in selfish behavior.

\subsection{Taro Exchange}\label{subsec_taroexchange}

The next social network is a gift giving network between households in a Papuan village from \cite{schwimmer1973}. Vertices represent households, while an edge from vertex $i$ to vertex $j$ represents that household $i$ gave a gift to household $j$. While these edges are directed, the network is symmetric in that for every edge $(i,j)$ there is an edge $(j,i)$ and so for our analysis it is equivalent to an undirected network. The original network has 22 vertices and 78 edges, while the undirected version has 22 vertices and 39 edges.

\begin{figure}[ht]
    \centering
    \begin{minipage}{.5\linewidth}
      \centering
      \captionsetup{width=.85\linewidth}
      \includegraphics[width=\linewidth]{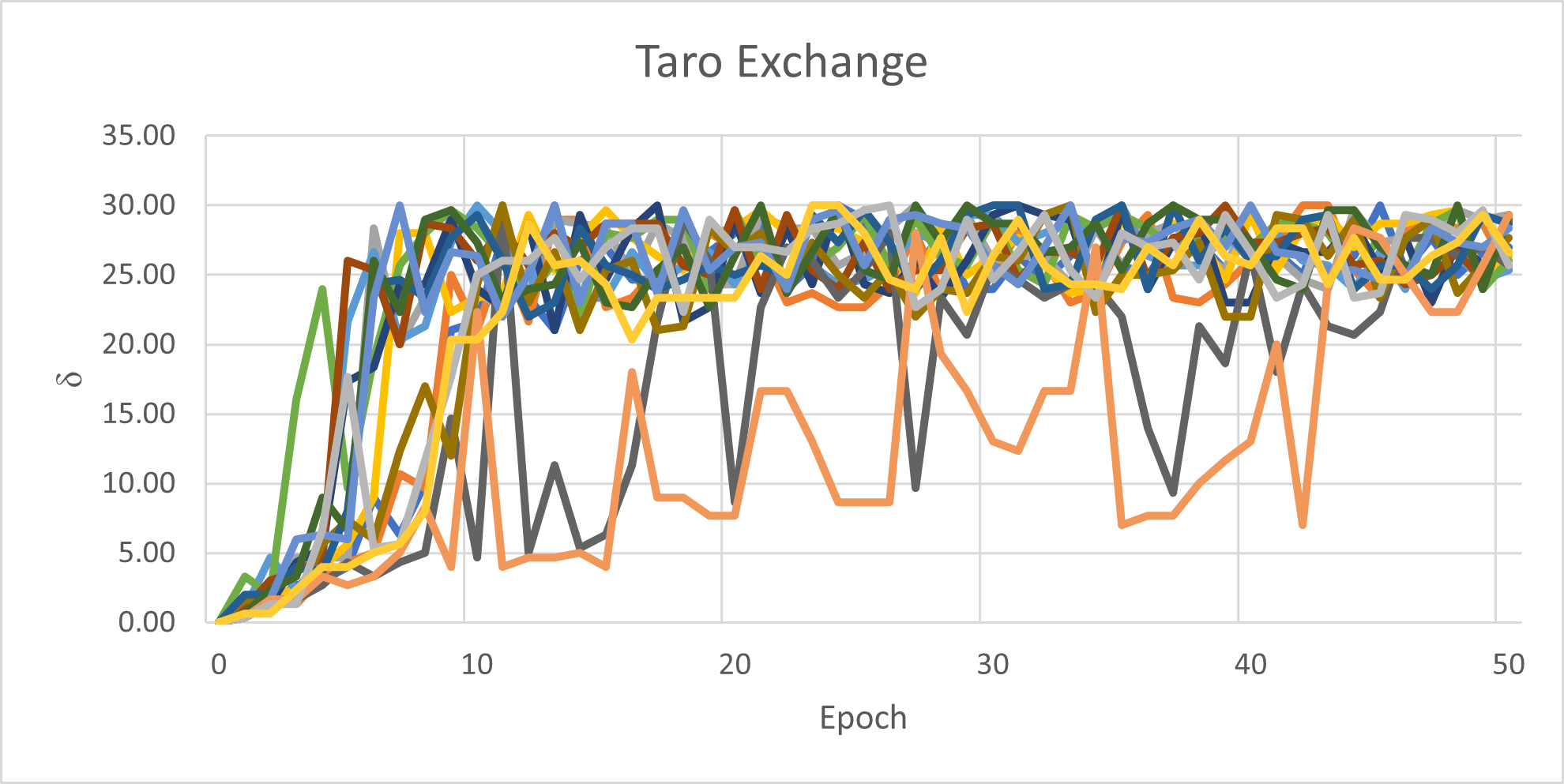}
      \caption{$\delta$ in Taro Exchange Gift-Giving Network, $\delta$ unknown, epoch = 100}
      \label{fig_te_epoch_unknown}
    \end{minipage}%
    \begin{minipage}{.5\linewidth}
      \centering
      \captionsetup{width=.85\linewidth}
      \includegraphics[width=\linewidth]{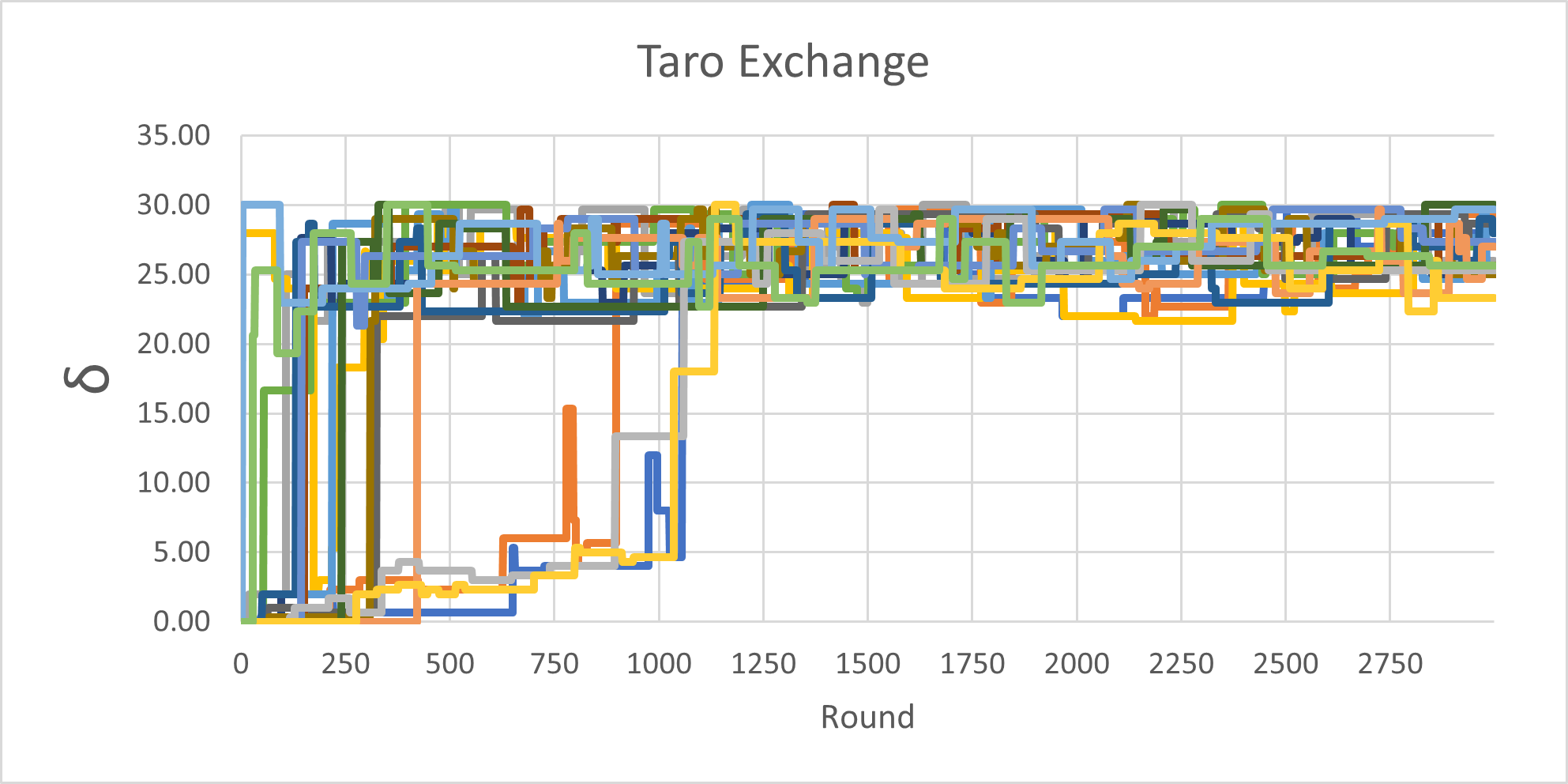}
      \caption{$\delta$ in Taro Exchange Gift-Giving Network, $\delta$ unknown, update probability = $\frac{1}{100}$}
      \label{fig_te_prob_unknown}
    \end{minipage}
\end{figure}

\begin{figure}[ht]
    \centering
    \begin{minipage}{.5\linewidth}
      \centering
      \captionsetup{width=.85\linewidth}
      \includegraphics[width=\linewidth]{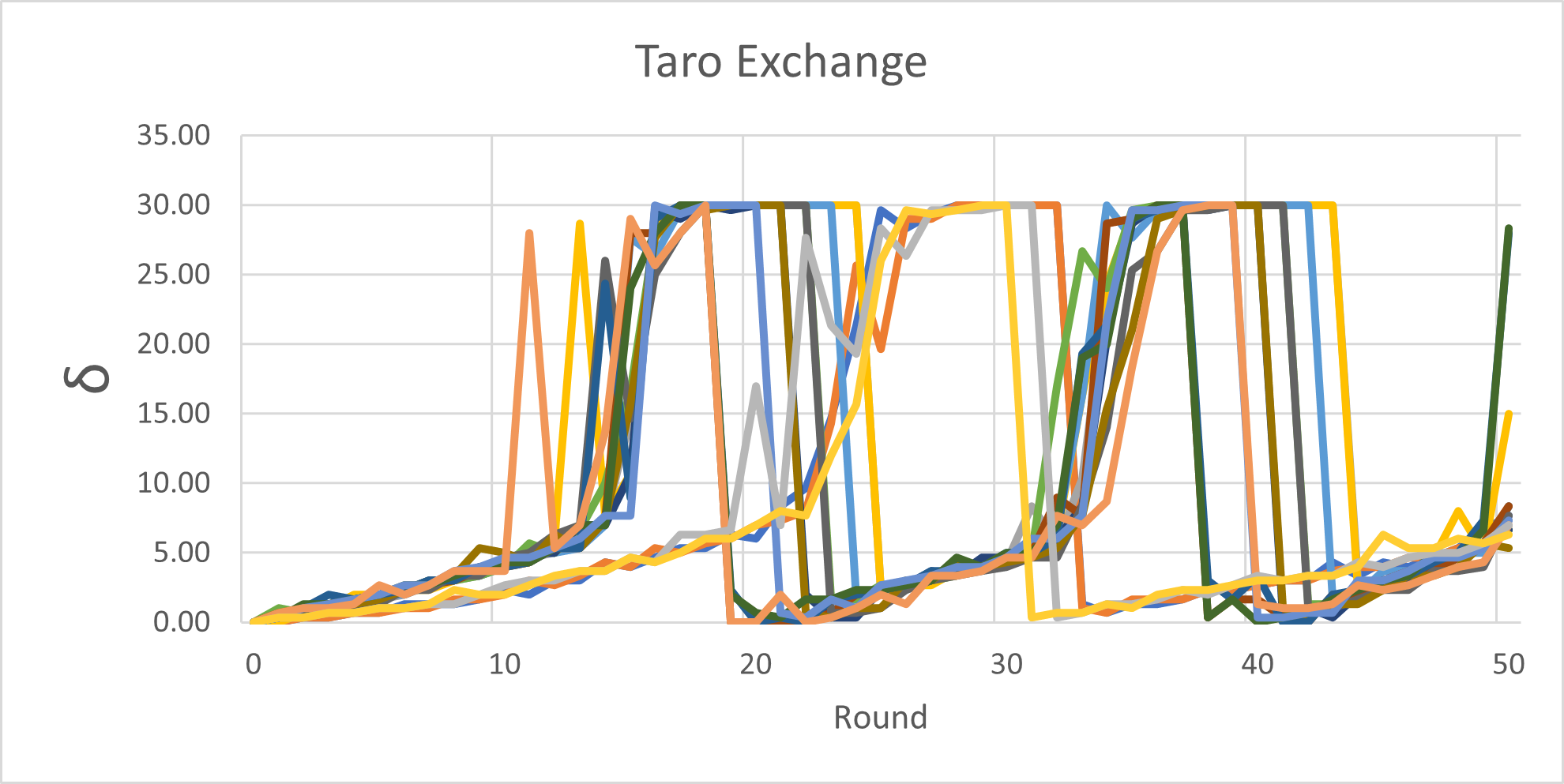}
      \caption{$\delta$ in Taro Exchange Gift-Giving Network, $\delta$ known, lexicographic tie-breaking}
      \label{fig_te_epoch_known_lex}
    \end{minipage}%
    \begin{minipage}{.5\linewidth}
      \centering
      \captionsetup{width=.85\linewidth}
      \includegraphics[width=\linewidth]{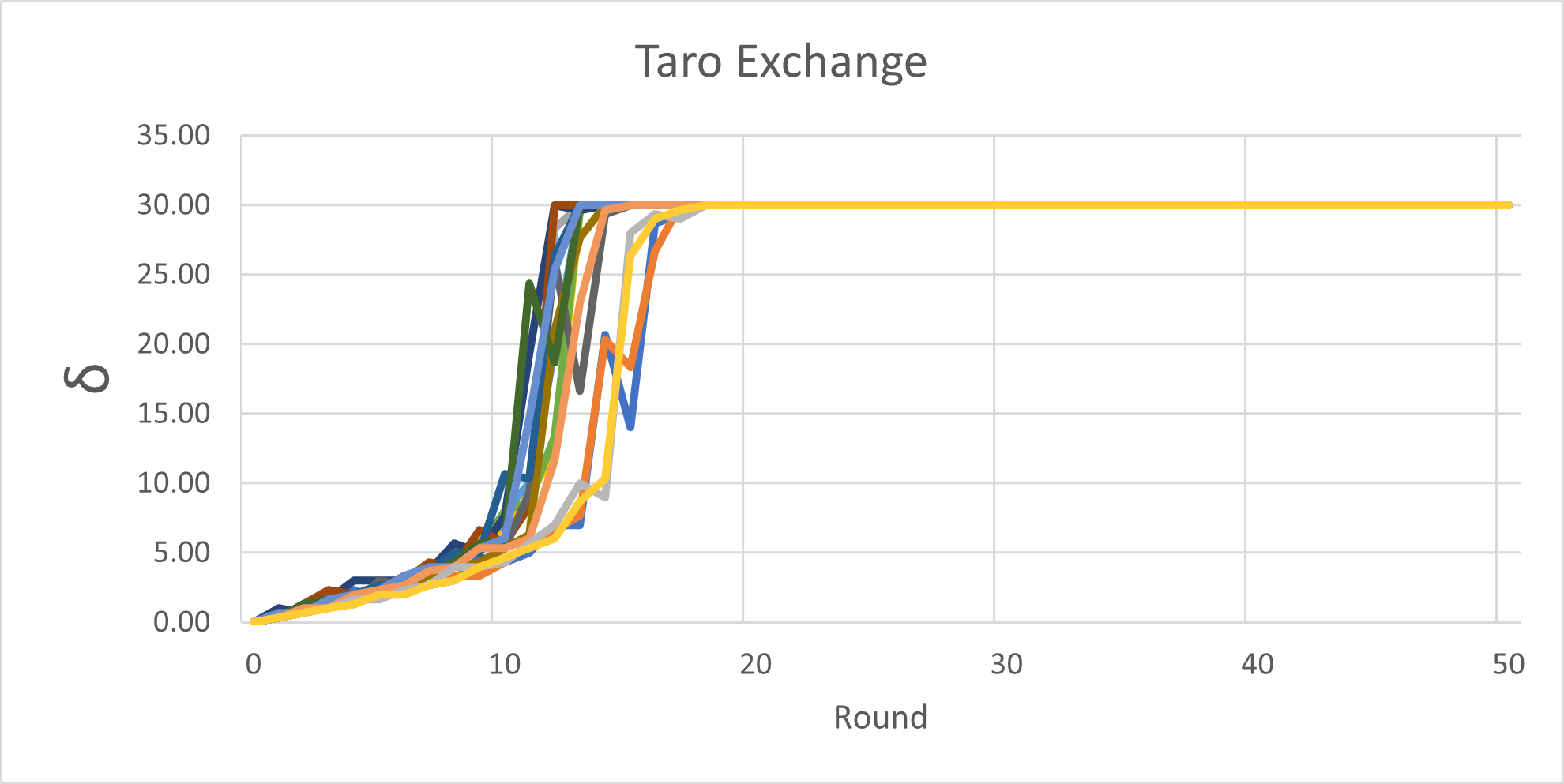}
      \caption{$\delta$ in Taro Exchange Gift-Giving Network, $\delta$ known, random tie-breaking}
      \label{fig_te_known_rand}
    \end{minipage}
\end{figure}

Figures \ref{fig_te_epoch_unknown} and \ref{fig_te_prob_unknown} display a strong resemblance to Figures \ref{fig_zkc_unknown_epoch} and \ref{fig_zkc_unknown_random} respectively, exhibiting the same behavior of a quick jump to $\delta=30$ followed by oscillating values in the mid-to-high 20's as agents work to understand each others' $\delta$ values and change their own accordingly. Figures \ref{fig_te_epoch_known_lex} and \ref{fig_te_known_rand} display the repeating s-curves associated with lexicographic tie-breaking as seen in Figure \ref{fig_zkc_known_arbitrary}. As in Figures \ref{fig_ht_known_rand} and \ref{fig_sw_known_rand}, while \ref{fig_te_known_rand} does not initially appear to display the same results as Figures \ref{fig_zkc_known_random} and Figure \ref{fig_sw_known_rand}, as it has no agents deviating from a high $\delta$ value after reaching the max, they are actually very similar: over 90\% of agents in the Karate Club network and 80\% in the Southern Women network maintain a high $\delta$ value at any given time. The difference in the ratios between the three networks is due to internal structure.

The average utility per round for the settings displayed in Figures \ref{fig_te_epoch_unknown}-\ref{fig_te_known_rand} is 28.140, 27.860, 27.040, and 27.394 which correspond to percentage increases of 23.709\%, 22.478\%, 18.873\% and 20.429\% over the average utility of 22.747 per round produced by agents engaging in selfish behavior.

\subsection{High School}\label{subsec_highschool}

The next social network is a friendship network among boys attending a small high school in Illinois from \cite{coleman1964}. Each boy was asked to name whether other boys were their friends, with a directed edge from vertex $i$ to vertex $j$ indicating that boy $i$ said boy $j$ was a friend. Each boy was asked twice, once in the fall of 1957 and once in the spring of 1958, and this network is an aggregate of the results. The edges have weights in $\{1,2\}$, with 1 indicating that $i$ identified $j$ as a friend during only one of the two asking dates and 2 indicating that $i$ identified $j$ as a friend both time.

While it is a fascinating topic that we are looking to explore in our future work, our model does not consider edge weights and so we treated all edges as having weight 1. Additionally, while our model does allow for directed edges, we used an undirected version of the network in order to keep consistency in our analysis, as the majority of our networks are undirected. The resulting network has 70 vertices and 274 edges.

\begin{figure}[ht]
    \centering
    \begin{minipage}{.5\linewidth}
      \centering
      \captionsetup{width=.85\linewidth}
      \includegraphics[width=\linewidth]{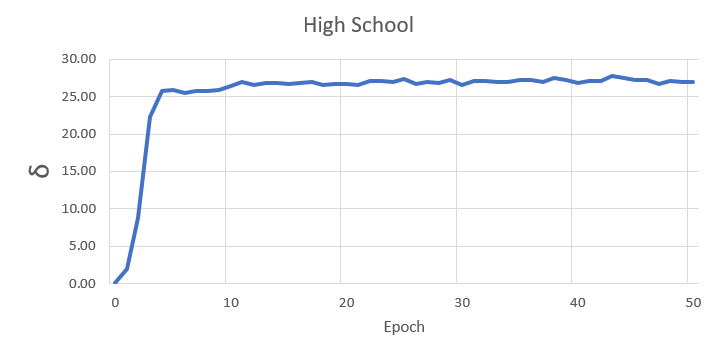}
      \caption{Mean $\delta$ in High School Friendship Network, $\delta$ unknown, epoch = 100}
      \label{fig_hs_epoch_unknown}
    \end{minipage}%
    \begin{minipage}{.5\linewidth}
      \centering
      \captionsetup{width=.85\linewidth}
      \includegraphics[width=\linewidth]{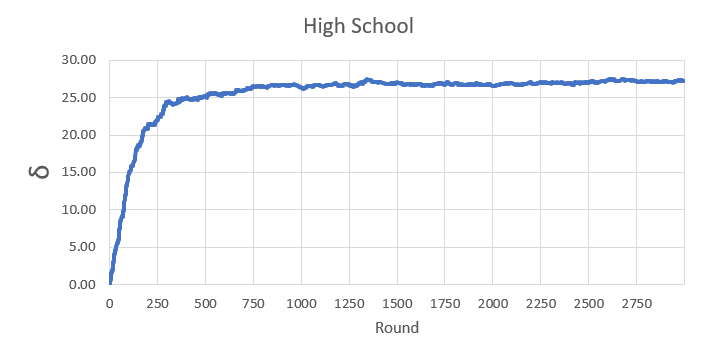}
      \caption{Mean $\delta$ in High School Friendship Network, $\delta$ unknown, update probability = $\frac{1}{100}$}
      \label{fig_hs_prob_unknown}
    \end{minipage}
\end{figure}

\begin{figure}[ht]
    \centering
    \begin{minipage}{.5\linewidth}
      \centering
      \captionsetup{width=.85\linewidth}
      \includegraphics[width=\linewidth]{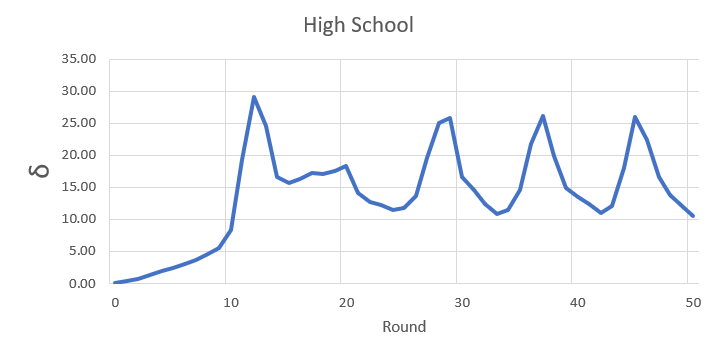}
      \caption{Mean $\delta$ in High School Friendship Network, $\delta$ known, lexicographic tie-breaking}
      \label{fig_hs_epoch_known_lex}
    \end{minipage}%
    \begin{minipage}{.5\linewidth}
      \centering
      \captionsetup{width=.85\linewidth}
      \includegraphics[width=\linewidth]{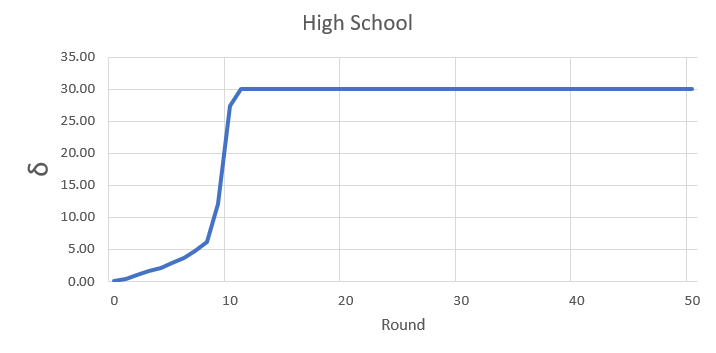}
      \caption{Mean $\delta$ in High School Friendship Network, $\delta$ known, random tie-breaking}
      \label{fig_hs_known_rand}
    \end{minipage}
\end{figure}

As in the Facebook-ego network and the Residence Hall network (Appendix \ref{subsec_residence_hall}), the number of vertices makes it impractical to view the set of all individual values of $\delta$ simultaneously, so in Figures \ref{fig_hs_epoch_unknown}-\ref{fig_hs_known_rand} we instead present the mean $\delta$ values. Figures \ref{fig_hs_epoch_unknown} and \ref{fig_hs_prob_unknown} echo the pattern of a sharp s-curve acceleration to a steady, high $\delta$ value seen in Figure \ref{fig_facebook_unknown}, while Figures \ref{fig_hs_epoch_known_lex} reinforces the repeating s-curve pattern that occurs in lexicographic tie-breaking, and Figure \ref{fig_hs_known_rand} exhibits the same pattern which will be seen in Figure \ref{fig_rh_known_rand}. 

The average utility per round for the settings displayed in Figures \ref{fig_hs_epoch_unknown}-\ref{fig_hs_known_rand} is 27.992, 28.194, 27.173, and 27.330 which correspond to percentage increases of 22.044\%, 22.925\%, 18.473\%, and 19.158\% over the average utility of 22.936 per round produced by agents engaging in selfish behavior.

\subsection{Residence Hall}\label{subsec_residence_hall}

The final social network is a friendship network among 217 residents of a residence hall located on the Australian National University campus from \cite{freeman1998}. Residents in 1987 were asked whether other residents were friends, and if they were then an edge was added to the network. As with the High School friendship network in Appendix \ref{subsec_highschool} edges indicating friendship are directed and positively weighted, with weight here corresponding to strength of friendship, and for the same reasons as in that case we use an unweighted, undirected version of the network. The resulting undirected network has 217 vertices and 1839 edges.

\begin{figure}[ht]
    \centering
    \begin{minipage}{.5\linewidth}
      \centering
      \captionsetup{width=.85\linewidth}
      \includegraphics[width=\linewidth]{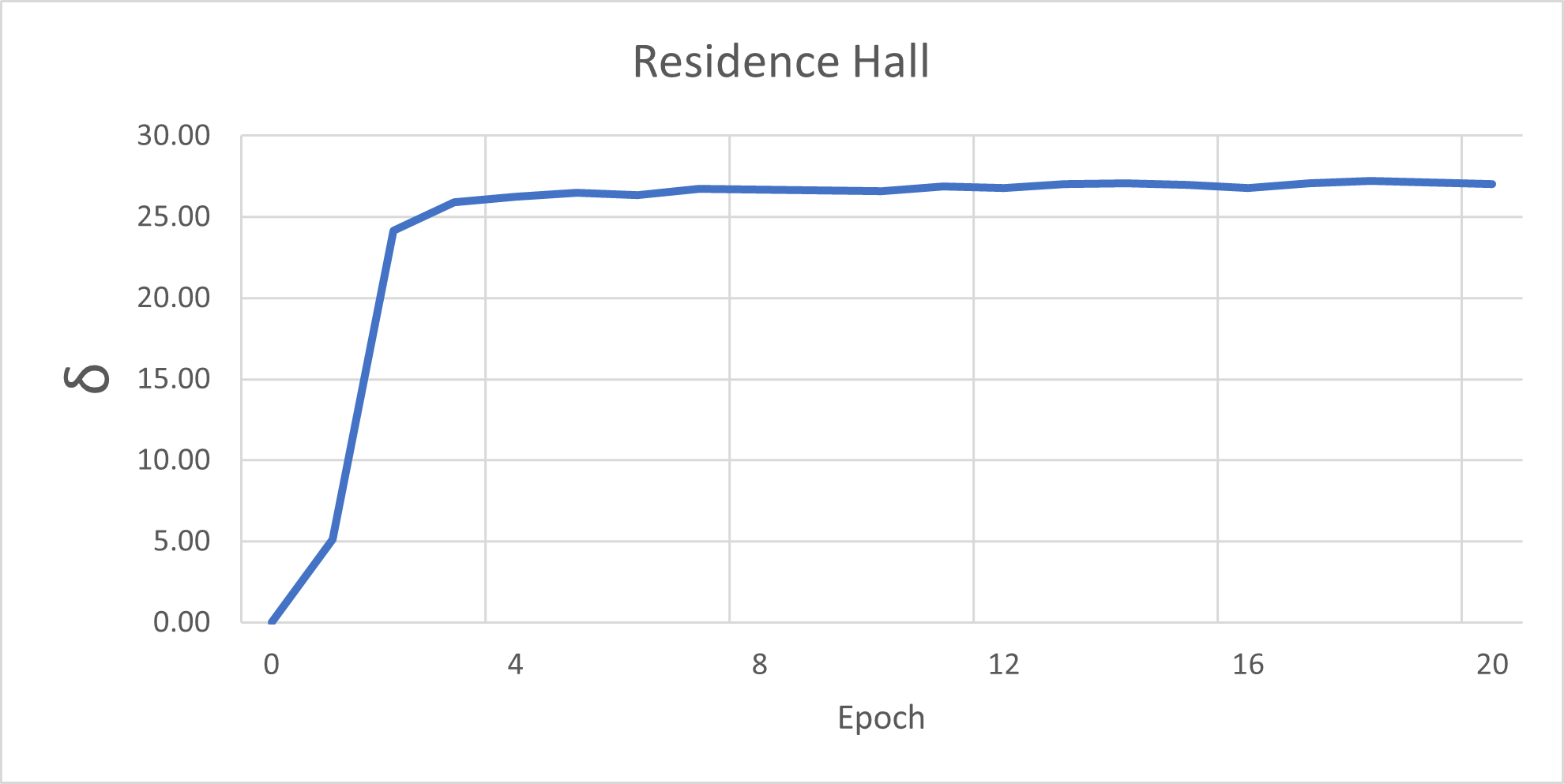}
      \caption{Mean $\delta$ in Residence Hall Friendship Network, $\delta$ unknown, epoch = 100}
      \label{fig_rh_epoch_unknown}
    \end{minipage}%
    \begin{minipage}{.5\linewidth}
      \centering
      \captionsetup{width=.85\linewidth}
      \includegraphics[width=\linewidth]{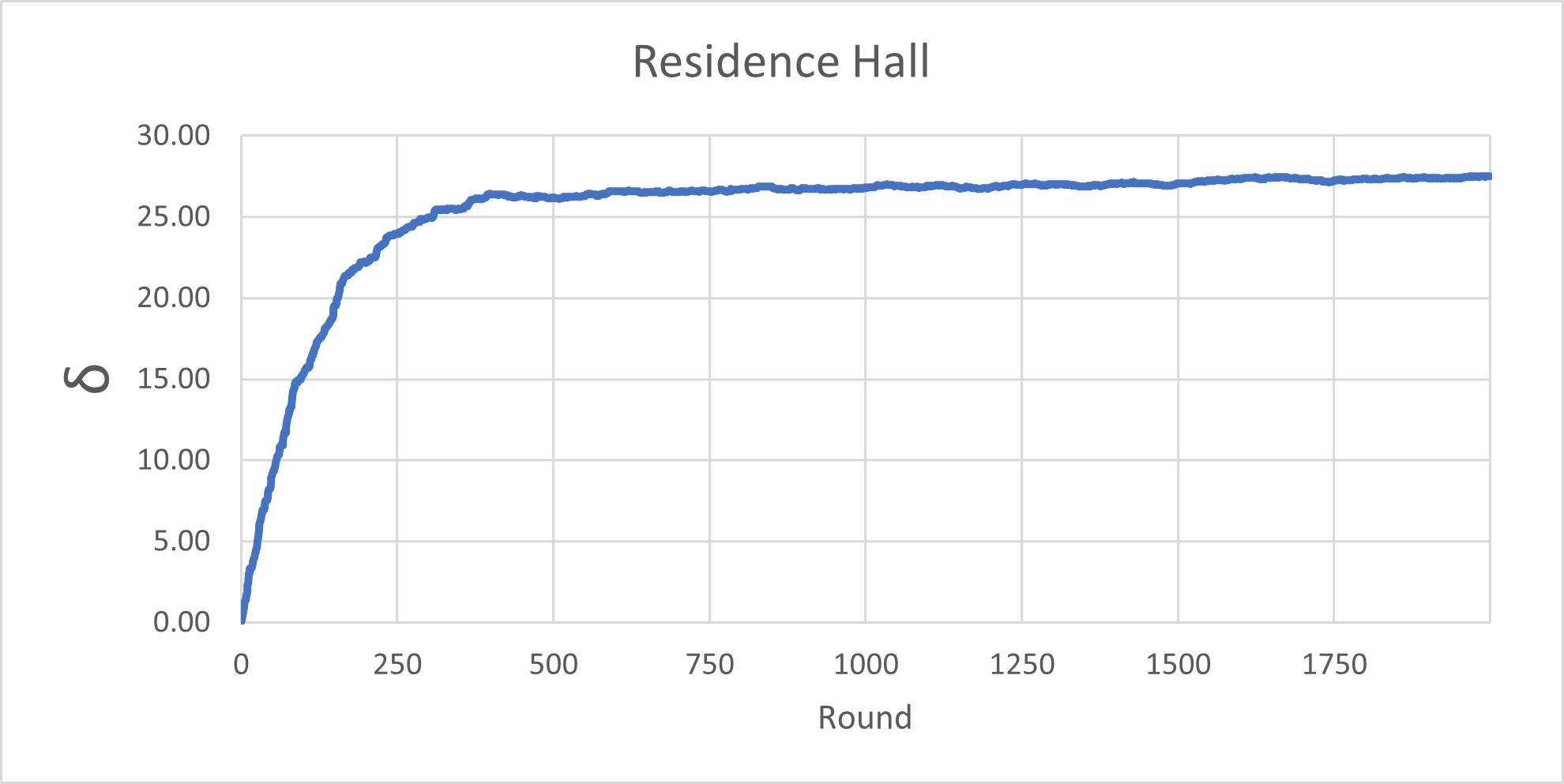}
      \caption{Mean $\delta$ in Residence Hall Friendship Network, $\delta$ unknown, update probability = $\frac{1}{100}$}
      \label{fig_rh_prob_unknown}
    \end{minipage}
\end{figure}

\begin{figure}[ht]
    \centering
    \begin{minipage}{.5\linewidth}
      \centering
      \captionsetup{width=.85\linewidth}
      \includegraphics[width=\linewidth]{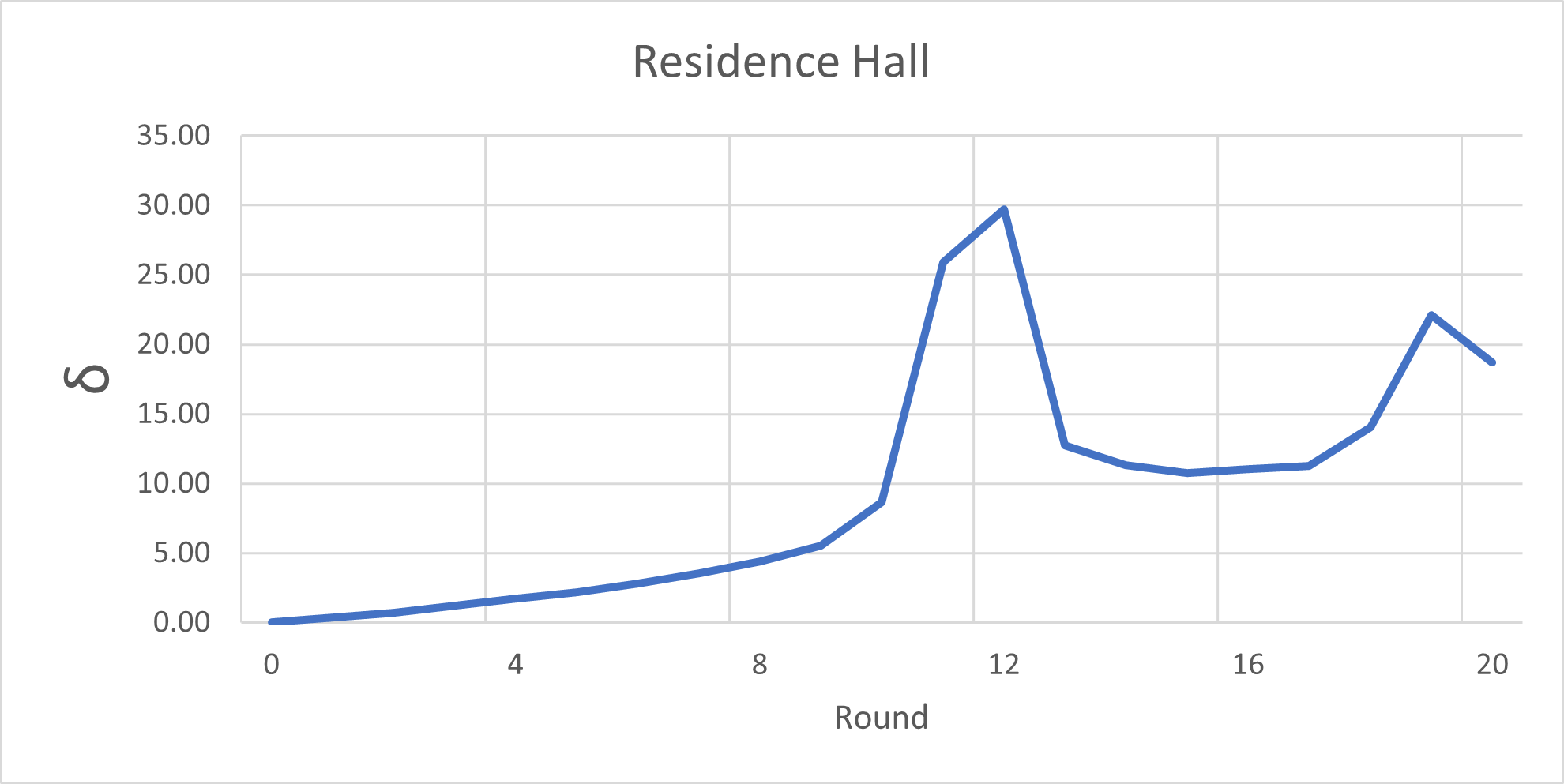}
      \caption{Mean $\delta$ in Residence Hall Friendship Network, $\delta$ known, lexicographic tie-breaking}
      \label{fig_rh_epoch_known_lex}
    \end{minipage}%
    \begin{minipage}{.5\linewidth}
      \centering
      \captionsetup{width=.85\linewidth}
      \includegraphics[width=\linewidth]{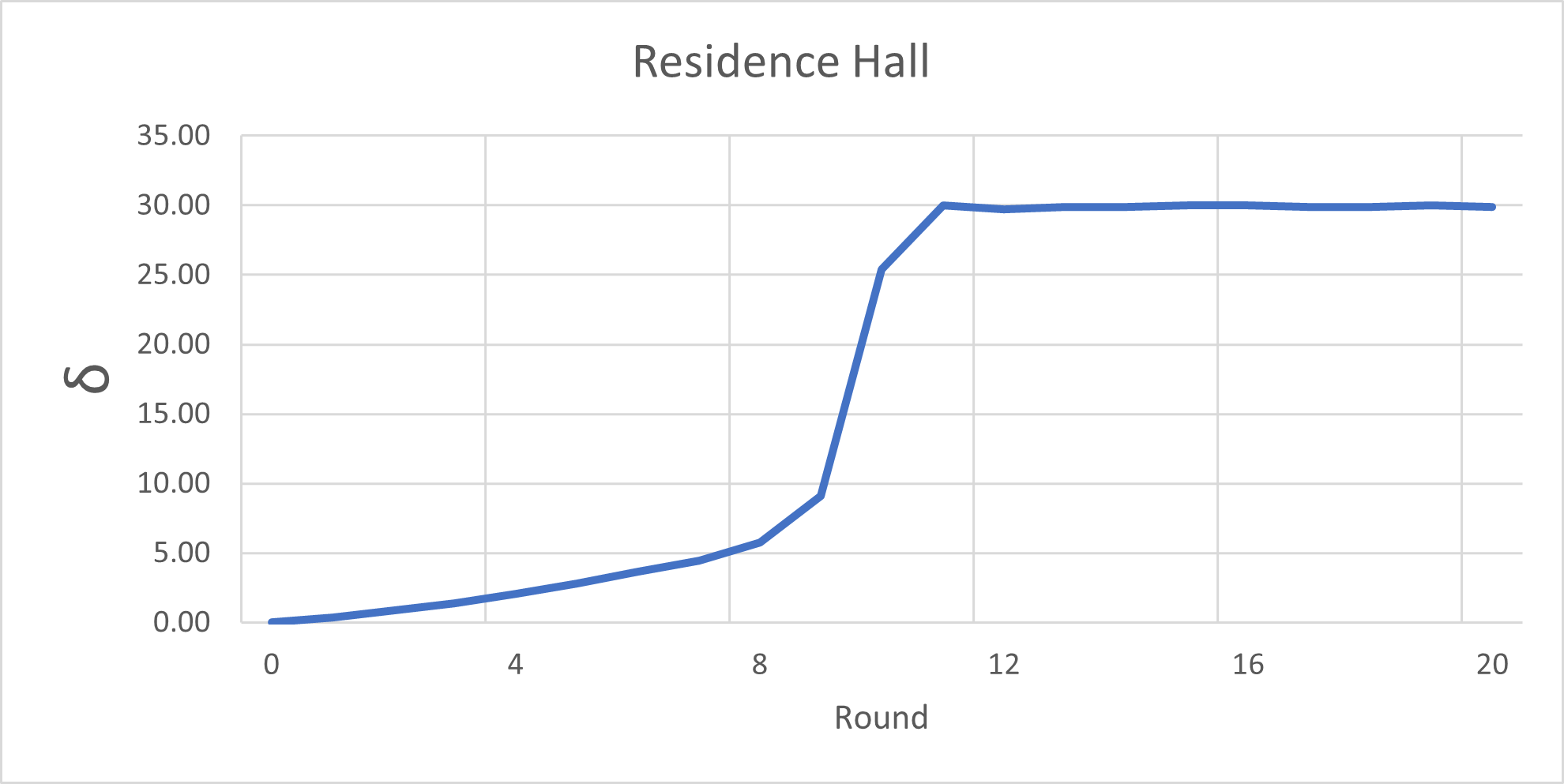}
      \caption{Mean $\delta$ in Residence Hall Friendship Network, $\delta$ known, random tie-breaking}
      \label{fig_rh_known_rand}
    \end{minipage}
\end{figure}

Given the number of individuals in this network, we present the mean value of $\delta$ rather than the individual values. Also due to computational practicalities, we reduce the number of epochs for the trials in Figures \ref{fig_rh_epoch_unknown} and \ref{fig_rh_epoch_known_lex} from 100 to 20, and the number of rounds in Figures \ref{fig_rh_prob_unknown} and \ref{fig_rh_known_rand} to 2000. However, we can already see the same patterns of behavior emerging, with Figures \ref{fig_rh_epoch_unknown} and \ref{fig_rh_prob_unknown} displaying the sharp s-curve followed by a steady high $\delta$ value characteristic of the cases where $\delta$ is unknown. Similarly, Figure \ref{fig_rh_epoch_known_lex} displays the repeated s-curves associated with lexicographic tie-breaking, while Figure \ref{fig_rh_known_rand} displays the same behavior seen in Figures \ref{fig_facebook_known} and  \ref{fig_hs_known_rand}.

The average utility per round for the settings displayed in Figures \ref{fig_rh_epoch_unknown}-\ref{fig_rh_known_rand} is 28.0267, 28.221, 26.880, and 27.350 which correspond to percentage increases of 25.742\%, 26.614\%, 20.598\%, and 22.706\% over the average utility of 22.289 per round produced by agents engaging in selfish behavior.

% \bibliographystyle{IEEEtran}
% \bibliography{references}

\fi
\end{document}